\documentclass[fontsize=10pt,a4paper]{scrartcl}
\usepackage[a4paper, top=4cm, bottom=4cm]{geometry}
\usepackage[utf8x]{inputenc}
\usepackage[T1]{fontenc}
\usepackage{tikz}
\usetikzlibrary{matrix,arrows}
\usepackage{enumitem}
\usepackage{mwe}
\usepackage{hyperref}
\usepackage{amsmath}
\usepackage{amsthm}
\usepackage{amssymb}
\usepackage{amsfonts}
\usepackage{mathrsfs}
\usepackage{bbm}
\usepackage{abstract}
\usepackage{graphicx}
\usepackage{placeins}
\usepackage{subfigure}
\usepackage{nicefrac}
\usepackage{booktabs}
\usepackage{bookmark}
\usepackage{multirow}
\usepackage{dsfont}
\usepackage{longtable}
\usepackage{braket}
\usepackage{tikz}
\usepackage{scalerel}
\usepackage{listings}
\usepackage{color}
\usepackage[all]{xy}
\usepackage{tensor}
\usepackage{faktor}
\usepackage{lipsum}
\usepackage{comment}
\usepackage{cite}

\theoremstyle{definition}
\newtheorem{definition}{Definition}[section] 
 
\theoremstyle{definition}
\newtheorem{remark}[definition]{Remark}

\theoremstyle{definition}
\newtheorem{example}[definition]{Example}

\theoremstyle{theorem}
\newtheorem{theorem}[definition]{Theorem}
\theoremstyle{prop}
\newtheorem{prop}[definition]{Proposition}
\theoremstyle{cor}

\theoremstyle{lemma}
\newtheorem{lemma}[definition]{Lemma}

\newcommand{\Exterior}{\mathchoice{{\textstyle\bigwedge}}%
    {{\bigwedge}}%
    {{\textstyle\wedge}}%
    {{\scriptstyle\wedge}}}

\begin{document}
\title{\LARGE Super Cartan geometry and the super Ashtekar connection}
\author{
  \large Konstantin Eder\thanks{konstantin.eder@gravity.fau.de} \\
	\large Institute for Quantum Gravity (IQG)\\
		\large Friedrich-Alexander-Universität Erlangen-Nürnberg (FAU)
}
\maketitle

\begin{abstract}
This work is devoted to the geometric approach to supergravity. More precisely, we interpret $\mathcal{N}=1$, $D=4$ supergravity as a super Cartan geometry which provides a link between supergravity and Yang-Mills theory.\\
To this end, we first review important aspects of the theory of supermanifolds and we establish a link between various different approaches. We then introduce super Cartan geometries using the concept of so-called enriched categories. This, among other things, will enable us to implement anticommutative fermionic fields. We will then also show that non-extended $D=4$ supergravity naturally arises in this framework.\\ 
Finally, using this gauge-theoretic interpretation as well as the chiral structure of the underlying supersymmetry algebra, we will derive graded analoga of Ashtekar's self-dual variables and interpret them in terms of generalized super Cartan connections. This gives canonical chiral supergravity the structure of a Yang-Mills theory with gauge supergroup similar to the self-dual variables in ordinary first-order Einstein gravity which was first observed by F\"ul\"op \cite{Fulop:1993wi}.\\
We then construct the parallel transport map corresponding to the super connection in mathematical rigorous way using again enriched categories. This provides the possibility of quantizing gravity and matter degrees of freedom in loop quantum gravity in a unified way.
\end{abstract}

\newpage

\section{Introduction}
Soon after the first discovery of supergravity by Freedman, Ferrara and van Nieuwenhuizen \cite{Freedman:1976xh} in 1976, Ne'eman and Regge studied a new geometric approach to supergravity based on the ideas of Cartan of a purely geometric interpretation of gravity\footnote{here, in the context of Cartan the term 'geometric' means the description of gravity in terms of (a particular kind of) a connection on a principal bundle}. In this theory, now commonly known as Cartan geometry, gravity arises by considering the underlying symmetry groups of flat Minkowski spacetime, i.e., a Klein geometry consisting of the isometry group given by the Poincaré group and the Lorentz group as stabilizer subgroup of a particular spacetime event. Gravity is then obtained by deforming this flat initial data in a particular way by studying a certain kind of connection forms, called Cartan connections, taking values in the Lie algebra of the isometry group of the flat model. This Cartan geometric approach gives gravity a very clear geometric interpretation and even allows the inclusion of matter fields via Kaluza-Klein reduction of higher dimensional pure gravity theories leading for instance to Einstein-Yang-Mills theories. However, it still has some limitations, as, for instance, it does not include fermionic fields. This changes in case of supersymmetry, as graded Lie algebras, by definition, naturally include fermionic generators. It was then realized extending Cartan geometry to the super category that this in fact leads to supergravity. The fermion field, given by the superpartner of the graviton field, then arises from the odd components of a super Cartan connection taking values in the graded extension of the Poincaré algebra. Besides, this description also yields a geometric interpretation of supersymmetry transformations in terms of local gauge transformations along the odd directions of the underlying (graded) structure group.

These ideas were studied more systematically and developed even further by D'Auria and Fré et al. \cite{DAuria:1982uck,Castellani:1991et} to include extended and higher dimensional supergravity theories. Moreover, generalizing the Maurer-Cartan equations to include higher $p$-form gauge fields which naturally appear in higher dimensions, such as the supergravity $C$-field in the unique maximal $D=11$, $\mathcal{N}=1$ supergravity theory, then lead to the concept of free graded differential algebras (FDA). These type of algebras then turned out to have a rigorous geometric interpretation in higher category theory describing the higher gauge fields as components of a higher Cartan connection \cite{Fiorenza:2013nha,Sati:2015yda}.

In the canonical approach to (quantum) supergravity such as, most prominently, in the framework of loop quantum gravity (LQG) (see e.g. \cite{Ashtekar:2004eh,Thiemann:2007pyv} and references therein), this naturally raises the question whether at least some remnant of the underlying geometrical structure of the full theory still can be seen in the resulting canonical theory. In fact, it was observed by F\"ul\"op in the seminal paper \cite{Fulop:1993wi} while studying a specific subfamily of the constraints of canonical $\mathcal{N}=1$, $D=4$ anti-de Sitter supergravity that the corresponding constraint algebra has the structure of a graded Lie algebra leading to some kind of a graded generalization of Ashtekar's self-dual variables. Using this graded connection, Gambini et al. \cite{Gambini:1995db} as well as Ling and Smolin \cite{Ling:1999gn} then studied the possibility of generalizing the notion of spin network states to super spin networks.

In this paper, we will study this idea more systematically using the strong link between supergravity and Cartan geometry. We will therefore provide a mathematical rigorous account to the formulation of a super Cartan geometry. However, the problem of modeling anticommuting classical fermion fields, which is crucial in the context of supersymmetry, turns out be by far non-straightforward. This seems to be usually ignored in the physical literature. Nevertheless, motivated from algebraic geometry, this problem has a beautiful resolution using the concept of enriched categories first studied already by Schmitt \cite{Schmitt:1996hp}. We will apply these ideas to our situation and then see that the resulting picture resembles very closely the description of fermion fields in mathematical rigorous approaches to quantum field theory involving anticommuting classical fermion fields such as perturbative algebraic quantum field theory (pAQFT) \cite{Rejzner:2011au,Rejzner:2016hdj}.
We will then see that non-extended $D=4$ supergravity arises naturally in this framework. For an interesting approach which is different from the present one, using the notion of '\emph{integral forms}' see \cite{Castellani:2014goa,Cremonini:2019aao,Catenacci:2018xsv}.

Based on this geometric formuation, we will be able to interpret the graded analog of Asthekar's self dual variables in terms of generalized super Cartan connections and give a conceptual explanation for the observation of F\"ul\"op. As it turns out, this connection appears very naturally when studying the chiral structure of the underlying supersymmetry algebra and is rooted on the special property of the (bosonic) self-dual variables. Moreover, as we will see, this particular structure of the supersymmetry algebra even survives in case of extended supersymmetry which shows that the existence of the super Ashtekar connection is not just a mere coincidence.

Another advantage of deriving the super Ashtekar connection using the strong relation between supergravity and super Cartan connection is that this naturally leads to an interpretation of the canonical theory as a Yang-Mills theory with supergroup as a gauge group. This is in fact in complete analogy to the classical theory since the self-dual variables give first-order general relativity the structure of a $\mathrm{SL}(2,\mathbb{C})$ Yang-Mills theory. We then use this connection in order to construct the corresponding parallel transport map using again the concept of enriched categories. These results are of course of independent interest and may even have applications in other areas of physics that involve super gauge fields.

As a final step, we will then use these holonomies in order to construct the state space of loop quantum supergravity based on the ideas of \cite{Gambini:1995db,Ling:1999gn}. Therefore, we will also derive the Haar measure of the underlying gauge supergroup. As it turns out, this gives the resulting Hilbert space a very intriguing structure which also uncoveres some of the mathematical beauty underlying the standard quantization of fermion fields in LQG as proposed in \cite{MoralesTecotl:1994ns,Thiemann:1997rq}.\\
\\
The structure of this paper is follows: At the beginning, we will provide a detailed account of the concept of supermanifolds and establish a link between the various different approaches to this subject. We will then recall some basic super Lie groups which play a fundamental role in supersymmetry and supergravity. At the end of this section, we summarize the construction of the super holonomies as considered in detail in \cite{Konsti-FB:2020} and which is based on the concept of enriched categories in order to describe anticommuting classical fermion fields.

In section \ref{Cartan geometry}, we describe the geometrical interpretation of gravity in terms of a Cartan geometry. We will then define super Cartan geometries in the framework of enriched categories in section\ref{superCartan geometry}. Finally, this formulation will be used in order to derive $\mathcal{N}=1$, $D=4$ (anti-de Sitter) supergravity via the super MacDowell-Mansouri action.

With all these ingredients, the graded Ashtekar will be derived in section \ref{Ashtekar} studying the chiral structure of the underlying (extended) super Poincaré/anti-de Sitter group. Using the relation to Yang-Mills theory, we will then use the super holonomies to construct the Hilbert space in the manifest approach to loop quantum supergravity. In particular, we will derive the invariant Haar measure of the gauge supergroup and describe the link to the standard quantization scheme of fermion fields in LQG. Finally, the generalization of the concept of invariant connections to the super category will be considered. This provides a mathematical solid basis to apply these results to symmetry reduced models in the framework of supersymmetric LQC which will be studied in more detail in \cite{Konsti-Kos:2020}. 

\section{Supersymmetry and Supergeometry}

\subsection{Three roads towards a theory of supermanifolds}\label{Supermanifolds}
In the literature, there exist various different approaches to formulate the notion of a supermanifold. Probably, the most popular one is the so-called \emph{algebro-geometric approach} introduced by Berezin, Kostant and Leites \cite{Berezin:1976,Kostant:1975qe}. As can be already inferred from its name, this approach borrows techniques from  algebraic geometry and starts formulating supermanifolds based on the observation that ordinary smooth manifolds can be equivalently be described in terms of the function sheaf defined on it. Albeit being very elegant, its definition is very abstract making it less accessible for physicists for concrete applications. Roughly speaking, this is due to the lack of points as points in this framework are implicitly encoded in the underlying structure sheaf.\\
Hence, another approach to supermanifolds, the so-called \emph{concrete approach}, was initiated by DeWitt \cite{DeWitt:1984} and Rogers \cite{Rogers:1980} defining them similar to ordinary smooth manifolds in terms of a topological space of points, i.e., a topological manifold that locally looks a flat superspace. However, as it turns out, this definition has various ambiguities in formulating the notion of a point yielding too many unphysical degrees of freedom.\\
It was then found by Molotkov \cite{Mol:10} and further developed by Sachse \cite{Sac:08} that both approaches are in fact two sides of the same coin. More precisely, as will be explained in more detail in what follows, it was shown that Rogers-DeWitt supermanfolds can be interpreted in terms of a functor constructed out of a algebro-geometric supermanifold. This functorial intepretation then resolved the ambiguities arising in the Rogers-De Witt approach and also opens the way for a generalization of the theory to infinite dimensional supermanifolds.\\
In the following, we want to introduce the Berezin-Kostant-Leites approach to supermanifold theory and explain in some detail its relation to the Rogers-DeWitt approach using the functor of points which naturally leads to their functorial interpretation as observed by Molotkov and Sachse.\\
\\
The Berezin-Kostant-Leites approach is based on the observation that ordinary smooth manifolds can equivalently be described in terms of locally ringed spaces. Therefore note that any smooth manifold canonically yields the locally ringed space $(M,C^{\infty}_M)$ which is locally isomorphic to some $(V,C^{\infty}_{\mathbb{R}^n}|_V)$ with $V\subseteq\mathbb{R}^n$ open.\\
In fact, it turns that all smooth manifold $M$ can be described this way. That is, if $(M,\mathcal{O}_M)$ is a locally ringed space with $\mathcal{O}_M$ a sheaf on $M$ such that $(M,\mathcal{O}_M)$ is locally isomorphic to some $(V,C^{\infty}_{\mathbb{R}^n})$ with $V\subseteq\mathbb{R}^n$ open. Then, $M$ can be given the structure of smooth manifold in a unique way such that $\mathcal{O}_M\cong C^{\infty}_M$. Even more, it follows that both categories are equivalent.\\
\\
Based on this idea, one defines supermanifolds as some sort of locally super ringed spaces generalzing appropriately the notion of a smooth function.  Therefore, a so-called \emph{supersmooth fucntion} or \emph{superfield} $f$ on the superspace $\mathbb{R}^{m|n}=\mathbb{R}^m\oplus\mathbb{R}^n$ is defined as a function of the form.
\begin{equation}
f=\sum_{I\in M_n}{f_I\theta^I}
\label{eq:2.2.1}
\end{equation} 
with $f_I$ ordinary smooth functions on $\mathbb{R}^m$ for any multi-index $I=(i_1,\ldots,i_k)\in M_n$, $0\leq|I|=k\leq n$, where $\theta^I:=\theta^{i_1}\ldots\theta^{i_k}$. In the following, we will follow very closely \cite{Carmeli:2011} for the definition of algebro-geometric supermanifolds and the construction of the functor of points (see also appendix \ref{SuperAlg} for our choice of conventions in super linear algebra). Therefore, we will omit most of the proofs. 
\begin{definition}
An \emph{algebro geometric supermanifold} of dimension $(m,n)$ is a locally super ringed space $\mathcal{M}=(M,\mathcal{O}_{\mathcal{M}})$ that is locally isomorphic to the superspace $\mathbb{R}^{m|n}$. More precisely, $(M,\mathcal{O}_M)$ consists of a topological space $M$ which is Hausdorff and second countable as well as a sheaf $\mathcal{O}_{\mathcal{M}}$ over $M$ of super commutative rings called \emph{structure sheaf} such that, for any $x\in M$, the stalk $\mathcal{O}_{M,x}$ is a local super ring. Moreover, for $x\in M$, there exists an open neighborhood $U\subset M$ of $x$ as well as an isomorphism $\phi_U=(|\phi_U|,\phi_U^{\sharp})$ of ordinary locally ringed spaces\footnote{A morphism $f=(|f|,f^{\sharp}):\,(X,\mathcal{O}_X)\rightarrow(Y,\mathcal{O}_Y)$ of locally ringed spaces consists of a continuous map $|f|:\,X\rightarrow Y$ between topological spaces as well as a morphism $f^{\sharp}:\,\mathcal{O}_Y\rightarrow f_{*}\mathcal{O}_X$ of sheaves of rings on $Y$ called \emph{pullback} such that, $\forall x\in X$, the induced morphism $f^{\sharp}_x:\,\mathcal{O}_{Y,f(x)}\rightarrow\mathcal{O}_{X,x}$ is local, i.e., $f^{\sharp}_x$ maps the maximal ideal of the stalk $\mathcal{O}_{Y,f(x)}$ to the maximal ideal of $\mathcal{O}_{X,x}$.} 
\begin{equation}
\phi_U=(|\phi_U|,\phi_U^{\sharp}):\,(U,\mathcal{O}_{\mathcal{M}}|_U)\rightarrow(|\phi_U|(U),C^{\infty}_{\mathbb{R}^m}|_{|\phi_U|(U)}\otimes\bigwedge{[\theta^1,\ldots,\theta^n]})
\label{eq:2.2.2}
\end{equation}  
such that $\phi_U^{\sharp}:\,C^{\infty}_{\mathbb{R}^m}|_{|\phi_U|(U)}\otimes\bigwedge{[\theta^1,\ldots,\theta^n]}\rightarrow\mathcal{O}_{\mathcal{M}}|_U$, in addition, is a (even) morphism of sheaves of super algebras. The tuple $(U,\phi_U)$ is called a local \emph{chart} or \emph{superdomain} of $x$. A family $\{(U_{\alpha},\phi_{\alpha})\}_{\alpha\in\Lambda}$ of charts is called an \emph{atlas} of $(M,\mathcal{O}_{\mathcal{M}})$ if $\bigcup_{\alpha\in\Lambda}{U_{\alpha}}=M$. \\
A morphism of $f=(|f|,f^{\sharp}):\,(M,\mathcal{O}_{\mathcal{M}})\rightarrow(N,\mathcal{O}_{\mathcal{N}})$ of algebro geometric supermanifolds is a morphism of the underlying ordinary locally ringed spaces such that $f^{\sharp}:\,\mathcal{O}_{\mathcal{N}}\rightarrow f_{*}\mathcal{O}_{\mathcal{M}}$ also is an (even) morphism of super algebras. Algebro geometric supermanifolds together with morphisms between them form a category $\mathbf{SMan}_{\mathrm{Alg}}$ called the \emph{category of algebro-geometric supermanifolds}.
\end{definition}
\begin{remark}
Choosing a chart $(U,\phi_U)$ of an algebro-geometric supermanifold $(M,\mathcal{O}_{\mathcal{M}})$, this induces a local coordinates $(t^{\sharp i},\theta^{\sharp j})$ on $\mathcal{M}|_U:=(U,\mathcal{O}_{\mathcal{M}}|_U)$ via $t^{\sharp i}:=\phi_U^{\sharp}(t^i)$ and $\theta^{\sharp j}:=\phi_U^{\sharp}(\theta^j)$ $\forall i=1,\ldots,m$, $j=1,\ldots,m$ where $\mathrm{dim}(M,\mathcal{O}_{\mathcal{M}})=(m,n)$. Moreover, any function $f\in\mathcal{O}_{\mathcal{M}}|_U$ is of the form
\begin{equation}
f=\sum_{I\in M_n}{f_I\theta^{\sharp I}}
\label{eq:}
\end{equation}
where, for $I=(i_1,\ldots,i_k)\in M_n$, $\theta^{\sharp I}:=\theta^{\sharp i_1}\cdots\theta^{\sharp i_n}$ and $f_I=\phi_U^{\sharp}(g_I)$ for some smooth function $g_I\in C^{\infty}(|\phi_U|(U))$.
\end{remark}
Any supermanifold naturally contains an ordinary smooth manifold as a submanifold. Therefore, for any algebro-geometric supermanifold $\mathcal{M}=(M,\mathcal{O}_{\mathcal{M}})$ and $U\subset M$ open, consider the set $\mathcal{J}_{\mathcal{M}}(U):=\{f\in\mathcal{O}_{\mathcal{M}}(U)|\,f\text{ is nilpotent}\}$. Then, it follows that $\mathcal{J}_{\mathcal{M}}(U)$ is an ideal in $\mathcal{O}_{\mathcal{M}}(U)$ yielding another sheaf $U\mapsto\mathcal{J}_{\mathcal{M}}(U)$. Hence, one can construct the quotient sheaf $\mathcal{O}_{\mathcal{M}}/\mathcal{J}_{\mathcal{M}}$ whose sections locally have the structure of an ordinary smooth functions. Hence, this yields a locally ringed space 
\begin{equation}
M_0:=(M,\mathcal{O}_{\mathcal{M}}/\mathcal{J}_{\mathcal{M}})
\label{eq:2.2.3}
\end{equation}
which is a submanifold and has the structure of an ordinary smooth manifold.\\
Before we continue, let us mention a central result in the theory of algebro-geometric supermanifolds as it will appear quite frequently in the discussion in what follows. It states that morphisms are uniquely characterized via the pullback of a basis of global sections.  
\begin{theorem}[Global Chart Thorem \cite{Carmeli:2011}]
Let $\mathcal{M}$ be an algebro-geometric supermanifold and $\mathcal{U}^{m|n}=(U,C^{\infty}_U)\subseteq\mathbb{R}^{m|n}$ be a superdomain with $U\subseteq\mathbb{R}^{m}$ open. There is a bijectice correspondence between supermanifold morphisms $\psi:\,\mathcal{M}\rightarrow\mathcal{U}^{m|n}$ and tuples $(t^{\sharp i},\theta^{\sharp j})$ of global sections of $\mathcal{O}(\mathcal{M})$ with $t^{\sharp i}$ even and $\theta^{\sharp j}$ odd, $i=1,\ldots,m$ and $j=1,\ldots,n$, such that $(t^{\sharp i}(x),\theta^{\sharp j}(x))\in U$ $\forall x\in M$. 
\end{theorem}
\begin{proof}
For a section $f\in\mathcal{O}(\mathcal{M})$, the \emph{value $f(x)\equiv\mathrm{ev}_x(f)$} of $f$ at $x\in M$ is defined as the unique real number such that $f-f(x)$ is not invertible in any open neighborhood of $x$ in $M$. It follows that, if $\phi:\,\mathcal{M}\rightarrow\mathcal{N}$ is a morphism of supermanifolds, then $\phi^{\sharp}(g)(y)=g(|\phi|(y))$ for any $g\in\mathcal{O}(\mathcal{N})$ and $y\in N$.\\
Hence, it is clear, by restricting the global sections $t^i$ and $\theta^j$ of $\mathbb{R}^{m|n}$ to the superdomain $\mathcal{U}^{m|n}$, that their respective pullback $t^{\sharp i}:=\psi^{\sharp}(t^i)$ and $\theta^{\sharp j}:=\psi^{\sharp}(\theta^j)$ w.r.t. a morphism $\psi:\,\mathcal{M}\rightarrow\mathcal{U}^{m|n}$ indeed satisfy the properties as stated in the theorem. The inverse direction follows from the local triviality property of supermanifolds.
\end{proof}
Typical examples of a supermanifolds are obtained via their strong relationship to vector bundles. Let $V\rightarrow E\stackrel{\pi}{\rightarrow}M$ be a real vector bundle over an $m$-dimensional manifold with typical fiber given by a vector space $V$ of dimenional $n$. This naturally yields a locally ringed space defining 
\begin{equation}
\mathbf{S}(E,M):=(M,\Gamma(\Exterior E^*))
\label{eq:2.2.5}
\end{equation}
where $\Gamma(\Exterior E^*)$ denotes the space of of smooth sections of the exterior bundle $\Exterior E^*$. Since $\Gamma(\Exterior E^*)\cong\Exterior\Gamma(E)^*$ naturally carries a $\mathbb{Z}_2$-grading, it follows that it has the structure of a sheaf of local super rings, that is, $\mathbf{S}(E,M)$ defines an algebro-geometric supermanifold of dimension $(m,n)$ also called a \emph{split supermanifold}. A morphism $(\phi,f):\,(E,M)\rightarrow (F,N)$ between two vector bundles induces a morphism $\mathbf{S}(\phi,f):\,\mathbf{S}(E,M)\rightarrow\mathbf{S}(F,N)$ between the corresponding split supermanifolds. Hence, this yields a functor
\begin{equation}
\mathbf{S}:\,\mathbf{Vect}_{\mathbb{R}}\rightarrow\mathbf{SMan}_{\mathrm{Alg}}
\label{eq:2.2.6}
\end{equation}
between the category of real vector bundles to the category of algebro-geometric supermanifolds which we call the \emph{split functor}. It is a general result due to Batchelor \cite{Batchelor:1979} that any algebro-geometric supermanfold is isomorphic to a split supermanifold of the form (\ref{eq:2.2.5}), i.e., (\ref{eq:2.2.6}) is a surjective on objects. However, the split functor is \emph{not} full, i.e., not  every morphism $f:\,\mathbf{S}(E,M)\rightarrow\mathbf{S}(F,N)$ between split manifolds arises from a morphism between the respective vector bundles $(E,M),(F,N)\in\mathbf{Ob}(\mathbf{Vect}_{\mathbb{R}})$. Hence, the structure of morphisms between supermanifolds in general turns out to be much richer than for ordinary vector bundles. This is of utmost importance in modelling for instance supersymmetry transformations.\\
\\
As a next step, we want to find a relation between algebro-geometric and Rogers-DeWitt supermanifolds. A very elegant way in describing this relationship is given by the so-called \emph{functors of point} approach. It is a general technique in algebraic geometry which can be used in order to give, a priori, very abstract objects a more concrete reinterpretation making proofs much easier in certain instances.\\
As already explained at the beginning, a general 'problem' concerning algebro-geometric supermanifolds is the lack of points. In fact, in contrast to ordinary smooth manifolds, the points of the underlying topological space do not suffice to uniquely characterize the sections of the structure sheaf. As it turns out, this can be cured by studying the morphisms between them.
\begin{definition}\label{def:2.3}
Let $\mathcal{M}$ be an algebro-geometric supermanifold. The \emph{functor of points of $\mathcal{M}$} is defined as the covariant functor $\mathcal{M}:\,\mathbf{SMan}_{\mathrm{Alg}}^{\mathrm{op}}\rightarrow\mathbf{Set}$ on the opposite category $\mathbf{SMan}_{\mathrm{Alg}}^{\mathrm{op}}$ associated to $\mathcal{M}$\footnote{The \emph{opposite category} $\mathcal{C}^{\mathrm{op}}$ of a category $\mathcal{C}$ consists of the same class of objects as $\mathcal{C}$, i.e., $\mathbf{Ob}(\mathcal{C}^{\mathrm{op}})=\mathbf{Ob}(\mathcal{C})$, and for each $X,Y\in\mathbf{Ob}(\mathcal{C}^{\mathrm{op}})$ the set $\mathrm{Hom}_{\mathcal{C}^{\mathrm{op}}}(X,Y)$ of morphisms between $X$ and $Y$ is given by $\mathrm{Hom}_{\mathcal{C}}(Y,X)$} which on objects $\mathcal{T}\in\mathbf{Ob}(\mathbf{SMan}_{\mathrm{Alg}})$ is defined as 
\begin{equation}
\mathcal{M}(\mathcal{T}):=\mathrm{Hom}(\mathcal{T},\mathcal{M})
\label{eq:2.2.7}
\end{equation}
also called the \emph{$\mathcal{T}$-point of $\mathcal{M}$} and on morphisms $f\in\mathrm{Hom}(\mathcal{T},\mathcal{S})$, $\mathcal{M}(f)\in\mathrm{Hom}(\mathcal{M}(\mathcal{S}),\mathcal{M}(\mathcal{T}))$ is given by
\begin{equation}
\mathcal{M}(f):\,\mathcal{M}(\mathcal{S})\rightarrow \mathcal{M}(\mathcal{T}),\,g\mapsto g\circ f
\label{eq:2.2.8}
\end{equation}
Hence, the functor of points of $\mathcal{M}$ coincides with the partial $\mathrm{Hom}$-functor $h_\mathcal{M}\equiv\mathrm{Hom}(\mathcal{M},\,\cdot\,)$ on $\mathbf{SMan}_{\mathrm{Alg}}^{\mathrm{op}}$.
\end{definition}
If $\phi:\,\mathcal{M}\rightarrow \mathcal{N}$ is a morphism between algebro-geometric supermanifolds, this yields a map $\phi_\mathcal{T}:\,\mathcal{M}(\mathcal{T})\rightarrow \mathcal{N}(\mathcal{T})$ between the associates $\mathcal{T}$-points setting $\phi_\mathcal{T}(f):=\phi\circ f$ $\forall f\in \mathcal{M}(\mathcal{T})=\mathrm{Hom}(\mathcal{T},\mathcal{M})$. By definition, it then follows that for any morphism $f:\,\mathcal{S}\rightarrow \mathcal{T}$ one has
\begin{equation}
\mathcal{N}(f)\circ\phi_\mathcal{T}(g)=\phi_\mathcal{T}(g)\circ f=\phi\circ g\circ f=\phi_\mathcal{S}\circ \mathcal{M}(f)(g)
\label{eq:2.2.9}
\end{equation}
$\forall g\in \mathcal{M}(\mathcal{T})$, that is, the following diagram is commutative 
\begin{displaymath}
\xymatrix{
       \mathcal{M}(\mathcal{T}) \ar[r]^{\mathcal{M}(f)}  \ar[d]^{\phi_\mathcal{T}} &   \mathcal{M}(\mathcal{S})  \ar[d]^{\phi_\mathcal{S}} \\
         \mathcal{N}(\mathcal{T}) \ar[r]^{\mathcal{N}(f)} &     \mathcal{N}(\mathcal{S})    
     }
		\label{eq:2.2.10}
 \end{displaymath}
Hence, a morphism $\phi:\,\mathcal{M}\rightarrow \mathcal{N}$ induces a natural transformation between the associated functor of points. This poses the question whether all natural transformations arise in this way. This is an immediate consequence of the following well-known lemma.
\begin{lemma}[Yoneda Lemma]
Let $\mathcal{C}$ be a category and $F:\,\mathcal{C}\rightarrow\mathbf{Set}$ be a functor. Then, for any object $X\in\mathbf{Ob}(\mathcal{C})$, the assignment $\eta\mapsto\eta_X(\mathrm{id}_X)$ yields a bijective correspondence between natural transformations $\eta:\,\mathrm{Hom}(X,\,\cdot\,)\rightarrow F$ and the set $F(X)\in\mathbf{Ob}(\mathbf{Set})$. 
\end{lemma}
Applied to our concrete situation, this implies that for the functor of points $h_\mathcal{M}:\,\mathbf{SMan}_{\mathrm{Alg}}^{\mathrm{op}}\rightarrow\mathbf{Set}$ and $h_\mathcal{N}:\,\mathbf{SMan}_{\mathrm{Alg}}^{\mathrm{op}}\rightarrow\mathbf{Set}$ associated to algebro-geometric supermanifolds $\mathcal{M}$ and $\mathcal{N}$, one has a bijective correspondence between natural transformations between $h_\mathcal{M}$ and $h_\mathcal{N}$ and elements in $\mathrm{Hom}(\mathcal{N},\,\cdot\,)(\mathcal{M})=\mathrm{Hom}(\mathcal{M},\mathcal{N})$. In particular, the supermanifolds $\mathcal{M}$ and $\mathcal{N}$ are isomorphic iff the associated functor of points are naturally isomorphic.\\
\\
We next want to find an equivalent description of the $\mathcal{T}$-points of an algebro-geometric supermanifold $\mathcal{M}$ purely in terms of global sections of the structure sheaf $\mathcal{O}_{\mathcal{M}}$. Consider therefore the set $\mathrm{Spec}_{\mathbb{R}}(\mathcal{O}(\mathcal{M})):=\mathrm{Hom}_{\mathbf{SAlg}}(\mathcal{O}(\mathcal{M}),\mathbb{R})$ called the \emph{real spectrum of $\mathcal{O}(\mathcal{M}):=\mathcal{O}_{\mathcal{M}}(M)$}. Since, a morphism $\phi:\,\mathcal{O}(\mathcal{M})\rightarrow\mathbb{R}$ in the real spectrum is always surjective, it follows that the kernel $\mathrm{ker}(\phi)$ yields a maximal ideal in $\mathcal{O}(\mathcal{M})$, i.e., an element of the \emph{maximal spectrum}
\begin{equation}
\mathrm{MaxSpec}_{\mathbb{R}}(\mathcal{O}(\mathcal{M})):=\{I\subset\mathcal{O}(\mathcal{M})|\,I\text{ is a maximal ideal}\}
\label{eq:2.2.11}
\end{equation}
In fact, it follows that all maximal ideals in $\mathcal{O}(\mathcal{M})$ are of this form. This is a direct consequence of the super version of the classical 'Milnor's exercise' \cite{Carmeli:2011}.
\begin{prop}[Super Milnor's exercise]\label{Prop:2.4}
For an algebro-geometric supermanifold $\mathcal{M}$ all the maximal ideals in $\mathcal{O}(\mathcal{M})$ are of the form $\mathfrak{J}_x:=\mathrm{ker}(\mathrm{ev}_x:\,\mathcal{O}(\mathcal{M})\rightarrow\mathbb{R})$ for some $x\in M$, where $\mathrm{ev}_x\in\mathrm{Hom}_{\mathbf{SAlg}}(\mathcal{O}(\mathcal{M}),\mathbb{R})$ is the evaluation map at $x$ mapping $f\in\mathcal{O}(\mathcal{M})$ to its real value $\mathrm{ev}_x(f)\equiv f(x)$.
\end{prop}
\begin{proof}
Let $I\subset\mathcal{O}_{\mathcal{M}}$ be a maximal ideal. On $M_0:=(M,\mathcal{O}_{\mathcal{M}}/\mathcal{J}_{\mathcal{M}})$ consider the the subset $j^{\sharp}(I)\subseteq C^{\infty}(M_0)$ of $C^{\infty}(M_0)$, where $j^{\sharp}:\,\mathcal{O}_{\mathcal{M}}\rightarrow\mathcal{O}_{\mathcal{M}}/\mathcal{J}_{\mathcal{M}}\cong C^{\infty}_{M_0}$ is the pullback of the canonical embedding $j:\,M_0\hookrightarrow\mathcal{M}$. Since $j^{\sharp}$ is a surjective morphism of super rings and $1\notin I$, it follows that $j^{\sharp}(I)$ is a maximal ideal in $C^{\infty}(M_0)$. By the classical Milnor's exercise, we thus have $j^{\sharp}(I):=\mathrm{ker}(\mathrm{ev}_x:\,C^{\infty}(M_0)\rightarrow\mathbb{R})$ for some $x\in M$. Hence, $I\subseteq\mathfrak{J}_x$ implying $I=\mathfrak{J}_x$ by maximality of $I$.
\end{proof}
Hence, according to prop. \ref{Prop:2.4}, we can identify the real spectrum with $\mathrm{MaxSpec}_{\mathbb{R}}(\mathcal{O}(\mathcal{M}))$ and even obtain a bijection $\Psi:\,M\stackrel{\sim}{\longrightarrow}\mathrm{Spec}_{\mathbb{R}}(\mathcal{O}(\mathcal{M}))$ via
\begin{equation}
M\ni x\stackrel{\sim}{\mapsto}(\mathrm{ev}_x:\,\mathcal{O}(\mathcal{M})\rightarrow\mathbb{R})\in\mathrm{Hom}(\mathcal{O}(\mathcal{M}),\mathbb{R})\stackrel{\sim}{\mapsto}\mathrm{ker}(\mathrm{ev}_x)\in\mathrm{MaxSpec}_{\mathbb{R}}(\mathcal{O}(\mathcal{M}))
\label{eq:2.2.12}
\end{equation}
We want to define a topology on $\mathrm{Spec}_{\mathbb{R}}(\mathcal{O}(\mathcal{M}))$ such that $\Psi$ becomes a homoemorphism. Therefore, note that any section $f\in\mathcal{O}(\mathcal{M})$ canonically induces a morphsim $\phi_f:\,\mathrm{Spec}_{\mathbb{R}}(\mathcal{O}(\mathcal{M}))\rightarrow\mathbb{R}$ setting 
\begin{equation}
\hat{f}(\mathrm{ev}_x):=\mathrm{ev}_x(f)=f(x)
\label{eq:2.2.13}
\end{equation}
Hence, let us endow $\mathrm{Spec}_{\mathbb{R}}(\mathcal{O}(\mathcal{M}))$ with the \emph{Gelfand topology} which is defined as the coarsest topology such that the maps $\phi_f$ for all $f\in\mathcal{O}(\mathcal{M})$ are continuous. A basis of this topology is generated by open subsets of the form 
\begin{equation}
\phi_f^{-1}(B_{\epsilon}(x_0))=\{\mathrm{ev}_{x}\in\mathrm{Spec}_{\mathbb{R}}(\mathcal{O}(\mathcal{M}))|\,|(\mathrm{ev}_{x}-\mathrm{ev}_{x_0})(f)|<\epsilon\}
\label{eq:2.2.14}
\end{equation}
for some $f\in\mathcal{O}(\mathcal{M})$ and $B_{\epsilon}(x_0)\subset M$ an open ball of radius $\epsilon$ around $x_0\in M$. It then follows immediately that the map $\Psi:\,M\stackrel{\sim}{\longrightarrow}\mathrm{Spec}_{\mathbb{R}}(\mathcal{O}(\mathcal{M}))$ is continuous w.r.t. this topology, since
\begin{equation}
\Psi^{-1}(\phi_f^{-1}(B_{\epsilon}(x_0)))=|f|^{-1}(B_{\epsilon}(f(x_0)))
\label{eq:2.2.15}
\end{equation}
is open in $M$ as $|f|:\,M\rightarrow\mathbb{R}$ is continuous. In fact, $\Psi$ is even a homeomorphism. Therefore, consider a closed subset $X\subseteq M$ and let $\mathfrak{p}_X$ be the ideal in $\mathcal{O}(\mathcal{M})$ defined as the set of all sections $f\in\mathcal{O}(\mathcal{M})$ vanishing on $X$. Using a partition of unity argument, it can then be seen that
\begin{equation}
X=\{x\in M|\,f(x)=0,\,\forall f\in\mathfrak{p}_X\}
\label{eq:2.2.16}
\end{equation} 
and thus
\begin{equation}
\Psi(X)=\bigcap_{f\in\mathfrak{p}_X}\phi_f^{-1}(\{0\})
\label{eq:2.2.17}
\end{equation}
i.e., $\Psi(X)$ is closed in $\mathrm{Spec}_{\mathbb{R}}(\mathcal{O}(\mathcal{M}))$ proving that $\Psi$ is indeed a homeomorphism.
\begin{theorem}\label{theorem:2.6}
Let $\mathcal{M}$ and $\mathcal{N}$ be algebro-geometric supermanifolds. Then, their exists a bijective correspondence between the set $\mathrm{Hom}(\mathcal{M},\mathcal{N})$ of morphisms of algebro-geometric supermanifolds and the set $\mathrm{Hom}_{\mathbf{SAlg}}(\mathcal{O}(\mathcal{N}),\mathcal{O}(\mathcal{M}))$ of superalgebra morphisms between the superalgebras of global sections of the respective structure sheaves. 
\end{theorem}
\begin{proof}[Sketch of Proof]
One direction is immediate, i.e., that the pullback of a supermanifold morphism $\phi:\,\mathcal{M}\rightarrow\mathcal{N}$ induces a morphism $\phi^{\sharp}:\,\mathcal{O}(\mathcal{N})\rightarrow\mathcal{O}(\mathcal{M})$ of the respective structure sheaves. The proof of the inverse direction uses a standard tool in algebraic geometry called \emph{localization of rings}. See \cite{Carmeli:2011} for more details. 
\end{proof}
Hence, according to this theorem, in the following, we will identify the $\mathcal{T}$-points $\mathcal{M}(\mathcal{T})$ of an algebro-geometric supermanifold $\mathcal{M}$ with $\mathrm{Hom}(\mathcal{O}(\mathcal{M}),\mathcal{O}(\mathcal{T}))$.\\
For instance, let us consider the $\mathcal{T}$-points $\mathbb{R}^{m|n}(\mathcal{T})=\mathrm{Hom}(\mathcal{O}(\mathbb{R}^{m|n}),\mathcal{O}(\mathcal{T}))$ of the superspace $\mathbb{R}^{m|n}$. By the global chart theorem, this set can be identified with 
\begin{align}
\mathbb{R}^{m|n}(\mathcal{T})&\cong\{(t^1,\ldots,t^m,\theta^1,\ldots,\theta^n)|\,t^i\in\mathcal{O}(\mathcal{T})_0,\,\theta^j\in\mathcal{O}(\mathcal{T})_1\}=\mathcal{O}(\mathcal{T})_0^m\oplus\mathcal{O}(\mathcal{T})_1^n\nonumber\\
&=(\mathcal{O}(\mathcal{T})\otimes\mathbb{R}^{m|n})_0
\label{eq:2.2.18}
\end{align}
For $J(\mathcal{T}):=\{f\in\mathcal{O}(\mathcal{T})|\,\text{$f$ is nilpotent}\}$ the ideal of nilpotent sections of $\mathcal{O}(\mathcal{T})$, this yields the canonical projection $\epsilon:\,\mathcal{O}(\mathcal{T})\rightarrow\mathcal{O}(\mathcal{T})/J(\mathcal{T})\cong C^{\infty}(\mathbf{B}(\mathcal{T}))$ which can be extended to the body map
\begin{equation}
\epsilon_{m,n}:\,(\mathcal{O}(\mathcal{T})\otimes\mathbb{R}^{m|n})_0\rightarrow C^{\infty}(\mathbf{B}(\mathcal{T}))^m
\label{eq:2.2.19}
\end{equation}
In the following, we want to restrict to a subclass of supermanifolds $\mathcal{T}\in\mathbf{SMan}_{\mathrm{Alg}}$ for which $C^{\infty}(\mathbf{B}(\mathcal{T}))\cong\mathbb{R}$, i.e., for which the underlying topological space $T=\{*\}$ just consists of a single point. Hence, if follows $\mathcal{T}\cong(\{*\},\Lambda_N)=\mathbb{R}^{0|N}$ for some $N\in\mathbb{N}_0$. 
\begin{definition}
An algebro-geometric supermanifold $\mathcal{T}$ is called a \emph{super point} if the underlying topological space $T$ only consists of a single point. The subclass of super points form a full subcategory $\mathbf{SPoint}$ of $\mathbf{SMan}_{\mathrm{Alg}}$ called the \emph{category of superpoints}.
\end{definition}
\begin{prop}[see \cite{Mol:10,Sac:08}]
Let $\mathbf{Gr}$ be the \emph{category of (finite-dimensional) Grassmann algebras} whose objects are given by equivalence classes of Grassmann algebras $\Lambda_N\in\mathbf{Ob}(\mathbf{Gr})$, $N\in\mathbb{N}_0$, and for $\Lambda_N,\Lambda_{N'}\in\mathbf{Ob}(\mathbf{Gr})$, $\mathrm{Hom}_{\mathbf{Gr}}(\Lambda_N,\Lambda_{N'})$ is given by the set of superalgebra morphisms between Grassmann algebras. Then, the assignment
\begin{align}
\mathbf{Gr}^{\mathrm{op}}\rightarrow\mathbf{SPoint},\,\Lambda_N&\mapsto(\{*\},\Lambda_N)\\
(\phi:\,\Lambda_N\rightarrow\Lambda_{N'})&\mapsto(\mathrm{id}_{\{*\}},\phi)\nonumber
\label{eq:2.2.20}
\end{align}
yields an equivalence of categories.\qed
\end{prop}
In the following, we will therefore identify superpoints with finite-dimensional Grassmann algebras. From (\ref{eq:2.2.18}), it follows for $N\in\mathbb{N}_0$
\begin{equation}
\mathbb{R}^{m|n}(\Lambda_N)\cong(\Lambda_N\otimes\mathbb{R}^{m|n})_0=:\Lambda^{m,n}_{N}
\label{eq:2.2.21}
\end{equation} 
We equip $\Lambda^{m,n}_{N}$ with the coarsest topology such that the body map $\epsilon_{m,n}:\,\Lambda^{m,n}_{N}\rightarrow\mathbb{R}^m$ is continuous, the so-called \emph{DeWitt}-topology. Next, we want to find a suitable notion of smooth functions on $\Lambda^{m,n}_{N}$.\\
By the global chart theorem, a section $f\in\mathcal{O}(\mathbb{R}^{m|n})$ can be identified with morphism a morphism $f:\,\mathbb{R}^{m|n}\rightarrow\mathbb{R}^{1|1}$. According to \eqref{eq:2.2.9}, this in turn induces a natural transformation $f_\mathcal{T}:\,\mathbb{R}^{m|n}(\mathcal{T})\rightarrow\mathbb{R}^{1|1}(\mathcal{T})$ between the respective functor of points. Sticking to Grassmann algebras, we want to find an explicit form of $f_{\Lambda_N}$. Therefore, let $(x,\xi)\in\Lambda^{m,n}_{N}$ which we can identify with a morphism $g:\,\mathbb{R}^{0|N}\rightarrow\mathbb{R}^{m|n}$ such that $g^{\sharp}(t^i)=x^i$ and $g^{\sharp}(\theta^j)=\xi^j$ $\forall i=1,\ldots,m$, $j=1,\ldots,n$. It then follows again from the global chart theorem that $f_{\Lambda_N}(x,\xi)$ can be identified with an element in $\Lambda_N^{1,1}\equiv\Lambda_N$ whose even and odd part is given by $g^{\sharp}(f^{\sharp}(t))$ and $g^{\sharp}(f^{\sharp}(\theta))$, respectively, where $t$ and $\theta$ denote the global sections of $\mathcal{O}(\mathbb{R}^{1|1})$. Thus, expanding $f=\sum_I{f_I\theta^I}$, this yields
\begin{align}
f_{\Lambda_N}(x,\xi)&=g^{\sharp}(f^{\sharp}(t))+g^{\sharp}(f^{\sharp}(\theta))=\sum_{I,J}{\frac{1}{I!}\partial_If_J\big{|}_{|g|}s(g^{\sharp}(t))^I(g^{\sharp}(\theta))^J}\nonumber\\
&=\sum_{I,J}{\frac{1}{I!}\partial_If_J(\epsilon_{m,n}(x))s(x)^I\xi^J}=:\sum_J{\mathbf{G}(f_J)(x)\xi^J}
\label{eq:2.2.23}
\end{align}
where $s(x):=x-\epsilon_{m,n}(x)$ is the \emph{soul map} and 
\begin{equation}
\mathbf{G}(f)(x):=\sum_{I}{\frac{1}{I!}\partial_I f(\epsilon_{m,n}(x))s(x)^I}
\label{eq:2.2.23.1}
\end{equation}
is called the \emph{Grassmann-analytic continuation} of $f$ or simply its \emph{$\mathbf{G}$-extension}. Functions of the form (\ref{eq:2.2.23}) are precisely supersmooth functions in the sense of Rogers-DeWitt! In the standard literature, one also calls them of class $H^{\infty}$. As a result, $\Lambda^{m,n}_{N}$ together with functions of the form (\ref{eq:2.2.23}) yields a Rogers-DeWitt supermanifold or \emph{$H^{\infty}$-supermanifold}.
The assignment
\begin{align}
\mathrm{Hom}_{\mathbf{SMan}_{\mathrm{Alg}}}(\mathbb{R}^{m|n},\mathbb{R}^{1|1})&\rightarrow H^{\infty}(\Lambda^{m,n}_{N})\label{eq:2.2.22}\\
f&\mapsto f_{\Lambda_N}\nonumber
\end{align}
is clearly surjective but in general not injective unless $N\geq n$.\\
We next want to extend these considerations from superspaces to arbitrary algebro-geometric supermanifolds. Therefore, we make the following definition.
\begin{definition}
For $N\in\mathbb{N}$, the functor $\mathbf{H}_N:\,\mathbf{SMan}_{\mathrm{Alg}}\rightarrow\mathbf{Sets}$ is defined on objects $\mathcal{M}\in\mathbf{Ob}(\mathbf{SMan}_{\mathrm{Alg}})$ via 
\begin{equation}
\mathbf{H}_N(\mathcal{M}):=\mathcal{M}(\Lambda_N)=\mathrm{Hom}(\mathcal{O}(\mathcal{M}),\Exterior\mathbb{R}^N)
\label{eq:2.2.24}
\end{equation}
and on morphisms $f:\,\mathcal{M}\rightarrow\mathcal{N}$ according to 
\begin{equation}
\mathbf{H}_N(f):\,\mathrm{Hom}(\mathcal{O}(\mathcal{M}),\Exterior\mathbb{R}^N)\rightarrow\mathrm{Hom}(\mathcal{O}(N),\Exterior\mathbb{R}^N),\,\phi\mapsto\phi\circ f^{\sharp}
\label{eq:2.2.25}
\end{equation}
\end{definition}
The set $\mathcal{M}(\Lambda_N)$ contains the real spectrum $\mathrm{Spec}_{\mathbb{R}}(\mathcal{O}(\mathcal{M}))=\mathrm{Hom}(\mathcal{O}(\mathcal{M}),\mathbb{R})$ as a proper subset. According to prop. \ref{Prop:2.4} (see also (\ref{eq:2.2.12})), this set can be identified with $M$ and thus, in particular, naturally inherits a topology. Using this property, we again introduce the \emph{DeWitt-topology} on $\mathcal{M}(\Lambda_N)$ to be coarsest topology such that the projection\footnote{in case $\mathcal{M}$ is simply given by the superspace $\mathbb{R}^{m|n}$, this coincides with the body map (\ref{eq:2.2.19}).} 
\begin{equation}
\mathbf{B}:\,\mathcal{M}(\Lambda_N)\rightarrow\mathrm{Spec}_{\mathbb{R}}(\mathcal{O}(\mathcal{M}))\cong M,\,\psi\mapsto\epsilon\circ\psi
\label{eq:2.2.25.1}
\end{equation} 
is continuous. 
\begin{prop}\label{neuProp}
Let $U\subseteq M$ be an open subset of the underlying topological space $M$ of an algebro-geometric supermanifold $\mathcal{M}$. Let us identify $U$ via $\Psi:\,M\rightarrow\mathrm{Spec}_{\mathbb{R}}(\mathcal{O}(\mathcal{M})),\,x\mapsto\mathrm{ev}_x$ with an open subset in the real spectrum. Then, it follows that the open subsets $\mathbf{B}^{-1}(U)$ in $\mathcal{M}(\Lambda_N)$ are given by
\begin{equation}
\mathbf{B}^{-1}(U):=\{\psi:\,\mathcal{O}(\mathcal{M})\rightarrow\Lambda_N|\,\epsilon\circ\psi=\mathrm{ev}_x\text{ for some }x\in U\}
\label{eq:}
\end{equation}
In particular, one has $\mathbf{B}^{-1}(U)=\mathcal{M}|_U(\Lambda_N)$ with $\mathcal{M}|_U:=(U,\mathcal{O}_{\mathcal{M}}|_U)$.
\end{prop}
\begin{proof}
The first assertion is immediate, since $\psi\in\mathbf{B}^{-1}(U)$ if and only if $\epsilon\circ\psi\in\Psi(U)$, i.e., $\epsilon\circ\psi=\mathrm{ev}_x$ for some $x\in U$.\\
To prove the last one, note that, by theorem \ref{theorem:2.6}, one can identify a superalgebra morphism $\psi:\,\mathcal{O}(\mathcal{M})\rightarrow\Lambda_N$ with the pullback of a supermanifold morphism $\phi:=(|\phi|,\phi^{\sharp}):\,\mathbb{R}^{0|N}=(\{*\},\Lambda_N)\rightarrow\mathcal{M}$. For any $f\in\mathcal{O}(\mathcal{M})$, $\epsilon(\phi^{\sharp}(f))$ is defined as the unique real number such that $\phi^{\sharp}(f)-\epsilon(\phi^{\sharp}(f))$ is not invertible. This is precisely the definition of the value of a section of $\Lambda_N$ at $\{*\}$, i.e., $\epsilon(\phi^{\sharp}(f))=\phi^{\sharp}(f)(\{*\})=f(|\phi|(\{*\}))$. Since, $\epsilon\circ\phi^{\sharp}=\mathrm{ev}_x$ for some $x\in M$, this yields $f(|\phi|(\{*\}))=\mathrm{ev}_x(f)=f(x)$ for any $f\in\mathcal{O}(\mathcal{M})$ which implies $|\phi|(\{*\})=x$.\\
Note that $\phi^{\sharp}$ is a morphism of sheaves and thus, in particular, commutes with restrictions. Hence, if, for $f\in\mathcal{O}(\mathcal{M})$, there exists an open neighborhood $x\in V$ such that $f|_V=0$, then $\phi^{\sharp}(f)=0$. That is, $\phi^{\sharp}$ is uniquely determined by the induced stalk morphism $\phi^{\sharp}_x:\,\mathcal{O}_{\mathcal{M},x}\rightarrow\Lambda_N$. From this, it is immediate to see that any $\psi\in\mathcal{M}|_U(\Lambda_N)=\mathrm{Hom}(\mathcal{O}_{\mathcal{M}}(U),\Lambda_N)$ can trivially be extended to a morphism $\psi:\,\mathcal{O}(\mathcal{M})\rightarrow\Lambda_N$ satisfying $\epsilon\circ\psi=\mathrm{ev}_x$ for some $x\in U$, i.e., $\psi\in\mathbf{B}^{-1}(U)$. Conversely, it follows that any morphism in $\mathbf{B}^{-1}(U)$ arises in this way. This proves the last assertion.
\end{proof}
By the local property, for any $x\in M$, there exists an open subset $x\in U\subseteq M$ such that $\mathcal{M}|_U$ is isomorphic to a superdomain $\mathcal{U}^{m|n}$ which is a submanifold of the superspace $\mathbb{R}^{m|n}$. Applying the functor (\ref{eq:2.2.24}) and using prop. \ref{neuProp}, we thus obtain an isomorphism 
\begin{equation}
\mathbf{B}^{-1}(U)=\mathcal{M}|_{U}(\Lambda_N)\rightarrow\mathcal{U}^{m|n}(\Lambda_N)\subseteq\mathbb{R}^{m|n}(\Lambda_N)
\label{eq:}
\end{equation}
i.e., a local \emph{superchart} of $\mathcal{M}(\Lambda_N)$. By (\ref{eq:2.2.23}), it follows immediately that the transition map between two local supercharts defines a $H^{\infty}$-smooth function. As a consequence, $\mathcal{M}(\Lambda_N)$ indeed carries the structure of a $H^{\infty}$ supermanifold. Hence, it follows that the $\Lambda$-points of an algebro-geometric supermanfold naturally define supermanifolds in the sense of Rogers-DeWitt (or more generally $\mathcal{A}$-manifolds in the sense of Tuynman \cite{Tuynman:2004}). Moreover, the corresponding topological space $\mathbf{B}(\mathcal{M}(\Lambda_N))=\mathrm{Spec}_{\mathbb{R}}(\mathcal{O}(\mathcal{M}))$ has the structure of an ordinary $C^{\infty}$-manifold.
\begin{remark}
Just for sake of completeness, note that each $\mathcal{M}\in\mathbf{Ob}(\mathbf{SMan}_{\mathrm{Alg}})$ gives rise to the obvious functor
\begin{equation}
\mathcal{M}:\,\mathbf{Gr}\rightarrow\mathbf{Top}
\label{eq:2.2.26}
\end{equation}
which maps Grassmann-algebras $\Lambda$ to $\Lambda$-points $\mathcal{M}(\Lambda)$. This is precisely a supermanifold in the sense of Molotkov-Sachse.
\end{remark}
Similar to (\ref{eq:2.2.22}), for any $U\subseteq M$ open, one obtains a map
\begin{align}
\mathcal{O}_{\mathcal{M}}(U)\cong\mathrm{Hom}(\mathcal{M}|_U,\mathbb{R}^{1|1})&\rightarrow H^{\infty}(\mathcal{M}|_U(\Lambda_N))=\mathbf{B}_{*}H^{\infty}_{\mathcal{M}(\Lambda_N)}(U)\label{eq:2.2.27}\\
f&\mapsto f_{\Lambda_N}\nonumber
\end{align}
which is generally surjective but injective iff $N\geq n$. In particular, one can show that it defines a morphism of sheaves, i.e., it commutes with restrictions.\\ 
Consider next a $H^{\infty}$ supermanifold $\mathcal{K}\in\mathbf{Ob}(\mathbf{SMan}_{H^{\infty}})$. To $\mathcal{K}$, one can associate the \emph{body} $\mathbf{B}(\mathcal{K})$ defined as the subset of $\mathcal{K}$ given by
\begin{equation}
\mathbf{B}(\mathcal{K}):=\{x\in\mathcal{K}|\,f(x)\in\mathbb{R},\,\forall f\in H^{\infty}(\mathcal{K})\}
\label{eq:}
\end{equation}
which, by definition, has the structure of an ordinary smooth manifold.
This can be extended to morphisms $f:\,\mathcal{K}\rightarrow\mathcal{L}$ between $H^{\infty}$ supermanifolds, setting\footnote{Note that $f(\mathbf{B}(\mathcal{K}))\subseteq\mathbf{B}(\mathcal{L})$ so that $f|_{\mathbf{B}(\mathcal{K})}:\,\mathbf{B}(\mathcal{K})\rightarrow\mathbf{B}(\mathcal{L})$ is indeed well-defined. In fact, for any $g\in H^{\infty}(\mathcal{L})$ and $x\in\mathbf{B}(\mathcal{K})$, it follows $g(f(x))=(g\circ f)(x)\in\mathbb{R}$ as $g\circ f$ is smooth and therefore $f(x)\in\mathbf{B}(\mathcal{L})$.} $\mathbf{B}(f):=f|_{\mathbf{B}(\mathcal{K})}:\,\mathbf{B}(\mathcal{K})\rightarrow\mathbf{B}(\mathcal{L})$. This yields a functor $\mathbf{B}:\,\mathbf{SMan}_{H^{\infty}}\rightarrow\mathbf{Man}$ called the \emph{body functor}. In case $\mathcal{K}$ is given by a $\Lambda_N$-point $\mathcal{M}(\Lambda_N)$ of an algebro-geometric supermanifold $\mathcal{M}\in\mathbf{Ob}(\mathbf{SMan}_{\mathrm{Alg}})$ with odd dimension bounded by $N$, one can identify $\mathbf{B}(\mathcal{K})$ with the real spectrum $\mathrm{Spec}_{\mathbb{R}}(\mathcal{O}(\mathcal{M}))$ justifying the notation.\\
In fact, in this case, \eqref{eq:2.2.27} implies that smooth functions on $\mathcal{K}$ are given by natural transformations $f_{\Lambda_N}$ induced by morphisms $f\in\mathrm{Hom}(\mathcal{M},\mathbb{R}^{1|1})$. For $\phi\in\mathcal{K}=\mathrm{Hom}(\mathcal{O}(\mathcal{M}),\Lambda_N)$, $f_{\Lambda_N}(\phi)$ can be identified with the element $\phi\circ f^{\sharp}\in\Lambda_N$. Hence, $\phi\in\mathbf{B}(\mathcal{N})\Leftrightarrow f_{\Lambda_N}(\phi)\in\mathrm{Hom}(\mathcal{O}(\mathbb{R}^{1|1}),\mathbb{R})\cong\mathbb{R}$ $\forall f\in\mathrm{Hom}(\mathcal{M},\mathbb{R}^{1|1})$ if and only if $\phi\in\mathrm{Hom}(\mathcal{O}(\mathcal{M}),\mathbb{R})$, that is, iff $\phi$ is contained in the real spectrum $\mathrm{Spec}_{\mathbb{R}}(\mathcal{O}(\mathcal{M}))$.\\
\\
To any $H^{\infty}$ supermanifold $\mathcal{K}$, one can associate the locally ringed space
\begin{equation}
\mathbf{A}(\mathcal{K}):=(\mathbf{B}(\mathcal{K}),\mathbf{B}_{*}H^{\infty}_{\mathcal{K}})
\label{eq:}
\end{equation}
which has the structure of an algebro-geometric supermanifold. A morphism $f:\,\mathcal{K}\rightarrow\mathcal{L}$ between $H^{\infty}$ supermanifolds $\mathcal{K}$ and $\mathcal{L}$ canonically induces a morphism 
\begin{equation}
\mathbf{A}(f)=(f|_{\mathbf{B}(\mathcal{K})},f^*):\,\mathbf{A}(\mathcal{K})\rightarrow\mathbf{A}(\mathcal{L})
\label{eq:}
\end{equation}
between the corresponding algebro-geometric supermanifolds, where $f^*$ denotes the ordinary pullback of smooth functions. Hence, this yields a functor 
\begin{equation}
\mathbf{A}:\,\mathbf{SMan}_{H^{\infty}}\rightarrow\mathbf{SMan}_{\mathrm{Alg}}
\label{eq:}
\end{equation}
Let us restrict $\mathbf{H}_N$ to the full subcategory $\mathbf{SMan}_{\mathrm{Alg},N}\subset\mathbf{SMan}_{\mathrm{Alg},N}$ of algebro-geometric supermanifolds with odd dimension bounded by $N$. Then, based on the previous observations, if follows $\mathbf{A}(\mathbf{H}_N(\mathcal{M}))\cong\mathcal{M}$ for any $\mathcal{M}\in\mathbf{Ob}(\mathbf{SMan}_{\mathrm{Alg},N})$. In fact, we have the following.
\begin{theorem}
The functor $\mathbf{A}\circ\mathbf{H}_N:\,\mathbf{SMan}_{\mathrm{Alg},N}\rightarrow\mathbf{SMan}_{\mathrm{Alg},N}$ is naturally equivalent to the identity functor $\mathrm{id}:\,\mathbf{SMan}_{\mathrm{Alg},N}\rightarrow\mathbf{SMan}_{\mathrm{Alg},N}$. 
\end{theorem}
\begin{proof}
We have to show that the following diagrams are commutative
\begin{displaymath}
\xymatrix{
       M \ar[r]^{|f|}\ar[d]_{\Psi} &   N \ar[d]^{\Psi} \\
         \mathrm{Spec}_{\mathbb{R}}(\mathcal{O}(\mathcal{M})) \ar[r]  &     \mathrm{Spec}_{\mathbb{R}}(\mathcal{O}(\mathcal{N}))    
     }\quad\quad		
\xymatrix{
        \mathcal{O}_{\mathcal{N}} \ar[rr]^{f^{\sharp}} \ar[d]_{\cong}&  &  f_{*}\mathcal{O}_{\mathcal{M}} \ar[d]^{\cong} \\
         \mathbf{B}_{*}H^{\infty}_{\mathcal{N}(\Lambda_N)} \ar[rr]^{\mathbf{H}_N(f)^*}  &     &  f_{*}(\mathbf{B}_{*}H^{\infty}_{\mathcal{M}(\Lambda_N)})    
     }
		\label{eq:}
 \end{displaymath}
for any $\mathcal{M},\mathcal{N}\in\mathbf{Ob}(\mathbf{SMan}_{\mathrm{Alg},N})$ and morphisms $f=(|f|,f^{\sharp}):\,\mathcal{M}\rightarrow\mathcal{N}$ where, in the diagram on the left, the lower arrow is given by the restriction of $\mathbf{H}_N(f)$ to the real spectrum $\mathrm{Spec}_{\mathbb{R}}(\mathcal{O}(\mathcal{M}))$.\\
That the left diagram commutes follows immediately, since, by definition \eqref{eq:2.2.25}, we have
\begin{equation}
\mathbf{H}_N(f)(\mathrm{ev_x})=\mathrm{ev}_x\circ f^{\sharp}=\mathrm{ev}_{|f|(x)}
\label{eq:}
\end{equation}
for any $x\in M$. To see the commutativity of the right diagram, note that, by identifying $g\in\mathcal{O}(\mathcal{N})$ with a morphism $g:\,\mathcal{N}\rightarrow\mathbb{R}^{1|1}$. the pullback $f^{\sharp}(g)$ is given by the morphism $g\circ f:\,\mathcal{M}\rightarrow\mathbb{R}^{1|1}$. Moreover, identifying $\phi\in\mathcal{M}(\Lambda_N)$ with a morphism $\phi:\,\mathbb{R}^{0|N}\rightarrow\mathcal{M}$, we have $\mathbf{H}_N(f)(\phi)=f\circ\phi$. Thus, this yields
\begin{equation}
\mathbf{H}_N(f)^*(g_{\Lambda_N})(\phi)=g_{\Lambda_N}(\mathbf{H}_N(f)(\phi))=g\circ f\circ\phi=f^{\sharp}(g)\circ\phi=(f^{\sharp}(g))_{\Lambda_N}(\phi)
\label{eq:}
\end{equation}
for any $g\in\mathcal{O}(\mathcal{N})$ and $\phi\in\mathcal{M}(\Lambda_N)$. This proves the theorem.
\end{proof}
Conversely, it can be shown that $\mathbf{H}_N\circ\mathbf{A}$ is naturally equivalent to the identity functor on the full subcategory $\mathbf{SMan}_{H^{\infty},N}$ of $H^{\infty}$ supermanifolds with odd dimension bounded by $N$ (see also \cite{Batchelor:1980,Rogers:1980}). Thus, the functors $\mathbf{H}_N:\,\mathbf{SMan}_{\mathrm{Alg},N}\rightarrow\mathbf{SMan}_{H^{\infty},N}$ and $\mathbf{A}:\,\mathbf{SMan}_{H^{\infty},N}\rightarrow\mathbf{SMan}_{\mathrm{Alg},N}$ provide an equivalence of categories.\\
\\
To summarize, any algebro-geometric supermanfold induces a functor of the form (\ref{eq:2.2.26}) assigning Grassmann-algebras to the corresponding $H^{\infty}$ supermanifold. Moreover, in case that the number of odd generators of the Grassmann-algebra is large enough, via (\ref{eq:2.2.24}) one even obtains an equivalence of categories which allows one to uniquely reconstruct the underlying algebro-geometric supermanifold. For this reason, many constructions on algebro-geometric supermanifolds can equivalently be performed on the corresponding $H^{\infty}$ supermanifolds (in fact, we will mainly do so in what follows as $H^{\infty}$ manifolds are often easier to handle for physical applications).\\
However, the choice of a particular Grassmann-algebra is completely ambigiuous and therefore tend to have superfluous (physical) degrees of freedom. Consequently, any definition made on a $H^{\infty}$ supermanifold should not depend on a particular but, in the sense of Molotkov-Sachse, behave functorially under the change of a Grassmann-algebra. In the following, working with a particular $H^{\infty}$ supermanifold $\mathcal{M}$, we will only assume that the number of odd generators of the Grassmann-algebra $\Lambda$ over which $\mathcal{M}$ is modeled is large enough, i.e., greater than the odd dimension of $\mathcal{M}$\footnote{for this reason, in the standard literature, one typically chooses the infinite dimensional Grassmann-algebra $\Lambda_{\infty}$ generated by infinite number of Grassmann-generators which may be obtained as an inductive limit of the finite-dimensional ones. Also in this case, one can show that the category of algebro-geometric and $H^{\infty}$ supermanifolds modeled over $\Lambda_{\infty}$ are indeed equivalent \cite{Batchelor:1980,Rogers:1980}.}. 

\subsection{The super Poncaré and anti-de Sitter group}\label{section:2.3}
Let us recall the basic super Lie groups and their corresponding super Lie algebras which play a central role in supersymmetry and supergravity and which will provide the foundations for the derivation of the super Ashtekar connection for extended supergravity theories in four spacetime dimensions.\\
A $H^{\infty}$ super Lie group $\mathcal{G}$ is a group object in the category $\mathbf{SMan}_{H^{\infty}}$ (see for instance \cite{Rogers:2007,Tuynman:2004} and references therein). The \emph{super Lie module} $\mathrm{Lie}(\mathcal{G})$ is defined as the tangent space $T_e\mathcal{G}$ at the identity $e\in\mathcal{G}$ and turns out to have the structure of a super $\Lambda$-vector space $\mathrm{Lie}(\mathcal{G})=:\mathfrak{g}\otimes\Lambda$ with $\mathfrak{g}$ a super Lie algebra over $\mathbb{R}$. In fact, $\mathfrak{g}$ can be identified with the super Lie algebra of the associated algebro-geometric super Lie group.
\begin{definition}\label{bilinear map}
A (homogeneous) \emph{super bilinear form} $B$ of parity $|B|\in\mathbb{Z}_2$ on a super $\Lambda^{\mathbb{C}}$-vector space $\mathcal{V}=V\otimes\Lambda^{\mathbb{C}}$, $\Lambda^{\mathbb{C}}:=\Lambda\otimes\mathbb{C}$, is a right-linear map $B:\,\mathcal{V}\times\mathcal{V}\rightarrow\Lambda^{\mathbb{C}}$ which satisfies $B(\mathcal{V}_i,\mathcal{V}_j)\subseteq(\Lambda^{\mathbb{C}})_{|B|+i+j}$ and is graded symmetric, i.e.,
\begin{equation}
B(v,w)=(-1)^{|v||w|}B(w,v),\quad\forall v,w\in\mathcal{V}
\label{eq:2.3.1}
\end{equation}
Moreover, we require $B$ to be \emph{smooth} in the sense that $B(V,V)\subseteq\mathbb{C}$ \cite{Tuynman:2018}. Finally, the smooth super bilinear form $B$ is called non-degenerate if for any $v\in V$ there exists $w\in V$ such that $B(v,w)\neq 0$.
\end{definition}

\begin{remark}
From the smoothness requirement of a super bilinear form $B$, it follows immediately that $B|_{V\times V}$ defines a graded symmetric bilinear form on the super vector space $V$ in the sense of \cite{Cheng:2012}, i.e., $B(V_i,V_j)=0$ unless $i+j=|B|$. Moreover, $B|_{V_0\times V_0}$ is symmetric and $B|_{V_1\times V_1}$ is antisymmetric. If $B$ is furthermore even and non-degenerate then so is $B|_{V\times V}$ which implies that $V_1$ is necessarily even dimensional and one can always find a homogeneous basis $(e_i,f_j)$ of $V$ (resp. $\mathcal{V}$) such that, w.r.t. this basis, $B$ takes the form
\begin{equation}
\begin{pmatrix}
	\mathds{1}_m & 0\\
	0 & J_{2n}
\end{pmatrix}
\label{eq:2.3.2}
\end{equation}
where $\mathrm{dim}\,V_0=m$ and $\mathrm{dim}\,V_1=2n$ and $J_{2n}$ is the standard symplectic form on $\mathbb{C}^{2n}$
\begin{equation}
J_{2n}:=\begin{pmatrix}
	0 & \mathds{1}_n\\
	-\mathds{1}_n & 0
\end{pmatrix}
\label{eq:2.3.3}
\end{equation}
We will call (\ref{eq:2.3.3}) the \emph{standard representation} of $B$.
\end{remark}
\begin{definition}
Let $(\mathcal{V},B(\,\cdot\,,\,\cdot\,))$ be a super $\Lambda^{\mathbb{C}}$-vector space equipped with an even non-degenerate super bilinear form $B:\,\mathcal{V}\times\mathcal{V}\rightarrow\Lambda^{\mathbb{C}}$. The \emph{orthosymplectic super Lie group} $\mathrm{OSp}(\mathcal{V})$ is defined as the super Lie subgroup of the general linear group $\mathrm{GL}(\mathcal{V})$ consisting of all those group elements that preserve $B$, i.e., $g\in\mathrm{OSp}(\mathcal{V})$ if and only if
\begin{equation}
B(gv,gw)=B(v,w),\quad\forall v,w\in\mathcal{V}
\label{eq:2.3.4}
\end{equation}
It follows from the 'stabilizer theorem' (see e.g. Prop. 5.13 in \cite{Tuynman:2004} or Prop. 8.4.7. in \cite{Carmeli:2011} in the pure algebraic setting) that $\mathrm{OSp}(\mathcal{V})$ defines an embedded super Lie subgroup of the general linear group $\mathrm{GL}(\mathcal{V})$ with super Lie algebra
\begin{equation}
\mathfrak{osp}(\mathcal{V}):=\{X\in\mathfrak{gl}(\mathcal{V})|\,B(Xv,w)+(-1)^{|X||v|}B(v,Xw)\,\forall v,w\in V\}
\label{eq:2.3.5}
\end{equation}
If $(e_i,f_j)$ is homogeneous basis of $\mathcal{V}$ such that $\mathcal{V}\cong(\Lambda^{\mathbb{C}})^{m|2n}$ and $B$ acquires the standard representation (\ref{eq:2.3.3}), the orthosymplectic super Lie group is also simply denoted by $\mathrm{OSp}(m|2n)$. Accordingly. the bosonic sub super Lie algebra takes the form $\mathfrak{osp}(m|2n)_0=\mathfrak{so}(m)\oplus\mathfrak{sp}(2n)$.
\end{definition}
In order to find a graded generalization for the isometry group $\mathrm{SO}(2,3)$ of anti-de Sitter spacetime $\mathrm{AdS}_4$\footnote{The four-dimensional anti-de Sitter spacetime is an embedded submanifold of the semi-Riemannian manifold $\mathbb{R}^{2,3}$ equipped with the metric $\eta_{AB}=\mathrm{diag}(-+++-)$ given by
\begin{equation}
\mathrm{AdS}_4:=\{x\in\mathbb{R}^{2,3}|\,\eta_{AB}x^Ax^B=-L^2\}
\label{eq:2.3.6}
\end{equation}
with $L$ the so-called anti-de Sitter radius}, we consider the following Lie algebra representation of $\mathfrak{so}(2,3)$.\\
Let $\gamma^I$, $I=0,\ldots,3$, be the gamma matrices satisfying the Clifford algebra relations in $D=4$ 
\begin{equation}
\{\gamma_I,\gamma_J\}=2\eta_{IJ}
\label{eq:2.3.7}
\end{equation}
where $\eta_{IJ}=\mathrm{diag}(-+++)$. Similar as in \cite{Nicolai:1984hb} and \cite{Freedman:1983na}, we define totally antisymmetric matrices $\Sigma^{AB}$, $A,B=0,\ldots,4$, via
\begin{equation}
\Sigma^{IJ}:=\frac{1}{2}\gamma^{IJ}:=\frac{1}{4}[\gamma^I,\gamma^J]\quad\text{as well as}\quad\Sigma^{4I}:=-\gamma^{I4}:=\frac{1}{2}\gamma^I
\label{eq:2.3.8}
\end{equation}
where indices are raised and lowered w.r.t. the metric $\eta_{AB}=\mathrm{diag}(-+++-)$. These satisfy the following commutation relations
\begin{equation}
[\Sigma_{AB},\Sigma_{CD}]=\eta_{BC}\Sigma_{AD}-\eta_{AC}\Sigma_{BD}-\eta_{BD}\Sigma_{AC}+\eta_{AD}\Sigma_{BC}
\label{eq:2.3.9}
\end{equation}
and thus indeed provide a representation of $\mathfrak{so}(2,3)$. In fact, choosing a real representation of the gamma matrices, it follows that the charge conjugation matrix $C$ is of the form $C=-iJ_4$ and, by the symmetry properties of the gamma matrices,
\begin{equation}
(C\Sigma_{AB})^T=C\Sigma_{AB}
\label{eq:2.3.10}
\end{equation}
Hence, the $\Sigma_{AB}$ generate $\mathfrak{sp}(4)$ the Lie algebra universal covering group $\mathrm{Sp}(4,\mathbb{R})$ of $\mathrm{SO}(2,3)$. Thus, a candidate for the graded extension of the anti-de Sitter group with $\mathcal{N}$-fermionic generators is given by the orthosymplectic Lie group $\mathrm{OSp}(\mathcal{N}|4)$. We therefore choose $\mathcal{V}=(\Lambda^{\mathbb{C}})^{\mathcal{N}|4}$ as super vector space equipped with the bilinear form
\begin{equation}
\Omega=\begin{pmatrix}
	\mathds{1} & 0\\
	0 & C
\end{pmatrix}
\label{eq:2.3.11}
\end{equation}
The algebra $\mathfrak{osp}(\mathcal{N}|4)$ is then generated by all $X\in\mathfrak{gl}(\mathcal{V})$ satisfying
\begin{equation}
X^{sT}\Omega+\Omega X=0
\label{eq:2.3.12}
\end{equation}
where $X^{sT}$ denotes the super transpose of $X$. Writing $X$ in the block form
\begin{equation}
X=\begin{pmatrix}
	X_{11} & X_{12}\\
	X_{21} & X_{22}
\end{pmatrix}
\label{eq:2.3.13}
\end{equation}
(\ref{eq:2.3.12}) is equivalent to
\begin{equation}
X_{11}^T=-X,\quad(CX_{22})^T=CX_{22},\quad X_{12}=-X_{21}^TC
\label{eq:2.3.14}
\end{equation}
and thus, in particular, $X_{11}\in\mathfrak{so}(\mathcal{N})$ and $X_{22}\in\mathfrak{sp}(4)$. Thus, following \cite{Wipf:2016}, based on the above observation, we set
\begin{equation}
M_{AB}:=\begin{pmatrix}
	0 & 0\\
	0 & \Sigma_{AB}
\end{pmatrix}\quad\text{and}\quad T^{rs}:=\begin{pmatrix}
	A^{rs} & 0\\
	0 & 0
\end{pmatrix}
\label{eq:2.3.15}
\end{equation}
as generators for the bosonic sub super Lie algebras $\mathfrak{sp}(4)$ and $\mathfrak{so}(\mathcal{N})$, respectively, where $(A^{rs})_{pq}:=2\delta_p^{[r}\delta_q^{s]}$, $p,q,r,s=1,\ldots,\mathcal{N}$. For the fermionic generators, we set \cite{Wipf:2016}
\begin{equation}
Q_{\alpha}^r:=\begin{pmatrix}
	0 & -\bar{e}_{\alpha}\otimes e_r\\
	e_{\alpha}\otimes e_r^T & 0
\end{pmatrix}
\label{eq:2.3.16}
\end{equation}
where $(\bar{e}_{\alpha})_{\beta}=C_{\alpha\beta}$. It then follows by direct computation that
\begin{equation}
[M_{AB},Q_{\alpha}^r]=Q_{\beta}^r\tensor{(\Sigma_{AB})}{^{\beta}_{\alpha}}\quad\text{and}\quad[T^{pq},Q_{\alpha}^r]=\delta^{qr}Q_{\alpha}^p-\delta^{pr}Q_{\alpha}^q
\label{eq:2.3.17}
\end{equation}
In order to compute the Lie bracket between two fermionic generators, one can use the Fierz identity $2(e_{\alpha}\bar{e}_{\beta}+e_{\beta}\bar{e}_{\alpha})=\gamma_I(C\gamma^I)_{\alpha\beta}+\frac{1}{2}\gamma_{JI}(C\gamma^{IJ})_{\alpha\beta}$ where the sum terminates after second order in the gamma matrices as the higher rank gamma matrices are antisymmetric with respect to the charge conjugation $C$. Thus, one finds
\begin{equation}
[Q_{\alpha}^r,Q_{\beta}^s]=\delta^{rs}(C\Sigma^{AB})M_{AB}-C_{\alpha\beta}T^{rs}
\label{eq:2.3.18}
\end{equation}
Defining $P_I:=\Sigma_{4I}$ and reintroducing the cosmological constant by rescaling $P_I\rightarrow P_I/L$ and $Q_{\alpha}^r\rightarrow Q_{\alpha}^r/\sqrt{2L}$, one finally ends up with the following (graded) commutation relations
\begin{align}
[M_{IJ},Q^r_{\alpha}]&=\frac{1}{2}Q_{\beta}^r\tensor{(\gamma_{IJ})}{^{\beta}_{\alpha}}\label{eq:2.3.19}\\
[P_I,Q^r_{\alpha}]&=-\frac{1}{2L}Q^r_{\beta}\tensor{(\gamma_I)}{^{\beta}_{\alpha}}\label{eq:2.3.20}\\
[P_I,P_J]&=\frac{1}{L^2}M_{IJ}\label{eq:2.3.21}\\
[Q_{\alpha}^r,Q_{\beta}^s]=\delta^{rs}\frac{1}{2}(C\gamma^I)_{\alpha\beta}P_I+&\delta^{rs}\frac{1}{4L}(C\gamma^{IJ})_{\alpha\beta}M_{IJ}-\frac{1}{2L}C_{\alpha\beta}T^{rs}
\label{eq:2.3.22}
\end{align}
which is in the form we will use in what follows. Performing the Inönü-Wigner contraction, i.e., taking the limit $L\rightarrow\infty$, one obtains the super Poincaré Lie algebra. 
\begin{remark}\label{remark:2.13}
The super Poincaré group can also be desribed as the isometry group $\mathrm{ISO}(\mathbb{R}^{1,3|4})$ of super Minkowski spacetime $\mathbb{R}^{1,3|4\mathcal{N}}$. For instance, for $\mathcal{N}=1$, super Minkowksi spacetime is given by the split supermanifold $\mathbf{S}(\Delta_{\mathbb{R}},\mathbb{R}^{1,3})$ associated to the trivial vector bundle $\mathbb{R}^{1,3}\times\Delta_{\mathbb{R}}$ with $\Delta_{\mathbb{R}}$ the four-dimensional real Majorana representation of $\mathrm{Spin}^+(1,3)$. The super Poincaré group is then given by the semi-direct product\footnote{for any $C^{\infty}$ manifold $M$, the supermanifold $\mathbf{S}(M)$ is defined by applying the split functor on the trivial vector bundle $M\times\{0\}\rightarrow M$. This yields a functor $\mathbf{S}:\,\mathbf{Man}\rightarrow\mathbf{SMan}$.}
\begin{equation}
\mathrm{ISO}(\mathbb{R}^{1,3|4})=\mathcal{T}^{1,3|4}\rtimes_{\Phi}\mathbf{S}(\mathrm{Spin}^+(1,3))
\label{eq:}
\end{equation}
where $\mathcal{T}^{1,3|4}\cong\mathbf{S}(\Delta_{\mathbb{R}},\mathbb{R}^{1,3})$ is the super translation group (see e.g. \cite{Rogers:2007}) and $\Phi:\,\mathbf{S}(\mathrm{Spin}^+(1,3))\rightarrow\mathrm{GL}(\mathcal{T}^{1,3|4})$ is the group action obtained by applying the split functor on the group representation
\begin{equation}
\mathrm{Spin}^+(1,3)\ni g\mapsto\mathrm{diag}(\lambda^+(g),\kappa_{\mathbb{R}}(g))\in\mathrm{GL}(\mathbb{R}^{1,3}\oplus\Pi\Delta_{\mathbb{R}})
\label{eq:}
\end{equation}
of $\mathrm{Spin}^+(1,3)$ on the super vector space $\mathbb{R}^{1,3}\oplus\Pi \Delta_{\mathbb{R}}$ with $\lambda^+:\,\mathrm{Spin}^+(1,3)\rightarrow\mathrm{SO}^+(1,3)$ the universal covering map and $\Pi:\,\mathbf{SVec}\rightarrow\mathbf{SVec}$ the \emph{parity functor}, i.e., $\Pi \Delta_{\mathbb{R}}$ is viewed as a purely odd super vector space.
\end{remark}
\subsection{Super parallel transport map}\label{Holonomies}
One of the main issues while working in the standard category of supermanifolds, both in the $H^{\infty}$ or the algebraic category, is that when restricting on the body, the only nonvanishing field components turn out to be purely bosonic. Hence, there are no fermionic degrees of freedom on the body. This, however, seems to be incompatible with various constructions in physics. For instance, in the D'Auria-Fré formalism of supergraviy \cite{DAuria:1982uck,Castellani:1991et}, one has the so-called \emph{rheonomy principle} stating that physical fields are completely fixed by their pullback to the body manifold. Moreover, supersymmetry transformation depend on an anticommutative fermionic generator and, in particular, are nontrivial by restricting superfields to the body.\\
As we will see, all these issues can be remedied simultaneously by factorising a given supermanifold $\mathcal{M}$ by an additional parametrising supermanifold $\mathcal{S}$ and studying superfields on $\mathcal{S}\times\mathcal{M}$. Therefore,  as will become clear in what follows, we are interested in a certain subclass of such superfields which depend as little as possible on this additional supermanifold making them covariant in a specific sense under a change of parametrisation. This idea is based on a proposal formulated already by Schmitt in \cite{Schmitt:1996hp} motivated by the functorial approach to supermanifold theory according to Molotkov \cite{Mol:10} and Sachse \cite{Sac:09} and which also recently found application in context of superconformal field theories on super Riemannian surfaces \cite{Jost:2014wfa,Kessler:2019bwp} as well as the local approach to super QFTs \cite{Hack:2015vna}.\\
The following presentation summarizes the main results obtained in \cite{Konsti-FB:2020}. We will therefore omit most of the proofs. We adopt the terminology of \cite{Hack:2015vna} introducing the notion of a relative supermanifold. However, unlike as in \cite{Hack:2015vna}, in order to study fermionic fields, we will not restrict on superpoints as parametrising supermanifolds. As we will see at the end of this section, the resulting picture resembles the construction of parallel transport map of super connections on super vector bundles as studied in the purely algebraic setting in \cite{Florin:2008,Groeger:2013aja}. Moreover, the description of fermionic fields turns out to be quite similar to considerations in perturbative algebraic QFT (pAQFT) \cite{Rejzner:2011au,Rejzner:2016hdj}. 
\begin{definition}\label{def:9.1}
Let $\mathcal{S}$ and $\mathcal{M}$ be $H^{\infty}$ supermanifolds. The pair $(\mathcal{S}\times\mathcal{M},\mathrm{pr}_{\mathcal{S}})$ with $\mathrm{pr}_{\mathcal{S}}:\,\mathcal{S}\times\mathcal{M}\rightarrow\mathcal{S}$ the projection onto the first factor is called a \emph{$\mathcal{S}$-relative supermanifold} also denoted by $\mathcal{M}_{/\mathcal{S}}$.\\
A morphism $\phi:\,\mathcal{M}_{/\mathcal{S}}\rightarrow\mathcal{N}_{/\mathcal{S}}$ between $\mathcal{S}$-relative supermanifolds is a morphism $\phi:\,\mathcal{S}\times\mathcal{M}\rightarrow\mathcal{S}\times\mathcal{N}$ of $H^{\infty}$ supermanifolds preserving the projections, i.e. the following diagram is commutative
\begin{displaymath}
	 \xymatrix{
         \mathcal{S}\times\mathcal{M}\ar[rr]^{\phi} \ar[dr]_{\mathrm{pr}_{\mathcal{S}}}  & &    \mathcal{S}\times\mathcal{N} \ar[dl]^{\mathrm{pr}_{\mathcal{S}}}\\
            &    \mathcal{S}   &   
     }
		\label{eq:9.1}
 \end{displaymath}
Hence, $\phi(s,p)=(s,\tilde{\phi}(s,p))$ $\forall (s,p)\in\mathcal{S}\times\mathcal{M}$ with $\tilde{\phi}:=\mathrm{pr}_{\mathcal{N}}\circ\phi:\,\mathcal{S}\times\mathcal{M}\rightarrow\mathcal{N}$. This yields a category $\mathbf{SMan}_{/\mathcal{S}}$ called the \emph{category of $\mathcal{S}$-relative supermanifolds}.
\end{definition}
The following proposition gives a different characterization of morphism between $\mathcal{S}$-relative supermanifolds.
\begin{prop}[after \cite{Hack:2015vna}]\label{prop:9.2}
Let $\mathcal{M}_{/\mathcal{S}},\mathcal{N}_{/\mathcal{S}}\in\mathbf{Ob}(\mathbf{SMan}_{/\mathcal{S}})$ be $\mathcal{S}$-relative supermanifolds. Then, the map
\begin{align}
\alpha_{\mathcal{S}}:\,\mathrm{Hom}_{\mathbf{SMan}_{/\mathcal{S}}}(\mathcal{M}_{/\mathcal{S}},\mathcal{N}_{/\mathcal{S}})&\rightarrow\mathrm{Hom}_{\mathbf{SMan}_{H^{\infty}}}(\mathcal{S}\times\mathcal{M},\mathcal{S}\times\mathcal{N})\label{eq:9.2}\\
(\phi:\,\mathcal{S}\times\mathcal{M}\rightarrow\mathcal{S}\times\mathcal{N})&\mapsto(\mathrm{pr}_{\mathcal{N}}\circ\phi:\,\mathcal{S}\times\mathcal{M}\rightarrow\mathcal{N})\nonumber
\end{align}
is a bijection with the inverse given by
\begin{align}
\alpha_{\mathcal{S}}^{-1}:\,\mathrm{Hom}_{\mathbf{SMan}_{H^{\infty}}}(\mathcal{S}\times\mathcal{M},\mathcal{N})&\rightarrow\mathrm{Hom}_{\mathbf{SMan}_{/\mathcal{S}}}(\mathcal{M}_{/\mathcal{S}},\mathcal{N}_{/\mathcal{S}})\label{eq:9.3}\\
(\psi:\,\mathcal{S}\times\mathcal{M}\rightarrow\mathcal{N})&\mapsto((\mathrm{id}_{\mathcal{S}}\times\psi)\circ(d_{\mathcal{S}}\times\mathrm{id}_{\mathcal{M}}):\,\mathcal{S}\times\mathcal{M}\rightarrow\mathcal{S}\times\mathcal{N})\nonumber
\end{align}
with $d_{\mathcal{S}}:\,\mathcal{S}\rightarrow\mathcal{S}\times\mathcal{S}$ the diagonal map.\qed
\end{prop}
Let $\lambda:\,\mathcal{S}\rightarrow\mathcal{S}'$ be a morphism between parametrising supermanifolds, we will also call such a morphism a \emph{change of parametrisation}. Then, any smooth map $\phi:\,\mathcal{S'}\times\mathcal{M}\rightarrow\mathcal{N}$ can be pulled back via $\lambda$ to a morphism $\lambda^*\phi:=\phi\circ(\lambda\times\mathrm{id}_{\mathcal{M}}):\,\mathcal{S}\times\mathcal{M}\rightarrow\mathcal{N}$. Using \ref{eq:9.3}, this yields the map \cite{Hack:2015vna}
\begin{align}
\lambda^{*}:\,\mathrm{Hom}_{\mathbf{SMan}_{/\mathcal{S}'}}(\mathcal{M}_{/\mathcal{S}'},\mathcal{N}_{/\mathcal{S}'})&\rightarrow\mathrm{Hom}_{\mathbf{SMan}_{/\mathcal{S}}}(\mathcal{M}_{/\mathcal{S}},\mathcal{N}_{/\mathcal{S}})\label{eq:9.4}\\
\phi&\mapsto\alpha_{\mathcal{S}}^{-1}(\alpha_{\mathcal{S}'}(\phi)\circ(\lambda\times\mathrm{id}_{\mathcal{M}}))\nonumber
\end{align}
Hence, for $\phi:\,\mathcal{M}_{/\mathcal{S}'}\rightarrow\mathcal{N}_{/\mathcal{S}'}$, $\lambda^*(\phi)$ explicitly reads
\begin{align}
\lambda^*(\phi)(s,p)=(s,\mathrm{pr}_{\mathcal{N}}\circ\phi(\lambda(s),p))
\label{eq:9.5}
\end{align}
$\forall(s,p)\in\mathcal{S}\times\mathcal{M}$. The following proposition demonstrates that the set of morphisms between relative supermanifolds is functorial in the parametrising supermanifold and thus indeed have the required properties under change of parametrization.
\begin{prop}[after \cite{Hack:2015vna}]
The assignment
\begin{align}
\mathbf{SMan}\rightarrow\mathbf{Set}:\,\mathbf{Ob}(\mathbf{SMan})\ni\mathcal{S}&\mapsto\mathrm{Hom}_{\mathbf{SMan}_{/\mathcal{S}}}(\mathcal{M}_{/\mathcal{S}},\mathcal{N}_{/\mathcal{S}})\in\mathbf{Ob}(\mathbf{Set})\nonumber\\
(\lambda:\,\mathcal{S}\rightarrow\mathcal{S}')&\mapsto\lambda^{*}\nonumber
\label{eq:9.6}
\end{align}
defines a contravariant functor on the category $\mathbf{SMan}$ of $H^{\infty}$ supermanifolds. Moreover, the map $\lambda^{*}$ associated to the morphism $\lambda:\,\mathcal{S}\rightarrow\mathcal{S}'$ preserves compositions, i.e., $\lambda^*(\phi\circ\psi)=\lambda^*(\phi)\circ\lambda^{*}(\psi)$ for any $\psi:\,\mathcal{M}\rightarrow\mathcal{N}$ and $\phi:\,\mathcal{N}\rightarrow\mathcal{L}$
\end{prop}

\begin{definition}
\begin{enumerate}[label=(\roman*)]
	\item Let $\mathcal{M}_{/\mathcal{S}}\in\mathbf{Ob}(\mathbf{SMan}_{/\mathcal{S}})$ be a $\mathcal{S}$-relative supermanifold. A  vector field $X$ on $\mathcal{S}\times\mathcal{M}$ is called \emph{$\mathcal{S}$-relative} if
\begin{equation}
X(f\otimes 1)=0\,\forall f\in H^{\infty}(\mathcal{S})
\label{eq:9.7}
\end{equation}
The $\mathcal{S}$-relative vector fields form a super $H^{\infty}(\mathcal{S}\times\mathcal{M})$-submodule $\Gamma(\mathcal{M}_{/\mathcal{S}})$ of $\Gamma(\mathcal{S}\times\mathcal{M})$ isomorphic to $H^{\infty}(\mathcal{S})\otimes\Gamma(\mathcal{M})$.
	\item A \emph{$\mathcal{S}$-relative $1$-form} $\omega$ on $\mathcal{M}_{/\mathcal{S}}$ is a left-linear morphism of super $H^{\infty}(\mathcal{S}\times\mathcal{M})$-modules $\omega\in\underline{\mathrm{Hom}}_L(\Gamma(\mathcal{M}_{/\mathcal{S}}),H^{\infty}(\mathcal{S}\times\mathcal{M}))$. The set $\Omega^1(\mathcal{M}_{/\mathcal{S}})$ of $\mathcal{S}$-relative $1$-forms on $\mathcal{M}_{/\mathcal{S}}$ defines a super $H^{\infty}(\mathcal{S}\times\mathcal{M})$-submodule of $\Omega^1(\mathcal{S}\times\mathcal{M})$ isomorphic to $\Omega^1(\mathcal{M})\otimes H^{\infty}(\mathcal{S})$.
\end{enumerate}
\end{definition}
Let us consider a $H^{\infty}$ principal super fiber bundle $\mathcal{G}\rightarrow\mathcal{P}\stackrel{\pi}{\rightarrow}\mathcal{M}$ with $\mathcal{G}$-right action $\Phi:\,\mathcal{P}\times\mathcal{G}\rightarrow\mathcal{P}$ as well as a supermanifold $\mathcal{S}$. Taking products, this yields a fiber bundle 
\begin{displaymath}
	 \xymatrix{
         \mathcal{G}\ar[r] &    \mathcal{S}\times\mathcal{P} \ar[d]^{\pi_{\mathcal{S}}}\\
            &    \mathcal{S}\times\mathcal{M}  
     }
		\label{eq:9.8}
 \end{displaymath}
with projection $\pi_{\mathcal{S}}:=\mathrm{id}_{\mathcal{S}}\times\pi$ and $\mathcal{G}$-right action $\Phi_{\mathcal{S}}:=\mathrm{id}_{\mathcal{S}}\times\Phi:\,(\mathcal{S}\times\mathcal{P})\times\mathcal{G}\rightarrow\mathcal{S}\times\mathcal{P}$. Since, by definition, $\pi_{\mathcal{S}}$ and $\Phi_{\mathcal{S}}$ are morphisms of $\mathcal{S}$-relative supermanifolds, this yields a fiber bundle $\mathcal{G}\rightarrow\mathcal{P}_{/\mathcal{S}}\stackrel{\pi_{\mathcal{S}}}{\rightarrow}\mathcal{M}_{/\mathcal{S}}$ in the category $\mathbf{SMan}_{/\mathcal{S}}$ of $\mathcal{S}$-relative supermanifolds which will be called a \emph{$\mathcal{S}$-relative principal super fiber bundle}.
\begin{definition}\label{Def:2.18}
A \emph{$\mathcal{S}$-relative super connection $1$-form} $\mathcal{A}$ on the $\mathcal{S}$-relative principal super fiber bundle $\mathcal{G}\rightarrow\mathcal{P}_{/\mathcal{S}}\stackrel{\pi_{\mathcal{S}}}{\rightarrow}\mathcal{M}_{/\mathcal{S}}$ is an even $\mathrm{Lie}(\mathcal{G})$-valued $\mathcal{S}$-relative 1-form $\mathcal{A}\in\Omega^1(\mathcal{P}_{/\mathcal{S}},\mathfrak{g}):=\Omega^1(\mathcal{P}_{/\mathcal{S}})\otimes\mathfrak{g}$ such that\footnote{As the notation already suggests, forms on a supermanifold $\mathcal{M}$ are regarded as left linear morphisms, i.e., sections of the exterior power of the left-dual tangent bundle $\tensor[^*]{T\mathcal{M}}{}$ of $\mathcal{M}$.}
\begin{enumerate}[label=(\roman*)]
	\item $\braket{\mathds{1}\otimes\widetilde{X}|\mathcal{A}}=X$ $\forall X\in\mathfrak{g}$
	\item $(\Phi_{\mathcal{S}})_g^*\mathcal{A}=\mathrm{Ad}_{g^{-1}}\circ\mathcal{A}$ $\forall g\in\mathcal{G}$.
\end{enumerate}
where $\widetilde{X}:=(\mathds{1}\otimes X_e)\circ\Phi^*\in\Gamma(\mathcal{P})$ is the fundamental vector field on $\mathcal{P}$ generated by $X\in \mathfrak{g}$ lifted to a ($\mathcal{S}$-relative) fundamental vector field $\mathds{1}\otimes\widetilde{X}$ on $\mathcal{P}_{/\mathcal{S}}$.
\end{definition}
\begin{remark}\label{remark:2.19}
Condition (ii) above for a $\mathcal{S}$-relative super connection $1$-form requires some explanation. It is a general fact in the $H^{\infty}$-category, given the $H^{\infty}$-smooth map $\Phi_{\mathcal{S}}:\,\mathcal{P}_{/\mathcal{S}}\times\mathcal{G}\rightarrow\mathcal{P}_{/\mathcal{S}}$ as well as a point $g\in\mathcal{G}$, the map
\begin{equation}
(\Phi_{\mathcal{S}})_g:=\Phi_{\mathcal{S}}(\cdot,g):\,\mathcal{P}_{/\mathcal{S}}\rightarrow\mathcal{P}_{/\mathcal{S}}
\label{eq:}
\end{equation}
in general, will not be of class $H^{\infty}$, unless $g\in\mathbf{B}(\mathcal{G})$.\footnote{Note that the set smooth functions on a $H^{\infty}$ supermanifold is a $\mathbb{R}$-vector space. For a general product supermanifold $\mathcal{M}\times\mathcal{N}$, one has $H^{\infty}(\mathcal{M}\times\mathcal{N})\cong H^{\infty}(\mathcal{M})\otimes H^{\infty}(\mathcal{N})$ . If then $f\otimes g$ is a smooth function on $\mathcal{M}\times\mathcal{N}$, it follows that $f\otimes g(\cdot,p)=f\cdot g(p)\in H^{\infty}(\mathcal{M})\Leftrightarrow g(p)\in\mathbb{R}\Leftrightarrow p\in\mathbf{B}(\mathcal{N})$ (cf. (\ref{eq:2.2.25.1})). In fact, this has its explanation in the algebraic category since, due to super Milnor's exercise, the real spectrum $\mathrm{Hom}(\mathcal{O}(\mathcal{M}),\mathbb{R})$ is given by the set of morphisms $\mathrm{ev}_p:\,\mathcal{O}(\mathcal{M})\rightarrow\mathbb{R}$ associated to points $p$ on the underlying topological space of an algebro-geometric supermanifold.} However, following \cite{Tuynman:2004} one can still associate a tangent map to $(\Phi_{\mathcal{S}})_g$ even if $g$ is not an element of the body. Therefore, one uses the identification $T(\mathcal{P}_{/\mathcal{S}}\times\mathcal{G})\cong T(\mathcal{P}_{/\mathcal{S}})\times T\mathcal{G}$ and set
\begin{equation}
(\Phi_{\mathcal{S}})_{g*}(X_{(s,p)}):=D_{((s,p),g)}\Phi_{\mathcal{S}}(X_{(s,p)},0_g)
\label{eq:}
\end{equation}
for any $X_{(s,p)}\in T_{(s,p)}(\mathcal{P}_{/\mathcal{S}})$. In \cite{Tuynman:2004}, this is called a \emph{generalized tangent map}. This can be extended to define a generalized pullback of forms w.r.t. $(\Phi_{\mathcal{S}})_g$ as used for instance in condition (ii).
\end{remark}
\begin{definition}
Let $\mathcal{G}\rightarrow\mathcal{P}_{/\mathcal{S}}\stackrel{\pi_{\mathcal{S}}}{\rightarrow}\mathcal{M}_{/\mathcal{S}}$ be a $\mathcal{S}$-relative principal super fiber bundle and $\gamma:\,\mathcal{S}\times\mathcal{I}\rightarrow\mathcal{M}$, with $\mathcal{I}\subseteq\Lambda^{1,0}$ a super interval, a smooth path on $\mathcal{M}_{\mathcal{S}}$. Given a $\mathcal{S}$-relative super connection $1$-form $\mathcal{A}$ on $\mathcal{P}_{/\mathcal{S}}$, a smooth path $\gamma^{hor}:\,\mathcal{S}\times\mathcal{I}\rightarrow\mathcal{P}$ on $\mathcal{P}_{/\mathcal{S}}$ is called a horizontal lift of $\gamma$ w.r.t. $\mathcal{A}$ if $\pi\circ\gamma^{hor}=\gamma$ and $\braket{(\mathds{1}\otimes\partial_t)\alpha_{\mathcal{S}}^{-1}(\gamma^{hor})(s,t)|\mathcal{A}}=0$ $\forall (s,t)\in\mathcal{S}\times\mathcal{I}$.  
\end{definition}
\begin{prop}\label{prop:9.7}
Let $\mathcal{A}\in\Omega^1(\mathcal{P}_{/\mathcal{S}},\mathfrak{g})$ be a $\mathcal{S}$-relative super connection $1$-form on the $\mathcal{S}$-relative principal super fiber bundle $\mathcal{G}\rightarrow\mathcal{P}_{/\mathcal{S}}\stackrel{\pi_{\mathcal{S}}}{\rightarrow}\mathcal{M}_{/\mathcal{S}}$ as well as $\gamma:\,\mathcal{S}\times\mathcal{I}\rightarrow\mathcal{M}$ a smooth path. Let furthermore $f:\,\mathcal{S}\rightarrow\mathcal{P}$ be a smooth map. Then, there exists a unique horizontal lift $\gamma^{hor}:\,\mathcal{S}\times\mathcal{I}\rightarrow\mathcal{M}$ of $\gamma$ w.r.t. $\mathcal{A}$ such that $\gamma^{hor}(\,\cdot\,,0)=f$.
\end{prop}
\begin{proof}[Sketch of proof.]
It suffices to assume that $\gamma$ is contained within a local trivialization neighborhood of $\mathcal{S}\times\mathcal{P}$, i.e., there exists an open subset $U\subseteq\mathcal{M}$ and a smooth morphism $\tilde{s}:\,U_{/\mathcal{S}}:=\mathcal{M}_{/\mathcal{S}}|_{\mathcal{S}\times U}\rightarrow\mathcal{P}_{/\mathcal{S}}$ of $\mathcal{S}$-relative supermanifolds such that $\pi_{\mathcal{S}}\circ\tilde{s}=\mathrm{id}_{U_{/\mathcal{S}}}$ and $\mathrm{im}\,\gamma\subseteq\pi^{-1}_{\mathcal{S}}(\mathcal{S}\times U)$. Furthermore, w.l.o.g., we can assume that $\tilde{s}(\,\cdot\,,\gamma(\,\cdot\,,0))=f$.\\
Set $\delta:=\tilde{s}\circ(\mathrm{id}\times\gamma):\,\mathcal{I}_{/\mathcal{S}}\rightarrow\mathcal{P}_{/\mathcal{S}}$. It follows that a horizontal lift has to be of the form $\gamma^{hor}:=\Phi_{\mathcal{S}}\circ(\delta\times g)$ for some smooth function $g:\,\mathcal{S}\times\tilde{\mathcal{I}}\rightarrow\mathcal{G}$ defined on some open subset $\tilde{\mathcal{I}}\subseteq\mathcal{I}$. Hence, this yields $\braket{(\mathds{1}\otimes\partial_t)\alpha_{\mathcal{S}}^{-1}(\gamma^{hor})(s,t)|\mathcal{A}}=0$ if and only if
\begin{equation}
(\mathds{1}\otimes\partial_t)g(s,t)=-R_{g(s,t)*}\braket{(\mathds{1}\otimes\partial_t)\delta(s,t)|\mathcal{A}}=:-R_{g(s,t)*}\mathcal{A}^{\gamma}(s,t)
\label{eq:9.9}
\end{equation}
where $\mathcal{A}^{\gamma}(s,t):=\braket{(\mathds{1}\otimes\partial_t)\delta(s,t)|\mathcal{A}}=\braket{(\mathds{1}\otimes\partial_t)\alpha_{\mathcal{S}}^{-1}(\gamma)(s,t)|\tilde{s}^*\mathcal{A}}$ and $R_g$ is the right translation on $\mathcal{G}$, with the initial condition $g(\,\cdot\,,0)=e$. Hence, the claim follows if one can show that (\ref{eq:9.9}) admits a $H^{\infty}$ smooth solution with $\tilde{\mathcal{I}}=\mathcal{I}$. This follows from the right-invariance of (\ref{eq:9.9}) (see \cite{Konsti-FB:2020} for more details). 
\end{proof}
\begin{remark}
For a smooth map $f:\,\mathcal{S}\rightarrow\mathcal{M}$, one can consider the pullback super fiber bundle 
\begin{equation}
f^*\mathcal{P}=\{(s,p)|\,f(s)=\pi(p)\}\subset\mathcal{S}\times\mathcal{P}
\label{eq:9.13}
\end{equation}
over $\mathcal{S}$. A smooth section $\tilde{\phi}:\mathcal{S}\rightarrow f^*\mathcal{P}$ of the pullback bundle is then of the form $\tilde{\phi}(s)=(s,\phi(s))$ $\forall s\in\mathcal{S}$ with $\phi:\,\mathcal{S}\rightarrow\mathcal{P}$ a smooth map satsifying $\pi\circ\phi=f$. Hence, we can identify
\begin{equation}
\Gamma(f^*\mathcal{P})=\{\phi:\,\mathcal{S}\rightarrow\mathcal{P}|\,\pi\circ\phi=f\}
\label{eq:9.14}
\end{equation}
\end{remark}
\begin{definition}\label{def:9.9}
Under the conditions of prop. \ref{prop:9.7}, the \emph{parallel transport in $\mathcal{P}_{/\mathcal{S}}$ along $\gamma$} w.r.t. the connection $\mathcal{A}$ is defined as
\begin{align}
\mathscr{P}^{\mathcal{A}}_{\mathcal{S},\gamma}:\,\Gamma(\gamma_0^*\mathcal{P})&\rightarrow\Gamma(\gamma_1^*\mathcal{P})\label{eq:9.15}\\
\phi&\mapsto\gamma^{hor}_{\phi}(\,\cdot\,,1)\nonumber
\end{align}
where, for $\Gamma(\gamma_0^*\mathcal{P})\ni\phi:\mathcal{S}\rightarrow\mathcal{P}$, $\gamma^{hor}_{\phi}$ is the unique horizontal lift of $\gamma$ with respect to $\mathcal{A}$ such that $\gamma^{hor}_{\phi}(\,\cdot\,,0)=\phi$.
\end{definition}
Given a morphism $\lambda:\,\mathcal{S}\rightarrow\mathcal{S}'$, this induces the pullback $\lambda^*:\,H^{\infty}(\mathcal{S}')\rightarrow H^{\infty}(\mathcal{S}),\,f\mapsto\lambda^*f=f\circ\lambda$ on the function sheaves. Since the super $H^{\infty}(\mathcal{S}\times\mathcal{M})$-module $\Gamma(\mathcal{M}_{/\mathcal{S}})$ of $\mathcal{S}$-relative vector fields on $\mathcal{M}_{/\mathcal{S}}$ is isomorphic to $H^{\infty}(\mathcal{S})\otimes\Gamma(\mathcal{M})$, this yields the morphism
\begin{align}
\lambda^*\equiv\lambda^*\otimes\mathds{1}:\,\Gamma(\mathcal{M}_{/\mathcal{S}'})&\rightarrow\Gamma(\mathcal{M}_{/\mathcal{S}})\label{eq:9.16}\\
f\otimes X&\mapsto\lambda^*f\otimes X\nonumber
\end{align}
Moreover, from $\Omega^1(\mathcal{M}_{/\mathcal{S}})\cong\Omega^1(\mathcal{M})\otimes H^{\infty}(\mathcal{S})$ we obtain the morphism
\begin{align}
\lambda^*\equiv\mathds{1}\otimes\lambda^*:\,\Omega^1(\mathcal{M}_{/\mathcal{S}'})&\rightarrow\Omega^1(\mathcal{M}_{/\mathcal{S}})\\
\omega\otimes f&\mapsto\omega\otimes\lambda^*f\nonumber
\label{eq:9.17}
\end{align}
By definition, it then follows 
\begin{equation}
\braket{\lambda^*X|\lambda^*\mathcal{A}}=\lambda^*\braket{X|\mathcal{A}}
\label{eq:9.18}
\end{equation}
In fact, since \ref{eq:9.18} is a local property, let us choose a local coordinate neighborhood such that $X$ and $\omega$ can be locally expanded in the form $X=f^i\otimes X_i$ and $\omega=\omega_j\otimes g^j$ with $X_i$ and $\omega_j$ smooth vector fields and 1-forms on $\mathcal{M}$, respectively. We then compute
\begin{align}
\braket{\lambda^*X|\lambda^*\mathcal{A}}&=\braket{\lambda^*f^i\otimes X_i|\omega_j\otimes\lambda^*g^j}=\lambda^*f^i\braket{X_i|\omega_j}\lambda^*g^j\nonumber\\
&=\lambda^*\braket{f^i\otimes X_i|\omega_j\otimes g^j}=\lambda^*\braket{X|\mathcal{A}}
\label{eq:9.19}
\end{align}
\begin{prop}\label{prop:9.10}
The parallel tansport map enjoys the following properties:
\begin{enumerate}[label=(\roman*)]
	\item $\mathscr{P}^{\mathcal{A}}_{\mathcal{S}}$ is functorial under compositions of paths, that is, for smooth paths $\gamma:\,\mathcal{S}\times\mathcal{I}\rightarrow\mathcal{M}$ and $\delta:\,\mathcal{S}\times\mathcal{I}\rightarrow\mathcal{M}$ on $\mathcal{M}_{/\mathcal{S}}$, one has
\begin{equation}
\mathscr{P}^{\mathcal{A}}_{\mathcal{S},\gamma\circ\delta}=\mathscr{P}^{\mathcal{A}}_{\mathcal{S},\gamma}\circ\mathscr{P}^{\mathcal{A}}_{\mathcal{S},\delta}
\label{eq:9.20}
\end{equation}
\item $\mathscr{P}^{\mathcal{A}}_{\mathcal{S},\gamma}$ is covariant under change of parametrization in the sense that if $\lambda:\,\mathcal{S}\rightarrow\mathcal{S}'$ is a morphism of supermanifolds, then the diagram
\begin{equation}
	 \xymatrix{
         \Gamma(f^*\mathcal{P})\ar[r]^{\mathscr{P}^{\mathcal{A}}_{\mathcal{S'},\gamma}}\ar[d]_{\lambda^*} &   \Gamma(g^*\mathcal{P}) \ar[d]^{\lambda^*}\\
            \Gamma((f\circ\lambda)^*\mathcal{P})\ar[r]^{\mathscr{P}^{\lambda^*\!\mathcal{A}}_{\mathcal{S},\lambda^*\!\gamma}} &   \Gamma((g\circ\lambda)^*\mathcal{P})  
     }
		\label{eq:9.21}
 \end{equation}
is commutative for any smooth path $\gamma:\,f\Rightarrow g$ on $\mathcal{M}_{/\mathcal{S}'}$.
\end{enumerate}
\end{prop}
\begin{proof}
As shown in the proof of Prop. \ref{prop:9.7}, the parallel transport map is locally defined via a first order differential equation. Hence, the functoriality property under the composition of paths is an immediate consequence of the uniqueness of solutions of differential equations once fixing the initial conditions. In fact, this implies $(\gamma\circ\delta)^{hor}=\gamma^{hor}\circ\delta^{hor}$ yielding (\ref{eq:9.20}) by definition \ref{eq:9.15}.\\
To prove the invariance under reparametrizations, notice that for a supermanifold morphism $\lambda:\,\mathcal{S}\rightarrow\mathcal{S}'$, one has $(\mathds{1}\otimes\partial_t)\lambda^*\gamma^{hor}=\lambda^*((\mathds{1}\otimes\partial_t)\gamma^{hor})$ so that, by definition \ref{eq:9.4}, it follows
\begin{equation}
(\mathds{1}\otimes\partial_t)\alpha_{\mathcal{S}}^{-1}(\lambda^*\gamma^{hor})=\lambda^*((\mathds{1}\otimes\partial_t)\alpha_{\mathcal{S}'}^{-1}(\gamma^{hor}))
\label{eq:9.22}
\end{equation}
and thus
\begin{align}
\braket{(\mathds{1}\otimes\partial_t)\alpha_{\mathcal{S}}^{-1}(\lambda^*\gamma^{hor})|\lambda^*\!\mathcal{A}}&=\braket{\lambda^*((\mathds{1}\otimes\partial_t)\alpha_{\mathcal{S}'}^{-1}(\gamma^{hor}))|\lambda^*\!\mathcal{A}}\nonumber\\
&=\braket{(\mathds{1}\otimes\partial_t)\alpha_{\mathcal{S}'}^{-1}(\gamma^{hor})|\mathcal{A}}=0
\label{eq:9.23}
\end{align} 
according to (\ref{eq:9.18}). Since $\lambda^*\gamma^{hor}_{\phi}(\,\cdot\,,0)=\phi\circ\lambda=\lambda^*\phi$ and $\pi\circ\lambda^*\gamma^{hor}_{\phi}=\lambda^*\gamma$, by uniqueness, this yields $\lambda^*\gamma^{hor}_{\phi}=(\lambda^*\gamma)^{hor}_{\lambda^*\phi}$ and therefore
\begin{equation}
\mathscr{P}^{\lambda^*\!\mathcal{A}}_{\mathcal{S},\lambda^*\!\gamma}(\lambda^*\phi)=(\lambda^*\gamma)^{hor}_{\lambda^*\phi}(\,\cdot\,,1)=\lambda^*\gamma^{hor}_{\phi}(\,\cdot\,,1)=\lambda^*(\mathscr{P}^{\mathcal{A}}_{\mathcal{S}',\gamma}(\phi))
\label{eq:9.24}
\end{equation}
$\forall\phi\in\Gamma(f^*\mathcal{P})$ proving the commutativity of the diagram (\ref{eq:9.21}).
\end{proof}
\begin{definition}
A \emph{global gauge transformation} $f$ on the $\mathcal{S}$-relative principal super fiber bundle $\mathcal{G}\rightarrow\mathcal{P}_{/\mathcal{S}}\rightarrow\mathcal{M}_{/\mathcal{S}}$ is a morphism $f:\,\mathcal{P}_{/\mathcal{S}}\rightarrow\mathcal{P}_{/\mathcal{S}}$ of $\mathcal{S}$-relative supermanifolds which is fiber-preserving and $\mathcal{G}$-equivariant, i.e., $\pi_{\mathcal{S}}\circ f=\pi_{\mathcal{S}}$ and $f\circ\Phi_{\mathcal{S}}=\Phi_{\mathcal{S}}\circ(f\times\mathrm{id})$. The set of global gauge transformations on $\mathcal{P}_{/\mathcal{S}}$ will be denoted by $\mathscr{G}(\mathcal{P}_{/\mathcal{S}})$.
\end{definition}
\begin{prop}
Their exists a bijective correspondence between the set $\mathscr{G}(\mathcal{P}_{/\mathcal{S}})$ of global gauge transformations on the $\mathcal{S}$-relative principal super fiber bundle $\mathcal{G}\rightarrow\mathcal{P}_{/\mathcal{S}}\rightarrow\mathcal{M}_{/\mathcal{S}}$ and the set 
\begin{equation}
H^{\infty}(\mathcal{S}\times\mathcal{P},\mathcal{G})^{\mathcal{G}}:=\{\sigma:\,\mathcal{S}\times\mathcal{P}\rightarrow\mathcal{G}|\,\sigma\circ\Phi_{\mathcal{S}}=\alpha\circ(\sigma\times\mathrm{id})\}
\label{eq:9.25}
\end{equation}
via
\begin{equation}
H^{\infty}(\mathcal{S}\times\mathcal{P},\mathcal{G})^{\mathcal{G}}\ni\sigma\mapsto\Phi_{\mathcal{S}}\circ(\mathrm{id}\times\sigma)\circ d_{\mathcal{S}\times\mathcal{P}}\in\mathscr{G}(\mathcal{P}_{/\mathcal{S}})
\label{eq:9.26}
\end{equation}
where $\alpha:\,\mathcal{G}\times\mathcal{G}\rightarrow\mathcal{G},\,(g,h)\mapsto h^{-1}gh$. In particular, global gauge transformations are super diffeomorphisms on $\mathcal{P}_{/\mathcal{S}}$ and $\mathscr{G}(\mathcal{P}_{/\mathcal{S}})$ forms an abstract group under composition of smooth maps. 
\end{prop}
\begin{prop}\label{Prop:2.26}
Let $\mathcal{A}\in\Omega^1(\mathcal{P}_{/\mathcal{S}},\mathfrak{g})$ be a $\mathcal{S}$-relative super connection $1$-form and $f\in\mathscr{G}(\mathcal{P}_{/\mathcal{S}})$ a global gauge transformation on the $\mathcal{S}$-relative principal super fiber bundle $\mathcal{G}\rightarrow\mathcal{P}_{/\mathcal{S}}\stackrel{\pi_{\mathcal{S}}}{\rightarrow}\mathcal{M}_{/\mathcal{S}}$. Then,
\begin{enumerate}[label=(\roman*)]
	\item $f^*\mathcal{A}\in\Omega^1(\mathcal{P}_{/\mathcal{S}},\mathfrak{g})$ is a $\mathcal{S}$-relative connection 1-form and, in particular,
\begin{equation}
f^*\mathcal{A}=\mathrm{Ad}_{\sigma_f^{-1}}\circ\mathcal{A}+\sigma_f^*\theta_{\mathrm{MC}}
\label{eq:9.29}
\end{equation}
\item the diagram
\begin{displaymath}
	 \xymatrix{
         \Gamma(g^*\mathcal{P})\ar[r]^{\mathscr{P}^{\mathcal{A}}_{\mathcal{S},\gamma}}\ar[d]_{\alpha_{\mathcal{S}}\circ f\circ\alpha_{\mathcal{S}}^{-1}} &   \Gamma(h^*\mathcal{P}) \ar[d]^{\alpha_{\mathcal{S}}\circ f\circ\alpha_{\mathcal{S}}^{-1}}\\
            \Gamma(g^*\mathcal{P})\ar[r]^{\mathscr{P}^{f^*\mathcal{A}}_{\mathcal{S},\gamma}} &   \Gamma(h^*\mathcal{P})  
     }
		\label{eq:9.30}
 \end{displaymath}
is commutative for any smooth path $\gamma:\,g\Rightarrow h$ on $\mathcal{M}_{/\mathcal{S}}$.
\end{enumerate}
\end{prop}
We want to give an explicit local expression of the parallel transport map making it more accessible for applications in physics, in particular, in section \ref{OSp}. Therefore, we assume that $\mathcal{G}$ is a super matrix Lie group, i.e., an embedded super Lie subgroup of the general linear group $\mathrm{GL}(\mathcal{V})$ of a super $\Lambda$-vector space $\mathcal{V}=V\otimes\Lambda$. Let $\gamma:\,\mathcal{S}\times\mathcal{I}\rightarrow\mathcal{M}$ be a smooth path which is contained within a local trivialization neighborhood of $\mathcal{P}_{/\mathcal{S}}$ and $\tilde{s}:\,U_{/\mathcal{S}}:=\mathcal{M}_{/\mathcal{S}}|_{\mathcal{S}\times U}\rightarrow\mathcal{P}_{/\mathcal{S}}$ the corresponding smooth section. Then, equation (\ref{eq:9.9}) in the proof of prop. \ref{prop:9.7} reads
\begin{equation}
(\mathds{1}\otimes\partial_t)g(s,t)=-\mathcal{A}^{\gamma}(s,t)\cdot g(s,t)
\label{eq:9.38}
\end{equation}
with $\mathcal{A}^{\gamma}(s,t):=\braket{(\mathds{1}\otimes\partial_t)\alpha_{\mathcal{S}}^{-1}(\gamma)(s,t)|\tilde{s}^*\mathcal{A}}$. Furthermore, suppose that $U$ defines a local coordinate neighborhood of $\mathcal{M}$. The 1-form $\tilde{s}^*\mathcal{A}$ on $\mathcal{S}\times U$ can then be expanded as
\begin{equation}
\tilde{s}^*\mathcal{A}=\mathrm{d}x^{\mu}\mathcal{A}^{(\tilde{s})}_{\mu}+\mathrm{d}\theta^{\alpha}\mathcal{A}^{(\tilde{s})}_{\alpha}
\label{eq:9.39}
\end{equation} 
with smooth even and odd functions $\mathcal{A}^{(\tilde{s})}_{\mu}$ and $\mathcal{A}^{(\tilde{s})}_{\alpha}$ on $\mathcal{S}\times U$, respectively. This yields 
\begin{equation}
\mathcal{A}^{\gamma}(s,t)=:\dot{x}^{\mu}\mathcal{A}^{(\tilde{s})}_{\mu}(s,t)+\dot{\theta}^{\alpha}\mathcal{A}^{(\tilde{s})}_{\alpha}(s,t)
\label{eq:9.40}
\end{equation}
Hence, the solution of equation (\ref{eq:9.38}) with the initial condition $g(\,\cdot\,,0)=\mathds{1}$ takes the form
\begin{equation}
g(s,t)=\mathcal{P}\exp\left(-\int_{0}^t{\mathrm{d}t'\,\dot{x}^{\mu}\mathcal{A}^{(\tilde{s})}_{\mu}(s,t')+\dot{\theta}^{\alpha}\mathcal{A}^{(\tilde{s})}_{\alpha}(s,t')}\right)
\label{eq:9.41}
\end{equation}
This is the most general local expression of the parallel transport map corresponding to a $\mathcal{S}$-relative super connection 1-form. This form is used for instance in \cite{Mason:2010yk} in the discussion about the relation between super twistor theory and $\mathcal{N}=4$ super Yang-Mills theory (see also \cite{Groeger:2013aja}). Note that in case $\mathcal{S}=\{*\}$ is a single point, the odd coefficients in (\ref{eq:9.41}) become zero so that this expression just reduces the parallel transport map of an ordinary connection 1-form on a principal bundle.\\
\\
Finally, let us restrict to a subclass of smooth paths on $\mathcal{M}_{/\mathcal{S}}$ obtained via the lift of smooth paths $\gamma:\,\mathcal{I}\rightarrow\mathcal{M}$ on the bosonic sub supermanifold\footnote{for a $H^{\infty}$ supermanifold $\mathcal{M}$, the bosonic sub supermanifold $\mathcal{M}_0$ is defined as the split supermanifold $\mathbf{S}(\mathbf{B}(\mathcal{M}))$ (in \cite{Tuynman:2004,Tuynman:2018} this is also called the $\mathbf{G}$-extension of $\mathbf{B}(\mathcal{M})$ as this can be viewed as a generalization of the ordinary $\mathbf{G}$-extension of functions (cf. (\ref{eq:2.2.23.1})))} $\mathcal{M}_0$ of $\mathcal{M}$. A $\mathcal{S}$-relative connection 1-form $\mathcal{A}\in\Omega^1(\mathcal{P}_{/\mathcal{S}},\mathfrak{g})$ induces via pullback along the inclusion $\iota:\,\mathcal{S}\times\mathcal{M}_0\hookrightarrow\mathcal{S}\times\mathcal{M}$ a $\mathcal{S}$-relative super connection 1-form $\iota^*\mathcal{A}$ on the pullback bundle $\mathcal{G}\rightarrow\iota^*\mathcal{P}_{/\mathcal{S}}\rightarrow(\mathcal{M}_0)_{/\mathcal{S}}$. Let 
\begin{equation}
\iota^*\mathcal{A}=\mathrm{pr}_{\mathfrak{g}_0}\circ\iota^*\mathcal{A}+\mathrm{pr}_{\mathfrak{g}_1}\circ\iota^*\mathcal{A}=:\omega+\psi
\label{eq:9.43}
\end{equation}
be the decomposition of $\iota^*\mathcal{A}$ according to the even and odd part of the super Lie algebra $\mathfrak{g}=\mathfrak{g}_0\oplus\mathfrak{g}_1$. Since $\omega\in\Omega^1(\iota^*\mathcal{P}_{/\mathcal{S}},\mathfrak{g}_0)_0\cong\Omega^1(\mathcal{P}|_{\mathcal{M}_0},\mathfrak{g}_0)_0\otimes H^{\infty}(\mathcal{S})_0$, it follows that $\omega$ can be reduced to a $\mathcal{S}$-relative super connection 1-form on the $\mathcal{S}$-relative principal super fiber bundle $\mathcal{G}_0\rightarrow(\mathcal{P}_0)_{/\mathcal{S}}\rightarrow(\mathcal{M}_0)_{/\mathcal{S}}$ which will be denoted by the same symbol. Hence, $\omega$ gives rise to a parallel transport map $\mathscr{P}^{\omega}_{\mathcal{S},\gamma}$ along $\alpha_{\mathcal{S}}(\mathrm{id}\times\gamma):\,\mathcal{S}\times\mathcal{I}\rightarrow\mathcal{M}_{0}$.\\
Suppose that $\gamma$ is contained within a local tirvialization neighborhood of $\mathcal{P}_0$ and let  $\tilde{s}:\,U_{/\mathcal{S}}\rightarrow(\mathcal{P}_0)_{/\mathcal{S}}$ be the corresponding local section with $U\subset\mathcal{M}_0$ open. Let $g[\mathcal{A}]:\,\mathcal{S}\times\mathcal{I}\rightarrow\mathcal{G}$ be the solution of the parallel transport equation (\ref{eq:9.9}) of $\mathcal{A}$
\begin{equation}
\partial_tg[\mathcal{A}](s,t)=-R_{g[\mathcal{A}](s,t)*}\mathcal{A}^{\gamma}(s,t)
\label{eq:9.44}
\end{equation}
with the initial condition $g(\,\cdot\,,0)=e$, where 
\begin{equation}
\mathcal{A}^{\gamma}:=\braket{\mathds{1}\otimes\partial_t\gamma|\tilde{s}^*\mathcal{A}}=\braket{\mathds{1}\otimes\partial_t\gamma|\tilde{s}^*\omega}+\braket{\mathds{1}\otimes\partial_t\gamma|\tilde{s}^*\psi}=:\omega^{\gamma}+\psi^{\gamma}
\label{eq:9.45}
\end{equation}
Furthermore, let $g[\omega]:\,\mathcal{S}\times\mathcal{I}\rightarrow\mathcal{G}_0$ be the solution of the corresponding parallel transport equation of $\omega$. Set $g[\psi]:=g[\omega]^{-1}\cdot g[\mathcal{A}]:\,\mathcal{S}\times\mathcal{I}\rightarrow\mathcal{G}$. Using $\partial_t(g[\omega]^{-1})=-L_{g[\omega]^{-1*}}R_{g[\omega]^{-1*}}(\partial_tg[\omega])=L_{g[\omega]^{-1*}}\omega^{\gamma}$, it then follows
\begin{align}
\partial_t g[\psi]&=D\mu_{\mathcal{G}}(\partial_t(g[\omega]^{-1}),\partial_tg[\mathcal{A}])=R_{g[\mathcal{A}]*}L_{g[\omega]^{-1*}}\omega^{\gamma}-L_{g[\omega]^{-1*}}R_{g[\mathcal{A}]*}\mathcal{A}^{\gamma}\nonumber\\
&=-R_{g[\mathcal{A}]*}L_{g[\omega]^{-1*}}\psi^{\gamma}=-R_{g[\psi]*}R_{g[\omega]*}L_{g[\omega]^{-1*}}\psi^{\gamma}=-R_{g[\psi]*}\mathrm{Ad}_{g[\omega]^{-1}}(\psi^{\gamma})
\label{eq:9.46}
\end{align}
that is, $g[\psi]$ is the solution of the equation
\begin{equation}
\partial_t g[\psi]=-R_{g[\psi]*}\mathrm{Ad}_{g[\omega]^{-1}}(\psi^{\gamma})
\label{eq:9.47}
\end{equation} 
For a super matrix Lie group $\mathcal{G}$, the solution of (\ref{eq:9.47}) can be explicitly written as
\begin{equation}
g[\psi](s,t)=\mathcal{P}\mathrm{exp}\left(-\int_0^t{\mathrm{d}\tau\,(\mathrm{Ad}_{g[\omega]^{-1}}\psi^{\gamma})(s,\tau)}\right)
\label{eq:9.48}
\end{equation}
Hence, for instance, if $\gamma$ is closed loop on $\mathcal{M}_0$, in this gauge the Wilson loop takes the form
\begin{equation}
W_{\gamma}[\mathcal{A}]=\mathrm{str}\left(g_{\gamma}[\omega]\cdot\mathcal{P}\mathrm{exp}\left(-\oint_{\gamma}{\mathrm{Ad}_{g_{\gamma}[\omega]^{-1}}\psi^{(\tilde{s})}}\right)\right):\,\mathcal{S}\rightarrow\mathcal{G}
\label{eq:9.49}
\end{equation}
where $\psi^{(\tilde{s})}:=\tilde{s}^*\psi$. It follows from \ref{Prop:2.26} that $W_{\gamma}[\mathcal{A}]$ is invariant under local gauge transformations. In fact, $g_{\gamma}[\mathcal{A}]$ transforms as
\begin{equation}
g_{\gamma}[\mathcal{A}](s)\rightarrow\phi(s)\cdot g_{\gamma}[\mathcal{A}](s)\cdot\phi(s)^{-1},\quad\forall s\in\mathcal{S}
\label{eq:}
\end{equation}
for some smooth function $\phi:\,\mathcal{S}\rightarrow\mathcal{G}$. Hence, due to cyclicity of the supertrace, (\ref{eq:9.49}) is indeed invariant. Finally, by prop. \ref{prop:9.10} (ii) the Wilson loop is also invariant under change of parametrisations. That is, if $\lambda:\,\mathcal{S}'\rightarrow\mathcal{S}$ is a supermanifold morphism, then
\begin{equation}
\lambda^*W_{\gamma}[\mathcal{A}]=W_{\gamma}[\lambda^*\mathcal{A}]:\,\mathcal{S}'\rightarrow\mathcal{G}
\label{eq:9.50}
\end{equation}
Thus, due these properties, $W_{\gamma}[\mathcal{A}]$ can be regarded as a fundamental physical quantity according to \cite{Schmitt:1996hp}.

\begin{remark}\label{remark:2.39}
Suppose $\mathcal{S}$ is a superpoint. Choosing a basis $Q_{\alpha}\in\mathfrak{g}_1$ of the odd part of the super Lie algebra $\mathfrak{g}$, it follows that the pullback of the fermionic components $\psi^{\alpha}$ of the $\mathcal{S}$-relative super connection form to the bosonic sub supermanifold $\mathcal{M}_0$ as defined via (\ref{eq:9.43}) are elements of the set
\begin{equation}
\Omega^1(\mathcal{M}_0)\otimes H^{\infty}(\mathcal{S})_1\cong\Omega^1(\mathbf{B}(\mathcal{M}))\otimes\left(\Exterior\,\mathbb{R}^{N}\right)_1^*
\label{eq:}
\end{equation}
with $N$ equal to the odd dimension of $\mathcal{S}$. Hence, for any smooth vector field $X$ on $\mathbf{B}(\mathcal{M})$, it follows that $\braket{X|\psi^{\alpha}}$ can be identified with an odd functional on $\Exterior\mathbb{R}^N$. This is similar to the approach in pAQFT, where fermionic fields are described in terms of odd functionals on configuration space \cite{Rejzner:2011au,Rejzner:2016hdj}. This suggests to set $\mathcal{S}:=\Omega^1_{hor}(\mathbf{B}(\mathcal{P}),\mathfrak{g}_1)\otimes\Lambda$ regarded as a purely odd infinite dimensional super vector space. This, however, would require a generalization of this present formalism to infinite dimensional (locally convex) supermanifolds (see \cite{Schuett:2018} for recent advances in this direction).
\end{remark} 
By definition, $g[\mathcal{A}]$ defines a smooth map $g[\mathcal{A}]:\,\mathcal{S}\rightarrow\mathcal{G}$ from the parametrizing supermanifold $\mathcal{S}$ to the gauge group $\mathcal{G}$. Since, one has an equivalence of categories between algebro-geometric and $H^{\infty}$ supermanifolds, it follows that
\begin{equation}
H^{\infty}(\mathcal{S},\mathcal{G})\cong\mathrm{Hom}_{\mathbf{SMan}_{\mathrm{Alg}}}(\mathbf{A}(\mathcal{S}),\mathbf{A}(\mathcal{G}))
\label{eq:}
\end{equation}
Thus, it follows that $g[\mathcal{A}]$ can be identified with a $\mathbf{A}(\mathcal{S})$-point of $\mathbf{A}(\mathcal{G})$ according to definition \ref{def:2.3}. This coincides with the results of \cite{Florin:2008} and \cite{Groeger:2013aja} where the parallel transport of super connections on super vector bundles in the pure algebraic setting has been considered. It was found that the parallel transport map has the interpretation in terms of $\mathcal{T}$-points of a general linear group.\\     
\\   
Let us finally comment on the choice of the parametrizing supermanifold $\mathcal{S}$. Working in the algebraic category, a typical choice for $\mathcal{S}$ is a superpoint $\mathcal{S}=(\{*\},\Lambda_N)$ with $\Lambda_N$ a Grassmann algebra with $N$ fermionic generators. In this case, it follows $g[\mathcal{A}]\in\mathcal{G}(\Lambda_N)$, that is, $g[\mathcal{A}]$ is a $\Lambda_N$-point of $\mathcal{G}$. But this means, if $N$ is large enough (i.e., larger than the odd dimension of $\mathcal{G}$), we again end up with a Rogers-DeWitt supermanifold $\mathcal{G}(\Lambda_N)$ and $g[\mathcal{A}]$ can be identfied as elements of the group $\mathcal{G}(\Lambda_N)$. This, once more, reflects the strong link between these two approaches to supermanifold theory.

\section{Gravity as Cartan geometry}\label{Cartan geometry}
In this section, we want to review the interpretation of gravity as Cartan geometry as this will serve a starting point for a very elegant approach to supergravity as decribed in section \ref{superCartan geometry} and a derivation of a super analog of Asthekar's connection in section \ref{Ashtekar}. A detailed account on the relation between Cartan geometry and general relativity can be found for instance in \cite{Wise:2006sm}.\\
\\
In his famous Erlangen program, F. Klein studied the idea of classifying the geometry of space via the underlying group of symmetries. For instance, one can consider Minkowski spacetime $\mathbb{M}=(\mathbb{R}^{1,3},\eta)$ and study the corresponding Lie group $\mathrm{ISO}(\mathbb{R}^{1,3})$ of isometries which is isomorphic to the Poincaré group $\mathbb{R}^{1,3}\rtimes\mathrm{SO}^+(1,3)$. If one then chooses a specific event $p\in\mathbb{M}$, one can consider the corresponding stabilizer subgroup $\mathrm{SO}^+(1,3)$ which preserves that point. Since the isometry group acts transitively on $\mathbb{M}$, it follows that Minkowski spacetime can be described in terms of the coset space
\begin{equation}
\mathbb{M}\cong\mathrm{ISO}(\mathbb{R}^{1,3})/\mathrm{SO}^+(1,3)
\label{eq:3.1.1}
\end{equation}
Hence, the collection of spacetime events can equivalently be described in terms of the underling symmetry groups. A similar kind of reasoning applies in case of the other maximally symmetric homogeneous spacetimes such as de Sitter of anti-de Sitter spacetime playing a central role in general relativity and cosmology. Hence, one makes the following definition. 
\begin{definition}
A \emph{Klein geometry} is a pair $(G,H)$ consisting of a Lie Group $G$ and an embedded Lie subgroup $H\hookrightarrow G$ such that $G/H$ is connected.
\end{definition}
Given a Klein geometry $(G,H)$, the coset space $G/H$ has the structure of principal $H$-bundle
\begin{displaymath}
	 \xymatrix{
         G\ar[d]_{\pi}  &    H \ar[l] \\
            G/H &     
     }
		\label{eq:3.1.2}
 \end{displaymath}
Moreover, on $G$ there exists a canonical $\mathfrak{g}$-valued 1-form $\theta_{\mathrm{MC}}\in\Omega^1(G,\mathfrak{g})$, called \emph{Maurer-Cartan form} which, choosing a basis of left-invariant vector fields $X_i\in\mathfrak{g}$, $i=1,\ldots,\mathrm{dim}\,\mathfrak{g}$, is defined as
\begin{equation}
\theta_{\mathrm{MC}}=X_i\otimes\omega^i
\label{eq:3.1.3}
\end{equation} 
where $\omega^i\in\Omega^1(G)$ is the corresponding dual basis of left-invariant one-forms on $G$ satisfying $\omega^i(X_j)=\delta^i_j$. It follows by definition that the Maurer-Cartan form is $G$-equivariant, i.e.\footnote{this can be seen directly using the equivalent definition in terms of the left-translation $\theta_{\mathrm{MC}}(X_g)=L_{g^{-1}*}X_g$.} 
\begin{equation}
R^*_g\theta_{\mathrm{MC}}=\mathrm{Ad}_{g^{-1}}\circ\theta_{\mathrm{MC}}
\label{eq:3.1.4}
\end{equation}
$\forall g\in G$ with $R_g:\,G\rightarrow G$ denoting the right translation on $G$. By definition, $\theta_{\mathrm{MC}}$ maps left-invariant vector fields to themselves, i.e., $\theta_{\mathrm{MC}}(X)=X_e$ $\forall X\in\mathfrak{g}$ and, as a consequence, yields and isomorphism $\theta_{\mathrm{MC}}:\,T_gG\rightarrow\mathfrak{g}$ of vector spaces at any $g\in G$. Moreover, it satisfies the \emph{Maurer-Cartan equation}
\begin{equation}
\mathrm{d}\theta_{\mathrm{MC}}+[\theta_{\mathrm{MC}}\wedge\theta_{\mathrm{MC}}]=0
\label{eq:3.1.5}
\end{equation}
As seen above, standard examples of Klein geometries $(G,H)$ arising in physics are given by the Minkowski spacetime $(\mathrm{ISO}(\mathbb{R}^{1,3}),\mathrm{SO}^+(1,3))$, de Sitter $(\mathrm{SO}(1,4),\mathrm{SO}^+(1,3))$ and anti-de Sitter spacetime $(\mathrm{SO}(2,3),\mathrm{SO}^+(1,3))$. These have in common that the Lie algebra $\mathfrak{g}$ of $G$ can be split into $\mathrm{Ad}(H)$-invariant subspaces $\mathfrak{g}=\mathfrak{h}\oplus\mathfrak{g}/\mathfrak{h}$ with $\mathfrak{h}$ the Lie algebra of $H$. Moreover, on the moduli space $\mathfrak{g}/\mathfrak{h}$ there exists a canonical $\mathrm{Ad}(H)$-invariant bilinear form. In this case, the Klein geometry is called \emph{metric} and \emph{reductive}. \\
Hence, we see that flat spacetime can be equivalently be described in terms of Klein geometry. Based on this observation, Cartan formulated a theory now known as Cartan geometry which can be interpreted as a deformed Klein geometry such as gravity is a deformed version of flat Minkowski spacetime.   
\begin{definition}
A \emph{metric reductive Cartan geometry} $(\pi:\,P\rightarrow M,A)$ modeled on a metric reductive Klein geometry $(H,G;\eta)$ is a principal fiber bundle $H\rightarrow P\rightarrow M$ with structure group $H$ together with a $\mathfrak{g}$-valued 1-form $A\in\Omega^1(P,\mathfrak{g})$ on $P$ called \emph{Cartan connection} such that
\begin{enumerate}[label=(\roman*)]
	\item $A_p(X_p)=X$ $\forall X\in\mathfrak{h}=T_e\mathcal{H}$, $p\in\mathcal{P}$
	\item $\Phi_h^*A=\mathrm{Ad}_{h^{-1}}\circ A$ $\forall h\in\mathcal{H}$
	\item the map $A_p:\,T_pP\rightarrow\mathfrak{g}$ defines an isomorphism of vector spaces for any $p\in P$.
\end{enumerate}
where the last condition is also called the \emph{Cartan condition}.
\end{definition}
Given a metric reductive Cartan geometry $(\pi:\,P\rightarrow M,A)$, one can split the Cartan connection $A$ by projecting it according to the decomposition $\mathfrak{g}=\mathfrak{h}\oplus\mathfrak{g}/\mathfrak{h}$ of the Lie algebra of $G$ yielding
\begin{equation}
A=\mathrm{pr}_{\mathfrak{h}}\circ A+\mathrm{pr}_{\mathfrak{g}/\mathfrak{h}}\circ A=:\omega+\theta
\label{eq:3.1.6}
\end{equation}
with 1-forms $\omega\in\Omega^1(P,\mathfrak{h})$ and $\theta\in\Omega^1(P,\mathfrak{g/\mathfrak{h}})$.  Due to the conditions (i) and (ii) of the Cartan connections, it follows immediately that $\omega$ defines an ordinary principal connection 1-form in the sense of Ehresmann.\\
Let $\mathscr{H}:=\mathrm{ker}(\omega)$ be the induced horizontal distribution of the tangent bundle $TP$. If $\tilde{X}\in\mathscr{V}$ is a vertical vector field with $X\in\mathfrak{h}$ one has $A(\tilde{X})=X=\omega(X)$ and thus $\theta(X)=0$. Hence, since $\mathfrak{g}/\mathfrak{h}$ defines a $H$-invariant subspace, together with property (ii), this immediately implies that the soldering form is horizontal of type $(H,\mathrm{Ad})$, i.e. $\theta\in\Omega^1_{hor}(P,\mathfrak{g}/\mathfrak{h})^{(H,\mathrm{Ad})}$. In fact, $\theta$ even provides 
an identification of the principal bundle $P$ as a $H$-reduction of the frame bundle $\mathscr{F}(M)$ explaining its name.\\
Therefore, note that $\theta_p^{-1}:=(\theta_p|_{\mathscr{H}_p})^{-1}$ for any $p\in P$ defines an isomorphism on $\mathfrak{g}/\mathfrak{h}$ and $\omega_p|_{\mathscr{H}_p}$ is an isomorphism onto $T_{\pi(p)}M$ so that $D_p\pi\circ\theta_p^{-1}:\,\mathfrak{g}/\mathfrak{h}\stackrel{\sim}{\rightarrow}T_{\pi(p)}M$ is a linear frame at $\pi(p)$. Hence, this yields a map
\begin{align}
\iota:\,P\rightarrow\mathscr{F}(M),\,p\mapsto D_p\pi\circ\theta_p^{-1}
\label{eq:3.1.6.1}
\end{align}
By property (ii), we have $\Phi_h^*\theta_p(Y_p)=\theta_{ph}(D_p\Phi_{h}(Y_p))=\mathrm{Ad}_{h^{-1}}(\theta_p(Y_p))$ $\forall Y_p\in T_pP$ and $h\in H$ and therefore
\begin{equation}
\theta^{-1}_{ph}=D_p\Phi_{h}\circ\theta_p^{-1}\circ\mathrm{Ad}_{g}
\label{eq:3.1.7}
\end{equation}
from which we obtain
\begin{equation}
\iota(p\cdot h)=D_{ph}\pi\circ\theta_{ph}^{-1}=D_{ph}\pi\circ D_p\Phi_{h}\circ\theta_p^{-1}\circ\mathrm{Ad}_h=\iota(p)\circ\mathrm{Ad}_h
\label{eq:3.1.8}
\end{equation}
$\forall p\in P,\,h\in H$. That is, $\iota:\,P\rightarrow\mathscr{F}(M)$ is $H$-equivariant and fiber-preserving so that $P$ defines a $H$-reduction of the frame bundle w.r.t. the group morphism $\mathrm{Ad}:\,H\rightarrow\mathrm{GL}(\mathfrak{g}/\mathfrak{h})$.\\
Moreover, it follows that $\iota$ induces a an isomorphism (denoting by the same symbol)
\begin{align}
\iota:\,P\times_{\mathrm{Ad}}\mathfrak{g}/\mathfrak{h}&\stackrel{\sim}{\longrightarrow}TM\label{eq:3.1.9}\\
[(p,X)]&\longmapsto D_p\pi(\theta_p^{-1}(X))
\end{align}
between the associated vector bundle $P\times_{\mathrm{Ad}}\mathfrak{g}/\mathfrak{h}$ and the tangent bundle of $M$.\\
\\
To a Cartan connection $A$ one associates the \emph{Cartan curvature} $F(A)\in\Omega^2(P,\mathfrak{g})$ according to
\begin{equation}
F(A):=\mathrm{d}A+\frac{1}{2}[A\wedge A]
\label{eq:3.1.10}
\end{equation}
In case of a 'flat' Klein geometry, the Cartan connection is given by the Maurer-Cartan form (\ref{eq:3.1.3}) which satisfies the structure equation (\ref{eq:3.1.5}), i.e., the Cartan curvature is identically zero. Thus $F(A)$ indicates the deviation of a Cartan geometry from a flat Klein geometry. In fact, one can prove that a Cartan geometry is isomorphic to the homogeneous model $(G\rightarrow G/H,\theta_{\mathrm{MC}})$ if and only if the corresponding Cartan curvature vanishes.\\
Decomposing $F(A)$ according to the decomposition $\mathfrak{g}=\mathfrak{h}\oplus\mathfrak{g}/\mathfrak{h}$ of the Lie algebra, one obtains
\begin{align}
F(A)=\mathrm{pr}_{\mathfrak{h}}\circ F(A)+\mathrm{pr}_{\mathfrak{g}/\mathfrak{h}}\circ F(A)=F(\omega)+\Theta^{(\omega)}+\frac{1}{2}[\theta\wedge\theta]
\label{eq:3.1.11}
\end{align}
where
\begin{equation}
F(\omega)=D^{(\omega)}\omega=\mathrm{d}\omega+\frac{1}{2}[\omega\wedge\omega]
\label{eq:3.1.12}
\end{equation}
is the curvature of the connection 1-form $\omega$ and
\begin{equation}
\Theta^{(\omega)}:=D^{(\omega)}\theta=\mathrm{d}\theta+[\omega\wedge\theta]
\label{eq:3.1.13}
\end{equation}
is the corresponding \emph{torsion 2-form}. To see that in fact encodes the torsion of the connection, note that $\omega$ indues a connection on the associated vector bundle $P\times_{\mathrm{Ad}}\mathfrak{g}/\mathfrak{h}$ and thus, via (\ref{eq:3.1.9}), a connection $\nabla:\,\Gamma(TM)\rightarrow\Gamma(TM)$ on the tangent bundle. For vector fields $X,Y\in\Gamma(TM)$, it is given by 
\begin{equation}
(\nabla_XY)_x=\iota([p,X^{hor}\theta(Y^{hor})])
\label{eq:3.1.14}
\end{equation}
for any $x\in M$ and $p\in P$ with $\pi(p)=x$ where $X^{hor},Y^{hor}$ denote the horizontal lifts of $X$ and $Y$, respectively.\footnote{for a vector field $X\in\Gamma(TM)$, the corresponding horizontal lift $X^{hor}\in\Gamma(TP)$ is a vector field on $P$ which is horizontal, i.e., $X^{hor}_p\in\mathscr{H}_p$ $\forall p\in P$, and satisfyies $D_p\pi(X^{hor})=X_{\pi(p)}$}. Moreover, in general, given a representation $\rho:\,H\rightarrow\mathrm{GL}(V)$ of $H$ on a vector space $V$, there exists an isomorphism 
\begin{equation}
\Omega^{k}_{hor}(P,V)^{(H,\rho)}\stackrel{\sim}{\rightarrow}\Omega^k(M,P\times_{\rho}V),\,\omega\mapsto\overline{\omega}
\label{eq:3.1.15}
\end{equation}
between $V$-valued $k$-forms of type $(H,\rho)$ and $k$-forms with values in the associated bundle $P\times_{\rho}V$. Hence, we can associate to $\Theta^{(\omega)}$ a 2-form $\overline{\Theta^{(\omega)}}\in\Omega^2(M,P\times_{\mathrm{Ad}}\mathfrak{g}/\mathfrak{h})$, which, applying (\ref{eq:3.1.9}), yields another form $\iota\circ\overline{\Theta^{(\omega)}}\in\Omega^2(M)\otimes\Gamma(TM)$. For vector fields $X,Y\in\Gamma(TM)$, we then compute
\begin{align}
\iota\circ\overline{\Theta^{(\omega)}}(X_x,Y_x)&=\iota\circ[p,\mathrm{d}\theta(X^{hor},Y^{hor})]\nonumber\\
&=\iota\circ[p,X^{hor}\theta(Y^{hor})-Y^{hor}\theta(X^{hor})-\theta([X^{hor},Y^{hor}])]\nonumber\\
&=\nabla_XY-\nabla_YX-D_p\pi([X,Y]^{hor})=T^{\nabla}(X_x,Y_x)
\label{eq:3.1.16}
\end{align}
for any $x\in M$ and $p\in P$ such that $\pi(p)=x$. Hence, $\Theta^{(\omega)}$ indeed encodes the torsion of the associated affine connection $\nabla$ on the tangent bundle of $M$.\\
\\
With all these observations, let us now make contact to general relativity. As seen already at the beginning, flat Minkowski spacetime can be described in terms of the metric Klein geometry $(\mathrm{ISO}(\mathbb{R}^{1,3}),\mathrm{SO}^+(1,3);\eta)$. Hence, we consider gravity as a metric reductive Cartan geometry $(P\rightarrow M,A;\eta)$ modeled on the metric reductive Klein geometry $(\mathrm{ISO}(\mathbb{R}^{1,3}),\mathrm{SO}^+(1,3);\eta)$ where $\mathrm{SO}^+(1,3)\rightarrow P\stackrel{\pi}{\rightarrow} M$ is a principal bundle with structure group $\mathrm{SO}^+(1,3)$ and $A\in\Omega^1(P,\mathfrak{iso}(\mathbb{R}^{1,3}))$ is a Cartan connection.\\
By \ref{eq:3.1.6.1}, we know that $P$ defines a $\mathrm{SO}^+(1,3)$-reduction of the frame bundle $\mathscr{F}(M)$ of $M$. As such, it induces a Lorentzian metric $g\in\Gamma(T^*M\otimes T^*M)$ on $M$ which, for vector fields $X,Y\in\Gamma(TM)$, is defined as
\begin{equation}
g(X_x,Y_x):=\eta(\iota^{-1}_p(X_x),\iota_p^{-1}(Y_x))
\label{eq:}
\end{equation}
$\forall p\in P_x,\,x\in M$. Note that $g$ is in fact well-defined, i.e., independent on the choice of $p\in P_x$, since $\iota$ is equivariant and $\eta$, by definition, is a bilinear form invariant under the Adjoint representation of $\mathrm{SO}^+(1,3)$ on $\mathbb{R}^{1,3}$. Hence, $M$ is in fact a Lorentzian manifold and $P$ can be identified with the bundle $\mathscr{F}_{\mathrm{SO}}(M)$ of Lorentz frames on $M$.\\
Let $s:\,M\supset U\rightarrow P$ be a local section of $P$ and $e^I\in\Omega^1(U)$ be defined via $s^*\theta=:e^IP_I$. Then
\begin{align}
g_{\mu\nu}&=\eta(s^*\theta_{\mu},s^*\theta_{\nu})=e^I_{\mu}e^J_{\nu}\eta_{IJ}
\label{eq:}
\end{align}
i.e. $(e^I)$ defines a local coframe with the corresponding frame fields being given by $e_I:=s^*\iota(P_I)$. With these ingredients, we can define an action on $M$ via 
\begin{equation}
S[A]=\frac{1}{4\kappa}\int_{M}{s^*(F(\omega)^{IJ}\wedge\theta^K\wedge\theta^L)\epsilon_{IJKL}}
\label{eq:3.1.20}
\end{equation}
where $\kappa=8\pi G$. This action turns out to coincide with the first order Palatini action of pure gravity. As we see, the whole theory including the underlying structure of the spacetime is completely encoded in the Cartan connection. \\
There also exists another version of action (\ref{eq:3.1.20}) which depends on the Cartan connection in a more explicit way. This requires a non-vanishing cosmological constant which we take as negative for convenience (for a positive cosmological constant, this in fact completely analogous). Therefore, let us consider Cartan geometry on anti-de Sitter $(\mathrm{SO}(2,3),\mathrm{SO}^+(1,3))$. Since then $[P_I,P_J]=\frac{1}{L^2}M_{IJ}$, it follows that the Lorentzian part of the Cartan curvature acquires an additional contribution yielding
\begin{equation}
F(A)^{IJ}=F(\omega)^{IJ}+\frac{1}{L^2}\theta^I\wedge \theta^J
\label{eq:}
\end{equation}
The \emph{MacDowell-Mansouri action} is then defined as follows \cite{MacDowell:1977jt,Wise:2006sm}
\begin{equation}
S[A]=\frac{L^2}{8\kappa}\int_{M}{s^*(F(A)^{IJ}\wedge F(A)^{KL})\epsilon_{IJKL}}
\label{eq:3.1.21}
\end{equation} 
which has the form of an ordinary $\mathrm{SO}^+(1,3)$ Yang-Mills action. Using that the $F(\omega)\wedge F(\omega)$-term is a Gauss-Bonnet term an thus purely topological, this indeed yields first order Einstein gravity with a nontrivial cosmological constant.

\section{Super Cartan geometry and Supergravity}\label{superCartan geometry}
As explained in the previous section, gravity has a very elegant geometric interpretation in terms of a Cartan geometry modeled on Klein geometry of flat Minkowski, de Sitter or anti-de Sitter spacetime. As it turns out, this description also carries over to the super category providing a geometrical foundation of supergravity. This is in fact the starting point of the D'Auria-Fré approach to supergravity \cite{DAuria:1982uck,Castellani:1991et}. However, as already discussed in section \ref{Holonomies}, in order to obtain nontrivial fermionic degrees of freedom as well as supersymmetry transformations on the body of a supermanifold, we will define the notion of super Cartan geometry using the concept of enriched categories. In \cite{Hack:2015vna}, \emph{super Cartan structures} on supermanifolds were introduced and also lifted trivially to Cartan structures in the relative category. However, to the best of the author's knowledge, a precise definition of super Cartan geometry on (nontrivial) principal super fiber bundles in the framework of enriched categories seems not to exist so far in the mathematical literature (for a different approach towards a mathematical rigorous formulation of geometric supergravity using the notion of \emph{chiral triples} see \cite{Cortes:2018lan}).\\
For a motivation of super Cartan geometry, let us consider first the 'flat' case given by a super Klein geometry\footnote{All the following definitions will be formulated in the $H^{\infty}$ category. However, they can also be extended to the algebraic category without major changes (see also \cite{Konsti-FB:2020}).}.      
\begin{definition}
A \emph{super Klein geometry} is a pair $(\mathcal{G},\mathcal{H})$ consisting of a super Lie group $\mathcal{G}$ as well as an embedded super Lie subgroup $\mathcal{H}\hookrightarrow\mathcal{G}$ such that $\mathcal{G}/\mathcal{H}$ is connected.
\end{definition}
\begin{remark}
Suppose one has given a pair $(\mathcal{G},\mathcal{H})$ of super Lie groups with $\mathcal{H}\hookrightarrow\mathcal{G}$ an embedded super Lie subgroup. By definition of the DeWitt topology, $\mathcal{G}/\mathcal{H}$ is connected iff $\mathbf{B}(\mathcal{G}/\mathcal{H})\cong\mathbf{B}(\mathcal{G})/\mathbf{B}(\mathcal{H})$ is connected, that is, iff $(\mathbf{B}(\mathcal{G}),\mathbf{B}(\mathcal{H}))$ is a Klein geometry.
\end{remark}
As shown in \cite{Tuynman:2018}, as in the classical theory, a super Klein geometry $(\mathcal{G},\mathcal{H})$ canonically induces super fiber bundle with typical fiber $\mathcal{H}$ via
\begin{displaymath}
	 \xymatrix{
         \mathcal{G}\ar[d]_{\pi}  &    \mathcal{H}\ar[l] \\
            \mathcal{G}/\mathcal{H} &     
     }
		\label{eq:4.1}
 \end{displaymath}
together with the natural $\mathcal{H}$-right action $\Phi:\,\mathcal{G}\times\mathcal{H}\rightarrow\mathcal{G}$ on $\mathcal{G}$. Hence $\mathcal{H}\rightarrow\mathcal{G}\stackrel{\pi}{\rightarrow}\mathcal{G}/\mathcal{H}$ has the structure of a principal $\mathcal{H}$-bundle. Let $(X_i)_i$ be a homogeneous basis of $\mathrm{Lie}(\mathcal{G})\cong T_e\mathcal{G}=\mathfrak{g}\otimes\Lambda$ and $(\tensor[^i]{\omega}{})_i$ the associated left dual basis\footnote{given a super $\Lambda$-vector space $\mathcal{V}:=V\otimes\Lambda$ with homogeneous basis $(T_a)$, the corresponding left-dual basis $(\tensor[^b]{T}{})$ is a homogeneous basis of $\tensor[^*]{\mathcal{V}}{}$, i.e., left linear maps $\tensor[^a]{T}{}:\,\mathcal{V}\rightarrow\Lambda$ satisfying $\braket{T_{a}|\tensor[^b]{T}{}}=\delta^{b}_a$} of left-invariant 1-forms $\tensor[^i]{\omega}{}\in\Omega^1(\mathcal{G})$ on $\mathcal{G}$ satisfying $\braket{X_{i}|\tensor[^j]{\omega}{}}=\delta{_i^j}$, $\forall i,j=1,\ldots,n$. On $\mathcal{G}$, one can then define the \emph{(super) Maurer-Cartan form} $\theta_{\mathrm{MC}}\in\Omega^1(\mathcal{G},\mathfrak{g})$ via
\begin{equation}
\theta_{\mathrm{MC}}:=\tensor[^i]{\omega}{}\otimes X_i
\label{eq:4.2}
\end{equation}
yielding a smooth $\mathrm{Lie}(\mathcal{G})$-valued left-invariant 1-form on $\mathcal{G}$. By definition, the fundamental vector fields on $\mathcal{G}$ correspond to the subspace of left-invariant vector fields $X\in\mathrm{Lie}(\mathcal{H})$. Hence, it follows immediately from (\ref{eq:4.2}) that $\braket{X|\theta_{\mathrm{MC}}}=X_e$ $\forall X\in\mathrm{Lie}(\mathcal{H})$. Moreover, this also implies that the map $(\theta_{\mathrm{MC}})_g:\,T_g\mathcal{G}\rightarrow T_e\mathcal{G}$ is an isomorphism of super $\Lambda$-modules for any $g\in\mathcal{G}$ (and even an isomorphism of super $\Lambda$-vector spaces if $g\in\mathbf{B}(\mathcal{G})$). Finally, since the right action on $\mathcal{G}$ essentially coincides with the restriction of the group multiplication, it can be shown that \cite{Tuynman:2004}
\begin{equation}
R_{h}^{*}\theta_{\mathrm{MC}}=\mathrm{Ad}_{h^{-1}}\circ\theta_{\mathrm{MC}}
\label{eq:4.3}
\end{equation}
$\forall h\in\mathcal{H}$, where $R_{h}^{*}$ denotes the generalized pullback w.r.t. the right translation $R_h:=\Phi(\cdot,h)$, on $\mathcal{G}$ w.r.t. $h\in\mathcal{H}$ (see remark \ref{remark:2.19} and \cite{Tuynman:2004}). This motivates the following definition. 

\begin{definition}\label{Def:4.2}
A \emph{super Cartan geometry} $(\pi_{\mathcal{S}}:\,\mathcal{P}_{/\mathcal{S}}\rightarrow\mathcal{M}_{/\mathcal{S}},\mathcal{A})$ modeled on a super Klein geometry $(\mathcal{G},\mathcal{H})$ is a $\mathcal{S}$-relative principal super fiber bundle $\mathcal{H}\rightarrow\mathcal{P}_{/\mathcal{S}}\rightarrow\mathcal{M}_{/\mathcal{S}}$ with structure group $\mathcal{H}$ together with a smooth even $\mathrm{Lie}(\mathcal{G})$-valued $\mathcal{S}$-relative 1-form $\mathcal{A}\in\Omega^1(\mathcal{P}_{/\mathcal{S}},\mathfrak{g})_0$ on $\mathcal{P}_{/\mathcal{S}}$ called \emph{super Cartan connection} such that
\begin{enumerate}[label=(\roman*)]
	\item $\braket{\mathds{1}\otimes\widetilde{X}|\mathcal{A}}=X$,\quad $\forall X\in\mathfrak{h}$
	\item $(\Phi_{\mathcal{S}})^*_{h}\mathcal{A}=\mathrm{Ad}_{h^{-1}}\circ\mathcal{A}$,\quad $\forall h\in\mathcal{H}$
	\item choosing an embedding $\iota_{\mathcal{P}}:\,\mathcal{P}\hookrightarrow\mathcal{P}\times\mathcal{S}$, the pullback of $\mathcal{A}$ w.r.t. $\iota_{\mathcal{P}}$ yields an isomorphism $\iota_{\mathcal{P}}^*\mathcal{A}_p:\,T_p\mathcal{P}\rightarrow\mathrm{Lie}(\mathcal{G})$ of free super $\Lambda$-modules for any $p\in\mathcal{P}$
\end{enumerate}
where the last condition will be called the \emph{super Cartan condition}. If (iii) is not satisfied, we call $\mathcal{A}$ a \emph{generalized super Cartan connection}. 
\end{definition}
In the following, we will be interested on a particular subclass of super Cartan geometries. More precisely, we assume that the super Lie algebra $\mathfrak{g}$ of $\mathcal{G}$ admits a decomposition of the form $\mathfrak{g}=\mathfrak{g}/\mathfrak{h}\oplus\mathfrak{h}$ with $\mathfrak{h}$ the super Lie algebra of $\mathfrak{h}$ and $\mathfrak{g}/\mathfrak{h}$ a super vector space which is invariant w.r.t. the Adjoint action of $\mathcal{H}$. As in the classical theory, we will call such a super Cartan geometry \emph{reductive}.

Thus, let $\mathcal{A}\in\Omega^1(\mathcal{P}_{/\mathcal{S}},\mathfrak{g})_0$ be a super Cartan connection corresponding to a reductive super Cartan geometry. Let us split this connection according to the decomposition $\mathfrak{g}=\mathfrak{g}/\mathfrak{h}\oplus\mathfrak{h}$ of the super Lie algebra of $\mathcal{G}$ yielding
\begin{equation}
\mathcal{A}=\mathrm{pr}_{\mathfrak{g}/\mathfrak{h}}\circ\mathcal{A}+\mathrm{pr}_{\mathfrak{h}}\circ\mathcal{A}=:E+\omega
\label{eq:4.4}
\end{equation}
where $E$ will be called the \emph{supervielbein}. It then follows from condition (i) and (ii) above that $\omega$ defines $\mathcal{S}$-relative super connection 1-form according to \ref{Def:2.18}. Moreover, the supervielbein is even and, since $\mathfrak{g}/\mathfrak{h}$ defines a super $\mathcal{H}$-module, horizontal of type $(\mathcal{H},\mathrm{Ad})$, i.e., $E\in\Omega^1_{hor}(\mathcal{P}_{/\mathcal{S}},\mathfrak{g}/\mathfrak{h})_0$.

Let $s\in\mathbf{B}(\mathcal{S})$ and $\iota_{\mathcal{P}}:\,\mathcal{P}\rightarrow\mathcal{S}\times\mathcal{P}:\,p\mapsto(s,p)$ be a smooth embedding. This induces a smooth horizontal 1-form $\iota_{\mathcal{P}}^*E$ on $\mathcal{P}$ which, by condition (iii), is non-degenerate. Furthermore, it follows immediately that $\iota_{\mathcal{P}}^*\omega\in\Omega^1(\mathcal{P},\mathfrak{h})_0$ defines an ordinary super connection 1-form on $\mathcal{P}$. Hence, as in section \ref{Cartan geometry}, it follows that this induces a morphism between $\mathcal{P}$ and the frame bundle $\mathscr{F}(\mathcal{M})$ via  
\begin{align}
\mathcal{P}\rightarrow\mathscr{F}(\mathcal{M}),\,p\mapsto D_p\pi\circ E_p^{-1}
\label{eq:4.5}
\end{align}
where, for any $p\in\mathcal{P}$, $D_p\pi\circ E_p^{-1}:\,(\mathfrak{g}/\mathfrak{h})\otimes\Lambda\stackrel{\sim}{\rightarrow}T_{\pi(p)}M$ is an isomorphism of free super $\Lambda$-modules. Hence, $\mathcal{P}$ defines a $\mathcal{H}$-reduction of $\mathscr{F}(\mathcal{M})$. Moreover, it follows that $\mathcal{M}$ has the same dimension as the super vector space $\mathfrak{g}/\mathfrak{h}$.\\
\\ 
We want to apply the above definition to supergravity. Hence, we consider super Cartan geometry $(\pi_{\mathcal{S}}:\,\mathcal{P}_{/\mathcal{S}}\rightarrow\mathcal{M}_{/\mathcal{S}},\mathcal{A})$ modeled over the Klein geometry\footnote{for notational convenience, we will identify $\mathrm{Spin}^+(1,3)$ with the bosonic supermanifold $\mathbf{S}(\mathrm{Spin}^+(1,3))$ in what follows.} $(\mathrm{ISO}(\mathbb{R}^{1,3|4}),\mathrm{Spin}^+(1,3))$ corresponding to super Minkowski spacetime or $(\mathrm{OSp}(1|4),\mathrm{Spin}^+(1,3))$ for super AdS spacetime. In both cases the super Lie algebra $\mathfrak{g}$ can be decomposed as 
\begin{equation}
\mathfrak{g}=\mathfrak{g}_0\oplus\mathfrak{g}_1\cong\mathfrak{spin}^+(1,3)\oplus\mathbb{R}^{1,3}\oplus\Delta_{\mathbb{R}}=:\mathfrak{spin}^+(1,3)\oplus\mathfrak{t}
\label{eq:}
\end{equation}
with $\Delta_{\mathbb{R}}$ the vector space of the four-dimensional real Majorana representation $\kappa_{\mathbb{R}}$ and the super vector space $\mathfrak{t}\equiv\mathfrak{t}^{1,3|4}$, which in case of vanishing cosmological constant, can be identfied with the super Lie algebra of the super translation group $\mathcal{T}^{1,3|4}$ (see remark \ref{remark:2.13}). Hence, the supervielbein decomposes as 
\begin{equation}
E=\mathrm{pr}_{\mathfrak{t}}\circ\mathcal{A}=:\theta+\psi=:\theta^IP_I+\psi^{\alpha}Q_{\alpha}
\label{eq:}
\end{equation}
with $\psi\in\Omega^1_{hor}(\mathcal{P}_{/\mathcal{S}},\Delta_{\mathbb{R}})_0$ and $\theta\in\Omega^1_{hor}(\mathcal{P}_{/\mathcal{S}},\mathbb{R}^{1,3})_0$ defining $\mathcal{S}$-relative horizontal 1-forms of type $(\mathrm{Spin}^+(1,3),\mathrm{Ad})$ called the \emph{Rarita-Schwiger field} and co-frame, respectively. Hence, the super Cartan connection takes the form
\begin{equation}
\mathcal{A}=\theta^IP_I+\frac{1}{2}\omega^{IJ}M_{IJ}+\psi^{\alpha}Q_{\alpha}
\label{eq:}
\end{equation}
Due to (\ref{eq:4.5}), the supervielbein induces a $\mathrm{Spin}^+(1,3)$-reduction $\mathcal{P}\rightarrow\mathscr{F}(\mathcal{M})$ of the frame bundle. Applying the body functor, this in turn induces a $\mathrm{Spin}^+(1,3)$-reduction $P:=\mathbf{B}(\mathcal{P})\rightarrow\mathscr{F}(M)$ of the frame bundle of the body $M:=\mathbf{B}(\mathcal{M})$. That is, the body carries a spin structure.
\begin{remark}
Let $\iota:\,\mathcal{P}\rightarrow\mathcal{S}\times\mathcal{P}$ be an embedding and suppose $\mathcal{P}$ is trivial, i.e., $\mathcal{P}\cong\mathcal{M}\times\mathcal{G}$ with respect to a global trivialization $s:\,\mathcal{M}\rightarrow\mathcal{P}$. This in turn induces a (homogeneous) basis $(s_I,s_{\alpha})$ of global sections $s_I:=[s,P_I]$, $s_{\alpha}:=[s,Q_{\alpha}]$ of the associated super vector bundle $\mathcal{E}:=\mathcal{P}\times_{\mathrm{Spin}^+(1,3)}(\mathfrak{t}\otimes\Lambda)$. This yields an isomorphism
\begin{align}
\Omega^1(\mathcal{M},\mathfrak{t}^*)&\rightarrow\Omega^1(\mathcal{M},\mathcal{E})\cong\Omega^1_{hor}(\mathcal{P},\mathfrak{t})^{(\mathrm{Spin}^+(1,3),\mathrm{Ad})}\nonumber\\
\omega&\mapsto \omega^{I}s_{I}+\omega^{\alpha}s_{\alpha}
\label{eq:4.0.0}
\end{align}
It thus follows from condition (iii) of a super Cartan connection that, via (\ref{eq:4.0.0}), the pullback $\iota^*E\in\Omega^1_{hor}(\mathcal{P},\mathfrak{t})$ induces a non-degenerate 1-form $\tilde{E}\in\Omega^1(\mathcal{M},\mathfrak{t}^*)$. Consequently, the pair $(\mathcal{M},\tilde{E})$ defines super Cartan structure in the sense of \cite{Hack:2015vna}. Conversely, if $(\mathcal{M},\tilde{E})$ is a super Cartan structure with $\tilde{E}\in\Omega^1(\mathcal{M},\mathfrak{t})$ being non-degenerate, one can use (\ref{eq:4.0.0}) to get a non-degenerate 1-form $E\in\Omega^1_{hor}(\mathcal{P},\mathfrak{t})$ which can be lifted trivially to a $\mathcal{S}$-relative 1-form $\mathds{1}\otimes E\in\Omega^1_{hor}(\mathcal{P}_{/\mathcal{S}},\mathfrak{t})$ satisfying condition (iii) above. Hence, definition (\ref{Def:4.2}) provides a generalization of super Cartan structures in the sense of \cite{Hack:2015vna} to a generalized notion of super Cartan connections on nontrivial $\mathcal{S}$-relative principal super fiber bundles.    
\end{remark} 
The action of $\mathcal{N}=1$, $D=4$ supergravity can be obrtained as a supersymmetric generalization of the MacDowell-Mansouri action (\ref{eq:3.1.21}) on AdS spacetime \cite{MacDowell:1977jt}. We therefore consider super Cartan geometry modeled on the super Klein geometry $(\mathrm{OSp}(1|4),\mathrm{Spin}^+(1,3))$. The \emph{Cartan curvature} of $\mathcal{A}$ is given by 
\begin{equation}
F(\mathcal{A})=\mathrm{d}\mathcal{A}+\frac{1}{2}[\mathcal{A}\wedge\mathcal{A}]=\mathrm{d}\mathcal{A}+\frac{1}{2}(-1)^{|T_{\underline{A}}||T_{\underline{B}}|}\mathcal{A}^{\underline{A}}\wedge\mathcal{A}^{\underline{B}}\otimes[T_{\underline{A}},T_{\underline{B}}]
\label{eq:4.0.2}
\end{equation}
with respect to a homogeneous basis $(T_{\underline{A}})_{\underline{A}}$ of $\mathfrak{osp}(1|4)$, $\underline{A}\in(I,IJ,\alpha)$, where the minus sign in (\ref{eq:4.0.2}) appears due to the (anti)commutation of $T_{\underline{A}}$ and $\mathcal{A}^{\underline{B}}$. It then follows from $[M_{IJ},P_K]=\eta_{IK}P_J-\eta_{JK}P_I$ as well as (\ref{eq:2.3.19})-(\ref{eq:2.3.22}) that the components of $F(\mathcal{A})$ in the translational part of the super Lie algebra, also called the \emph{supertorsion}, take the form 
\begin{align}
F(\mathcal{A})^I&=\mathrm{d}\theta^I+\tensor{\omega}{^I_J}\wedge \theta^J+\frac{1}{4}((-1)^{|Q_{\alpha}||Q_{\beta}|}\psi^{\alpha}\wedge\psi^{\beta}\otimes[Q_{\alpha},Q_{\beta}])^{I}\nonumber\\
&=\Theta^{(\omega) I}-\frac{1}{4}\bar{\psi}\wedge\gamma^I\psi
\label{eq:torsion}
\end{align}
since $(-1)^{|Q_{\alpha}||Q_{\beta}|}=-1$, with $\Theta^{(\omega)}$ is the torsion 2-form associated to the spin connection $\omega$. For the spinorial components, we find
\begin{align}
F(\mathcal{A})^{IJ}&=\mathrm{d}\omega^{IJ}+\tensor{\omega}{^I_K}\wedge\omega^{KJ}+\frac{1}{2L^2}\theta^I\wedge \theta^J-\frac{1}{2}(\psi^{\alpha}\wedge\psi^{\beta}\otimes[Q_{\alpha},Q_{\beta}])^{IJ}\nonumber\\
&=F(\omega)^{IJ}+\frac{1}{L^2}\theta^I\wedge \theta^J-\frac{1}{4L}\bar{\psi}\wedge\gamma^{IJ}\psi
\label{eq:4.0.3}
\end{align}
with $F(\omega)$ the curvature of $\omega$. Finally, for the odd part, we obtain
\begin{align}
\Sigma^{\alpha}:=F(\mathcal{A})^{\alpha}&=\mathrm{d}\psi^{\alpha}+\frac{1}{4}\omega^{IJ}\tensor{(\gamma_{IJ})}{^{\alpha}_{\beta}}\wedge\psi^{\beta}-\frac{1}{2L}e^I\wedge\psi^{\beta}\tensor{(\gamma_{I})}{^{\alpha}_{\beta}}\nonumber\\
&=D^{(\omega)}\psi^{\alpha}-\frac{1}{2L}e^I\wedge\psi^{\beta}\tensor{(\gamma_{I})}{^{\alpha}_{\beta}}
\label{eq:4.0.4}
\end{align}
with $D^{(\omega)}\psi=\mathrm{d}\psi+\frac{1}{4}\omega^{IJ}\gamma_{IJ}\wedge\psi$ the exterior covariant derivative in the Majorana representation. Before we state the supergravity action, note that, by the rheonomy principle, the physical degrees of freedom are completely encoded on the body of the supermanifold. Hence, it suffices to consider local sections $\tilde{s}:\,U_{/\mathcal{S}}\rightarrow\iota_{\mathcal{M}_0}^*\mathcal{P}_{/\mathcal{S}}$ of the pullback bundle $\iota_{\mathcal{M}_0}^*\mathcal{P}_{/\mathcal{S}}$ on the bosonic sub supermanifold $\iota_{\mathcal{M}_0}:\,(\mathcal{M}_0)_{/\mathcal{S}}\hookrightarrow\mathcal{M}_{/\mathcal{S}}$. With respect to these type of localizations, the \emph{super MacDowell-Mansouri action} for $\mathcal{N}=1$, $D=4$ supergravity then reads as follows \cite{MacDowell:1977jt,Castellani:2013iq}
\begin{equation}
S(\mathcal{A})=\frac{L^2}{2\kappa}\int_{M}{\tilde{s}^*\left(\frac{1}{4}F(\mathcal{A})^{IJ}\wedge F(\mathcal{A})^{KL}\epsilon_{IJKL}+\frac{i}{L}\bar{\Sigma}\wedge\gamma_{*}\Sigma\right)}
\label{eq:4.0.5}
\end{equation}
Since the trace over internal indices in (\ref{eq:4.0.5}) is manifestly $\mathrm{Spin}^+(1,3)$-invariant, it is clear that the action invariant under local $\mathrm{Spin}^+(1,3)$-gauge transformations and therefore does not depend on the choice of the local trivialization. To see that this in fact leads to $\mathcal{N}=1$ supergravity, let us further evaluate the individual terms in (\ref{eq:4.0.5}). From (\ref{eq:4.0.3}), we conclude
\begin{align}
F(\mathcal{A})^{IJ}\wedge F(\mathcal{A})^{KL}\epsilon_{IJKL}=&\frac{2}{L^2}F(\omega)^{IJ}\wedge\theta^K\wedge\theta^L\epsilon_{IJKL}+F(\omega)^{IJ}\wedge F(\omega)^{KL}\epsilon_{IJKL}\nonumber\\
&-\frac{1}{2L^3}\bar{\psi}\wedge\gamma^{IJ}\psi\wedge\theta^K\wedge\theta^L\epsilon_{IJKL}-\frac{1}{2L}F(\omega)^{IJ}\wedge\bar{\psi}\wedge\gamma^{KL}\psi\epsilon_{IJKL}\nonumber\\
&+\frac{1}{L^4}\theta^I\wedge\theta^J\wedge\theta^K\wedge\theta^L\epsilon_{IJKL}
\label{eq:4.0.6}
\end{align}
For the second term, note that the conjugate $\bar{\Sigma}$ is given by
\begin{align}
\bar{\Sigma}=\Sigma^TC=&\mathrm{d}\bar{\psi}+\frac{1}{4}\bar{\psi}\wedge\omega^{IJ}\gamma_{IJ}-\frac{1}{2L}\bar{\psi}\gamma_I\wedge\theta^I\nonumber\\
=&D^{(\omega)}\bar{\psi}-\frac{1}{2L}\bar{\psi}\gamma_I\wedge \theta^I
\label{eq:4.0.7}
\end{align}
Hence, this yields
\begin{align}
\bar{\Sigma}\wedge\gamma_{*}\Sigma=&D^{(\omega)}\bar{\psi}\wedge\gamma_{*}D^{(\omega)}\psi+\frac{1}{2L}D^{(\omega)}\bar{\psi}\wedge\gamma_{*}\gamma_{I}\psi\wedge\theta^I+\frac{1}{2L}\bar{\psi}\wedge\gamma_{*}\gamma_ID^{(\omega)}\psi\wedge\theta^I\nonumber\\
&-\frac{1}{4L^2}\bar{\psi}\wedge\gamma_{IJ}\gamma_{*}\psi\wedge\theta^I\wedge\theta^J\nonumber\\
=&D^{(\omega)}\bar{\psi}\wedge\gamma_{*}D^{(\omega)}\psi+\frac{1}{L}\bar{\psi}\wedge\gamma_{*}\gamma_ID^{(\omega)}\psi\wedge\theta^I-\frac{1}{8iL^2}\bar{\psi}\wedge\gamma^{KL}\gamma_{*}\psi\wedge\theta^I\wedge\theta^J\epsilon_{IJKL}
\label{eq:4.0.8}
\end{align}
where in the last line we used that $\tensor{\epsilon}{_{IJ}^{KL}}\gamma_{KL}=2i\gamma_{IJ}\gamma_{*}$. If we insert (\ref{eq:4.0.6}) and (\ref{eq:4.0.8}) into (\ref{eq:4.0.5}), this then leads to 
\begin{align}
S(\mathcal{A})=&\frac{1}{2\kappa}\int_{M}\left(\frac{1}{2}F(\omega)^{IJ}\wedge e^K\wedge e^L\epsilon_{IJKL}-\frac{1}{4L}\bar{\psi}\wedge\gamma^{IJ}\psi\wedge e^K\wedge e^L\epsilon_{IJKL}\right.\nonumber\\
&\left.+\frac{1}{4L^2}e^I\wedge e^J\wedge e^K\wedge e^L\epsilon_{IJKL}+
i\bar{\psi}\wedge\gamma_{*}\gamma_ID^{(\omega)}\psi\wedge e^I\right)\nonumber\\
&-\frac{L}{16\kappa}\int_M{\epsilon_{IJKL}F(\omega)^{IJ}\wedge\bar{\psi}\wedge\gamma^{KL}\psi-8iD^{(\omega)}\bar{\psi}\wedge\gamma_{*}D^{(\omega)}\psi}
\label{eq:4.0.9}
\end{align}
where we set $\tilde{s}^*\theta^I=:e^I$ and we have dropped a topological term proportional to the Gauss-Bonnet term. In fact, it turns out that (\ref{eq:4.0.9}) can be simplified even further. Therefore, recall that $\psi$, as odd part of the supervielbein $E$, defines a horizontal 1-form of type $(\mathrm{Spin}^+(1,3),\kappa_{\mathbb{R}})$, i.e. $\psi\in\Omega^1(P_{/\mathcal{S}},\Delta_{\mathbb{R}})^{(\mathrm{Spin}^+(1,3),\kappa_{\mathbb{R}})}$. Hence, from the standard rules for the exterior covariant derivative, one obtains the \emph{Bianchi identity}
\begin{equation}
D^{(\omega)}D^{(\omega)}\psi=\kappa_{\mathbb{R}*}(F(\omega))\wedge\psi=\frac{1}{4}F(\omega)^{IJ}\gamma_{IJ}\wedge\psi
\label{eq:}
\end{equation}
One can thus equivalently write for the first term on the last line of (\ref{eq:4.0.9})
\begin{align}
\epsilon_{IJKL}F(\omega)^{IJ}\wedge\bar{\psi}\wedge\gamma^{KL}\psi&=2i\bar{\psi}\wedge F(\omega)^{IJ}\gamma_{IJ}\gamma_{*}\wedge\psi\nonumber\\
&=8i\bar{\psi}\wedge\gamma_{*}D^{(\omega)}D^{(\omega)}\psi\nonumber\\
&=8iD^{(\omega)}\bar{\psi}\wedge\gamma_{*}D^{(\omega)}\psi-\mathrm{d}(8i\bar{\psi}\wedge\gamma_{*}D^{(\omega)}\psi)
\label{eq:}
\end{align}
Thus dropping the boundary term, one finally ends up with
\begin{align}
S(\mathcal{A})=&\frac{1}{2\kappa}\int_{M}\left(\frac{1}{2}F(\omega)^{IJ}\wedge e^K\wedge e^L\epsilon_{IJKL}+
i\bar{\psi}\wedge\gamma_{*}\gamma_ID^{(\omega)}\psi\wedge e^I\right.\nonumber\\
&\left.-\frac{1}{4L}\bar{\psi}\wedge\gamma^{IJ}\psi\wedge e^K\wedge e^L\epsilon_{IJKL}+\frac{1}{4L^2}e^I\wedge e^J\wedge e^K\wedge e^L\epsilon_{IJKL}  \right)
\label{eq:Action}
\end{align}
This is precisely the action of $\mathcal{N}=1$, $D=4$ anti-de Sitter supergravity as stated for instance in \cite{Freedman:2012zz}. By definition, $S(\mathcal{A})\in H^{\infty}(\mathcal{S})_0$ describes an even functional on the parametrizing supermanifold $\mathcal{S}$. In particular, since we are working in the relative category, it transforms covariantly under change of parametrization. That is, given a change of parametrization $\lambda:\,\mathcal{S}'\rightarrow\mathcal{S}$, one has
\begin{equation}
\lambda^*S(\mathcal{A})=S(\lambda^*\mathcal{A})\in H^{\infty}(\mathcal{S}')_0
\label{eq:143}
\end{equation}
These are precisely the properties to be satisfied by a physical quantity according to \cite{Schmitt:1996hp} since physical degrees of freedom should not depend on the choice of a particular parametrizing supermanifold.
\begin{remark}
Due to (\ref{eq:143}), it is suggested in \cite{Schmitt:1996hp} to take $\mathcal{S}$ as the (infinite dimensional) configuration space as the action functional defined on any other finite dimensional parametrizing supermanifold then may be obtained via pullback (i.e. change of parametrization). In this way, in particular, it follows that fermionic fields are described in terms of odd functionals on configuration space. This is in fact the interpretation of fermionic fields in pAQFT \cite{Rejzner:2011au,Rejzner:2016hdj} (cf. remark \ref{remark:2.39}).
\end{remark}
Note that, for $\mathcal{S}\cong\{*\}$, action (\ref{eq:4.0.5}) reduces to the action (\ref{eq:3.1.21}) of pure Einstein gravity. One thus needs to ensure that $\mathcal{S}$ has nontrivial odd dimensions because otherwise all fermionic fields in $S(\mathcal{A})$ would simply drop off. But, as we have seen from the computations above, the anticommuting nature of the fermion fields has been crucial for the derivation of the action (\ref{eq:Action}) and thus for the interpretation of supergravity in terms of super Cartan geometry. In fact, this is also important for the supersymmetry of the action (\ref{eq:Action}). 

As will become clear below, in the Cartan geometric picture, supersymmetry transformations have the  interpretation in terms of (field dependent) gauge transformations. Therefore, we need to lift $\mathcal{A}$ to a principal connection. This can be done considering associated bundles.
\begin{lemma}\label{cor:4.4}
Let $\mathcal{H}\rightarrow\mathcal{P}\rightarrow\mathcal{M}$ be a principal super fiber bundle with structure group $\mathcal{H}$ and $\mathcal{H}$-right action $\Phi:\,\mathcal{P}\times\mathcal{H}\rightarrow\mathcal{H}$. Let $\lambda:\,\mathcal{H}\rightarrow\mathcal{G}$ be a morphism of super Lie groups and $\rho_{\lambda}:=\mu_{\mathcal{G}}\circ(\lambda\times\mathrm{id}_{\mathcal{G}}):\,\mathcal{H}\times\mathcal{G}\rightarrow\mathcal{G}$ the induced $H^{\infty}$-smooth left action of $\mathcal{H}$ on $\mathcal{G}$. On $\mathcal{P}\times\mathcal{H}$ consider the map
\begin{equation}
\Phi^{\times}:\,(\mathcal{P}\times\mathcal{H})\times\mathcal{G}\rightarrow\mathcal{P}\times\mathcal{G},\,((p,h),g)\mapsto(\Phi(p,g),\rho_{\lambda}(g^{-1},h))
\label{eq:4.0.10}
\end{equation} 
Then, $\Phi^{\times}$ defines an effective $H^{\infty}$-smooth $\mathcal{G}$-right action on $\mathcal{P}\times\mathcal{G}$. Let $\mathcal{E}:=\mathcal{P}\times_{\rho_{\lambda}}\mathcal{G}:=(\mathcal{P}\times\mathcal{H})/\mathcal{G}$ be the corresponding coset space and $\pi_{\mathcal{E}}:\,\mathcal{E}\rightarrow\mathcal{M}$ be defined as
\begin{align}
\pi_{\mathcal{E}}:\,\mathcal{E}\rightarrow\mathcal{M},\,[p,g]\mapsto\pi_{\mathcal{P}}(p)
\label{eq:4.0.11}
\end{align}
Then, $\mathcal{E}$ can be equipped with the structure of a $H^{\infty}$ supermanifold such that $\pi_{\mathcal{E}}$ is a $H^{\infty}$-smooth surjective map and $(\mathcal{E},\pi_{\mathcal{E}},\mathcal{M},\mathcal{G})$ turns into a principal $\mathcal{G}$-bundle.\\ 
Furthermore, let $\iota:\,\mathcal{P}\rightarrow\mathcal{P}\times_{\mathcal{H}}\mathcal{G}$ be defined as $\iota(p):=[p,e]$ $\forall p\in\mathcal{P}$, then $\iota$ is smooth, fiber-preserving and $\mathcal{H}$-equivariant in the sense that $\iota\circ\Phi=\widetilde{\Phi}\circ(\iota\times\lambda)$. Moreover, if $\lambda:\,\mathcal{H}\hookrightarrow\mathcal{G}$ is an embedding, then $\iota$ is an embedding.
\end{lemma}
\begin{proof}
The proof is almost the same as in the classical theory (see \cite{Tuynman:2004,Konsti-FB:2020}). One only needs care about smoothness in the various constructions as smoothness is generally not preserved under partial evaluation of smooth functions (see remark \ref{remark:2.19}). But, it turns out that everything works fine since $e\in\mathbf{B}(\mathcal{G})$.  
\end{proof}
\begin{remark}\label{remark:4.5}
Given the associated bundle $\mathcal{P}\times_{\rho_{\lambda}}\mathcal{G}$ as constructed in corollary \ref{cor:4.4}, as in section \ref{Holonomies}, we can lift it trivially to a $\mathcal{S}$-relative principal $\mathcal{G}$-bundle $(\mathcal{P}\times_{\rho_{\lambda}}\mathcal{G})_{/\mathcal{S}}\cong\mathcal{P}_{/\mathcal{S}}\times_{\rho_{\lambda}}\mathcal{G}$.
\end{remark}
\begin{prop}\label{prop:4.6}
Let $\mathcal{H}\rightarrow\mathcal{P}_{/\mathcal{S}}\rightarrow\mathcal{M}_{/\mathcal{S}}$ be a $\mathcal{S}$-relative principal super fiber bundle with structure group $\mathcal{H}$ as well as $(\mathcal{G},\mathcal{H})$ a super Klein geometry. Then, there is a bijective correspondence between generalized super Cartan connections in $\Omega^1(\mathcal{P}_{/\mathcal{S}},\mathfrak{g})_0$ and super connection 1-forms in $\Omega^1(\mathcal{P}_{/\mathcal{S}}\times_{\mathcal{H}}\mathcal{G},\mathfrak{g})_0$ with $\mathcal{P}_{/\mathcal{S}}\times_{\mathcal{H}}\mathcal{G}$ the $\mathcal{G}$-extension of $\mathcal{P}_{/\mathcal{S}}$ as constructed in remark \ref{remark:4.5}.
\end{prop}
\begin{proof}[Sketch of Proof.]
One direction is immediate, i.e., given a $\mathcal{S}$-relative super connection 1-form $\mathcal{A}$ on $\mathcal{P}_{/\mathcal{S}}\times_{\mathcal{H}}\mathcal{G}$, the pullback $\hat{\iota}^*\mathcal{A}$ w.r.t. the embedding $\hat{\iota}:=\mathrm{id}\times\iota:\,\mathcal{P}_{/\mathcal{S}}\rightarrow\mathcal{P}_{/\mathcal{S}}\times_{\mathcal{H}}\mathcal{G}$, with $\iota$ as defined in \ref{cor:4.4}, yields a generalized super Cartan connection on $\mathcal{P}_{/\mathcal{S}}$ according to definition \ref{Def:4.2}.\\
Conversely, suppose $\mathcal{A}\in\Omega^1(\mathcal{P}_{/\mathcal{S}},\mathfrak{g})_0$ is a generalized super Cartan connection. Let $\hat{\pi}:\,\mathcal{P}_{/\mathcal{S}}\times\mathcal{G}\rightarrow\mathcal{P}_{/\mathcal{S}}\times_{\mathcal{H}}\mathcal{G}$ be the canonical projection. If $\hat{\Phi}_{\mathcal{S}}$ denotes the $\mathcal{G}$-right action on $\mathcal{P}_{/\mathcal{S}}\times_{\mathcal{H}}\mathcal{G}$, it follows that the fundamental vector fields are given by
\begin{align}
\widetilde{Y}_{[p,g]}&=(\hat{\Phi}_{\mathcal{S}})_{[p,g]*}(Y_e)=D_{(p,g)}\hat{\pi}(0_{p},L_{g*}Y)
\label{}
\end{align}
for any $Y\in\mathrm{Lie}(\mathcal{G})$ and $p\in\mathcal{P}_{/\mathcal{S}}$, $g\in\mathcal{G}$, where we used the generalized tangent map. Furthermore, for any $X_{p}\in T_{p}(\mathcal{P}_{/\mathcal{S}})$, one has
\begin{align}
D_{(p,g)}\hat{\pi}(X_{p},0_g)=\hat{\iota}_{*}((\Phi_{\mathcal{S}})_{g*}X_{p})
\label{}
\end{align}
$\forall (p,g)\in\mathcal{P}_{/\mathcal{S}}\times\mathcal{G}$. Hence, this yields
\begin{align}
D_{(p,g)}\hat{\pi}(X_{p},Y_g)&=D_{(p,g)}\hat{\pi}(X_{p},0_g)+D_{(p,g)}\hat{\pi}(0_{p},Y_g)\nonumber\\
&=\hat{\iota}_{*}((\Phi_{\mathcal{S}})_{g*}X_p)+D_{(p,g)}\hat{\pi}(0_p,L_{g*}\circ L_{g^{-1}*}(Y_g))\nonumber\\
&=\iota_{*}((\Phi_{\mathcal{S}})_{g*}X_p)+\widetilde{\braket{Y_g|\theta_{\mathrm{MC}}}}_{[p,g]}
\label{}
\end{align}
Therefore, if there exists a super connection 1-form $\hat{\iota}_{*}\mathcal{A}$ whose pullback under $\hat{\iota}$ is given by $\mathcal{A}$, then it necessarily has to be of the form
\begin{equation}
\braket{D_{(p,g)}\hat{\pi}(X_p,Y_g)|\hat{\iota}_{*}\mathcal{A}_{[p,g]}}=\mathrm{Ad}_{g^{-1}}\braket{X_p|\mathcal{A}_p}+\braket{Y_g|\theta_{\mathrm{MC}}}
\label{eq:p5.7.1}
\end{equation}
In particular, as $\hat{\pi}$ is a submersion, it is uniquely determined by (\ref{eq:p5.7.1}). Using that the kernel of $\hat{\pi}_{*}$ is given by
\begin{equation}
\mathrm{ker}\,D_{(p,g)}\hat{\pi}=\{((\mathds{1}\otimes Y)_{p},-R_{g*}Y)|Y\in\mathrm{Lie}(\mathcal{H})\}
\label{eq:6.15}
\end{equation}
it is easy to see that (\ref{eq:p5.7.1}) is indeed well-defined. That $\hat{\iota}_{*}\mathcal{A}$ is $\mathcal{G}$-equivariant and maps fundamental vector fields to the corresponding generator follows from the respective properties of $\mathcal{A}$ and the Maurer-Cartan form. Finally, by construction, it is clear that $\hat{\iota}_{*}\mathcal{A}$ defines an even $\mathcal{S}$-relative 1-form. For more details see \cite{Konsti-FB:2020}. 
\end{proof}
By Prop. \ref{prop:4.6}, we can thus lift the super Cartan connection $\mathcal{A}$ to principal connection 
\begin{equation}
\widehat{\mathcal{A}}:=\iota_{*}\mathcal{A}\in\Omega^1(\mathcal{P}_{/\mathcal{S}}\times_{\mathrm{Spin}^+(1,3)}\mathrm{OSp}(1|4),\mathfrak{osp}(1|4))_0
\label{eq:4.0.11}
\end{equation}
on the associated $\mathrm{OSp}(1|4)$-bundle. Since, in this way, it defines a super connection 1-form à la Ehresmann, we can consider local gauge transformations generated by vectors $Q_{\alpha}$ in the odd part of the super Lie algebra.\\ 
Let $s:\,U_{/\mathcal{S}}\rightarrow\mathcal{P}_{/\mathcal{S}}$ be a local section which can be extended to a local section $\widehat{s}:=(\mathrm{id}\times\iota)\circ s$ of the $\mathrm{OSp}(1|4)$-bundle. We consider a local gauge transformation $g:\,\mathcal{M}_{/\mathcal{S}}\rightarrow\mathrm{OSp}(1|4)$. If $\widehat{s}':=\widehat{s}\cdot g$, the connection $\widehat{\mathcal{A}}$ then transforms as
\begin{equation}
\widehat{\mathcal{A}}_{\widehat{s}'}=\mathrm{Ad}_{g^{-1}}\circ\mathcal{A}_s+g^{-1}\mathrm{d}g
\label{eq:4.0.12}
\end{equation}
If we consider an infinitesimal local gauge transformation generated by a smooth map $\epsilon\in H^{\infty}(\mathcal{M}_{/\mathcal{S}},\Delta_{\mathbb{R}}\otimes\Lambda_1)$, this yields
\begin{equation}
\delta_{\epsilon}\mathcal{A}_s=\mathrm{d}\epsilon+[\mathcal{A}_s\wedge\epsilon]=D^{(\mathcal{A})}\epsilon
\label{eq:4.0.13}
\end{equation}
Hence, using the super Lie algebra relations (\ref{eq:2.3.19})-(\ref{eq:2.3.22}), it follows that the individual components of the super connection 1-form transform as
\begin{align}
\delta_{\epsilon}e^I&=\frac{1}{2}\bar{\epsilon}\gamma^I\psi\label{eq:4.0.14}\\
\delta_{\epsilon}\psi^{\alpha}&=D^{(\omega)}\epsilon^{\alpha}-\frac{1}{2L}e^I\tensor{(\gamma_I)}{^{\alpha}_{\beta}}\epsilon^{\beta}\label{eq:4.0.15}\\
\delta_{\epsilon}\omega^{IJ}&=\frac{1}{2L}\bar{\epsilon}\gamma^{IJ}\psi
\label{eq:4.0.16}
\end{align}
The infinitesimal gauge transformations (\ref{eq:4.0.14}) and (\ref{eq:4.0.15}), when pulled back to the body manifold, are precisely the local supersymmetry transformations as stated in \cite{Freedman:2012zz}. However, under the additional transformation (\ref{eq:4.0.16}), $S(\mathcal{A})$ is not invariant. This is only true in case of a vanishing cosmological constant, i.e., $L\rightarrow\infty$, where (\ref{eq:4.0.16}) simply drops off. Varying $S(\mathcal{A})$ w.r.t. $\omega$ yields the field equations of the spin connection which are equivalent to the supertorsion constraint
\begin{equation}
F(\mathcal{A})^I=0\quad\Leftrightarrow\quad\Theta^{(\omega) I}=\frac{1}{4}\bar{\psi}\wedge\gamma^I\psi
\label{eq:4.0.17}
\end{equation}  
Thus, provided (\ref{eq:4.0.17}) holds, the supergravity action is indeed invariant under the full local gauge transformations generated by $\epsilon\in H^{\infty}((\mathcal{M}_0)_{/\mathcal{S}},\Delta_{\mathbb{R}}\otimes\Lambda_1)$ \cite{vanNieuwenhuizen:2004rh} (in fact, one can simply ignore the transformation of $\omega$ in this case). This is another nice feature of the super Cartan approach to supergravity as supersymmetry transformations have the geometrical intepretation in terms of gauge transformations which provides a link between supergravity and Yang-Mills theory \cite{vanNieuwenhuizen:2004rh}. This will also be crucial for the interpretation of the super Ashtekar connection in terms of a super gauge field in what follows.
\begin{remark}
Let us emphasize again that, technically, the existence of a non-vanishing $\epsilon$ relies crucially on the additional parametrizing supermanifold. Hence, working in the relative category resolves both, nontrivial anticommuting fermionic fields as well as supersymmetry transformations on the body of a supermanifold. 
\end{remark}

\section{The super Ashtekar connection}\label{Ashtekar}
\subsection{Derivation from the full theory}
In the previous section, we have derived $\mathcal{N}=1$, $D=4$ anti-de Sitter supergravity considering it geometrically in terms of a super Cartan geometry. In particular, all the basic entities of the theory turn out to be completely encoded in the super Cartan connection  
\begin{equation}
\mathcal{A}=e^IP_I+\frac{1}{2}\omega^{IJ}M_{IJ}+\psi^{\alpha}Q_{\alpha}
\label{eq:5.1.0}
\end{equation}
taking values in the super Lie algebra $\mathfrak{osp}(1|4)$ corresponding to the underlying super Klein geometry.\\
\\
In 1986 in \cite{Ashtekar:1986yd}, Ashtekar introduced his self-dual variables which give ordinary gravity the structure of a $\mathrm{SL}(2,\mathbb{C})$ Yang-Mills theory. This construction is based on a particular structure of the internal symmetry algebra. In fact, the complexification of the Lie algebra of the  orthochronous Lorentz group $\mathrm{SO}^+(1,3)$ has a decomposition of the form 
\begin{equation}
\mathfrak{so}^+(1,3)_{\mathbb{C}}=\mathfrak{su}(2)_{\mathbb{C}}\oplus\mathfrak{su}(2)_{\mathbb{C}}\cong\mathrm{sl}(2,\mathbb{C})\oplus\mathfrak{sl}(2,\mathbb{C})
\label{eq:5.1.0}
\end{equation}
and thus splits into two proper $\mathfrak{sl}(2,\mathbb{C})$ subalgebras (viewed as complex Lie algebras of complex $\mathrm{SL}(2,\mathbb{C})$). This precisely corresponds to the decomposition of the spin connection $\omega$ into its self-dual $A^+$ and anti self-dual part $A^-$, respectively. In this sense, the self-dual variables can be regarded as chiral sub components of the 4D spin connection.\\
Hence, the natural question arises whether such a construction carries over to the super category. As we will see in what follows, this will be indeed the case, even for extended supersymmetry. Therefore, recall that the Ashtekar variables $A^{\pm}$ are defined as the (anti) self-dual part of the four dimensional spin connection $\omega$ according to
\begin{equation}
A^{\pm}:=\frac{1}{2}\bigg{[}\frac{1}{2}\left(\omega^{IJ}\mp\frac{i}{2}\tensor{\epsilon}{^{IJ}_{KL}}\omega^{KL}\right)\bigg{]}M_{IJ}
\label{eq:5.1.1}
\end{equation}
which takes values in the complexification $\mathfrak{spin}(1,3)_{\mathbb{C}}$ of the Lie algebra of the spin double cover $\mathrm{Spin}^+(1,3)$ of the orthochronous Lorentz group generated by $M_{IJ}$. After some simple algebra, it follows that
\begin{align}
A^{\pm}&=\frac{1}{2}\bigg{[}\frac{1}{2}\left(\omega^{IJ}\mp\frac{i}{2}\tensor{\epsilon}{^{IJ}_{KL}}\omega^{KL}\right)\bigg{]}M_{IJ}\nonumber\\
&=\frac{1}{2}\left(\frac{1}{4}\tensor{\epsilon}{^i_{kl}}\tensor{\epsilon}{_i^{mn}}\omega^{kl}M_{mn}\mp\frac{i}{2}\tensor{\epsilon}{^{0i}_{kl}}\omega^{kl}M_{0i}\mp\frac{i}{2}\tensor{\epsilon}{_{0i}^{KL}}\omega^{0i}M_{KL}+\omega^{0i}M_{0i}\right)\nonumber\\
&=\left(-\frac{1}{2}\tensor{\epsilon}{^i_{kl}}\omega^{kl}\mp i\omega^{0i}\right)\frac{1}{2}\left(-\frac{1}{2}\tensor{\epsilon}{_{i}^{kl}}M_{kl}\pm iM_{0i}\right)=:A^{\pm i}T_i^{\pm}
\label{eq:5.1.3}
\end{align}
where $A^{\pm i}:=\Gamma^i\mp i K^i$, $i=1,\ldots,3$, with $\Gamma^i:=-\frac{1}{2}\tensor{\epsilon}{^i_{kl}}\omega^{kl}$ the 3D spin connection and $K^i:=\omega^{0i}$ the extrinsic curvature. Moreover, $T^{\pm}_i$ are given by 
\begin{equation}
T_i^{\pm}=\frac{1}{2}(J_i\pm i\tilde{K}_i)
\label{eq:5.1.4}
\end{equation}
with $J_i=-\frac{1}{2}\tensor{\epsilon}{_{i}^{jk}}M_{jk}$ and $\tilde{K}_i=M_{0i}$, the generators of local rotations and boosts, respectively. These satisfy the commutation relations 
\begin{equation}
[T^{\pm}_i,T_j^{\pm}]=\tensor{\epsilon}{_{ij}^k}T_k^{\pm}
\label{eq:5.1.5}
\end{equation}
and therefore generate the chiral $\mathfrak{sl}(2,\mathbb{C})$ subalgebras of $\mathfrak{spin}^+(1,3)_{\mathbb{C}}$. Since $\tensor{\gamma}{_{0i}}=\frac{i}{2}\tensor{\epsilon}{_{ijk}}\tensor{\gamma}{^{jk}}\gamma_{*}$, one has the important identity
\begin{align}
\frac{1}{4}\left(-\tensor{\epsilon}{_i^{jk}}\gamma_{jk}\pm i\gamma_{0i}\right)=\frac{\gamma_{*}\pm \mathds{1}}{2}\frac{i}{2}\gamma_{0i}
\label{eq:5.1.2}
\end{align} 
Hence, using (\ref{eq:5.1.2}) it immediately follows that the exterior covariant derivative induced by $A^+$ (resp. $A^-$) acts on purely unprimed (resp. primed) spinor indices according to 
\begin{equation}
D^{(A^+)}\psi^A=\mathrm{d}\psi^A+\tensor{{A^{+}}}{^A_B}\wedge\psi^B,\quad\text{and}\quad D^{(A^-)}\psi_{A'}=\mathrm{d}\psi_{A'}+\tensor{{A^{-}}}{_{A'}^{B'}}\wedge\psi_{B'} 
\label{eq:5.1.6}
\end{equation}
respectively, where $\tensor{{A^{+}}}{^A_B}=A^{+i}\tensor{(\tau_i)}{^A_B}$ and $\tensor{{A^{-}}}{_{A'}^{B'}}=A^{-i}\tensor{(\tau_i)}{_{A'}^{B'}}$ (note that the second identity in (\ref{eq:5.1.6}) can be obtained taking the complex conjugate of the first one). Hence, focusing for the moment on the self-dual sector, let us consider the chiral sub components $Q^r_A$ of the Majorana charges $Q^r_{\alpha}$. From (\ref{eq:2.3.19}) it then follows, using again (\ref{eq:5.1.2}),
\begin{align}
[T_k^+,Q^i_A]=Q^i_{B}\tensor{(\tau_k)}{^B_A}
\label{eq:5.1.7}
\end{align}
that is, the Weyl fermions $Q^r_A$ transform in the fundamental representation of $\mathfrak{sl}(2,\mathbb{C})$ as to be expected from (\ref{eq:5.1.7}). Next, let us consider the anti commutator relations between two Weyl fermions. In the Weyl-representation, the charge conjugation matrix $C$ admits a block diagonal form given by $C=\mathrm{diag}(i\epsilon,i\epsilon)$. From this, we immediately deduce
\begin{align}
(C\gamma^{IJ})_{AB}M_{IJ}&=\frac{i}{2}\left(\epsilon(\sigma^I\bar{\sigma}^J-\sigma^J\bar{\sigma}^I)\right)_{AB}M_{IJ}\nonumber\\
&=2i(\epsilon\sigma^i)_{AB}M_{0i}-\tensor{\epsilon}{^{ij}_k}(\epsilon\sigma^k)_{AB}M_{ij}\nonumber\\
&=2(\epsilon\sigma^i)_{AB}\left(iM_{0i}-\frac{1}{2}\tensor{\epsilon}{_{i}^{jk}}M_{jk}\right)=2(\epsilon\sigma^i)_{AB}T_i^+
\end{align}
Hence, using (\ref{eq:2.3.22}), it follows
\begin{align}
[Q^i_{A},Q^j_{B}]=\delta^{ij}\frac{1}{L}(\epsilon\sigma^k)_{AB}T_k^+-\frac{i}{2L}\epsilon_{AB}T^{ij}
\label{eq:5.1.8}
\end{align} 
The $R$-symmetry generators $T_{rs}$ do not mix the chiral components of the Majorana charges $Q^r_{\alpha}$. Thus, to summarize, we have found that $(T_i^+,T_{rs},Q_{A}^r)$ indeed form a proper chiral sub super Lie algebra of $\mathfrak{osp}(\mathcal{N}|4)_{\mathbb{C}}$ with the graded commutation relations
\begin{align}
[T^+_i,T_j^+]&=\tensor{\epsilon}{_{ij}^k}T_k^+\label{eq:5.1.9.1}\\
[T_i^+,Q^r_A]&=Q^r_{B}\tensor{(\tau_i)}{^B_A}\label{eq:5.1.9.2}\\
[Q^r_{A},Q^s_{B}]&=\delta^{rs}\frac{1}{L}(\epsilon\sigma^i)_{AB}T_i^+-\frac{i}{2L}\epsilon_{AB}T^{rs}\label{eq:5.1.9.3}\\
[T^{pq},Q_{A}^{r}]&=\delta^{qr} Q_{A}^p-\delta^{pr} Q_{A}^q
\label{eq:5.1.9.4}
\end{align}
which precisely coincide with graded commutation relations of the complex orthosymplectic Lie superalgebra $\mathfrak{osp}(\mathcal{N}|2)_{\mathbb{C}}$, the extended supersymmetric generalization of the isometry algebra of $D=2$ anti-de Sitter space \cite{Wipf:2016}. Performing the Inönü-Wigner contraction, i.e., taking the limit $L\rightarrow\infty$, this yields the extended $D=2$ super Poincaré algebra. Similarly, considering the anti self-dual sector, one obtains a proper sub super Lie algebra generated by the anti chiral components $(T_i^{-},T_{rs},Q ^{A'})$ which again forms $\mathfrak{osp}(\mathcal{N}|2)_{\mathbb{C}}$.\\
\\
For the construction of a super analog of Ashtekar's self-dual variables, in what follows, let us restrict to the non-extended case $\mathcal{N}=1$ in which case the theory is described in terms of the super Cartan connection $\mathcal{A}\in\Omega^1(\mathcal{P}_{/\mathcal{S}},\mathfrak{osp}(1|4))$ (\ref{eq:5.1.0}). Based on the above observations, we introduce the following graded self-dual variables 
\begin{equation}
\mathcal{A}^{+}:=A^{+ i}T_i^{+}+\psi^AQ_A\quad\text{and}\quad\mathcal{A}^{-}:=A^{- i}T_i^{-}+\psi_{A'}Q^{A'}
\label{eq:}
\end{equation}
which define $\mathcal{S}$-relative 1-forms on the $\mathcal{S}$-relative principal super fiber bundle $\mathrm{Spin}^+(1,3)\rightarrow\mathcal{P}_{/\mathcal{S}}\rightarrow\mathcal{M}_{/\mathcal{S}}$ with values in the chiral sub superalgbra $\mathfrak{osp}(1|2)_{\mathbb{C}}$ generated by $(T_i^+,Q_A)$ and $(T_i^-,Q^{A'})$, respectively. In case of a vanishing cosmological constant, this gives the $D=2$ super Poincaré algebra.
\begin{remark}
The $D=2$ super Poincaré algebra has an equivalent description in terms of a direct sum super Lie algebra $\mathrm{sl}(2,\mathbb{C})\oplus\Pi\mathbb{C}^2$, where $\Pi\mathbb{C}^2$ is regarded as a purely odd super vector space. Given the fundamental representation $\rho:\,\mathfrak{sl}(2,\mathbb{C})\rightarrow\mathrm{End}(\mathbb{C}^2)$ of $\mathfrak{sl}(2,\mathbb{C})$, the graded commutation relations are given by
\begin{equation}
[(x,v),(x',v')]:=([x,x'],\rho(x)(v)-\rho(x')(v'))
\label{eq:}
\end{equation}
$\forall (x,v),(x',v')\in\mathrm{sl}(2,\mathbb{C})\oplus\Pi\mathbb{C}^2$. In the mathematical literature, such kind of superalgebras are usually called \emph{generalized Takiff Lie superalgebras} \cite{Greenstein:2012}.
\end{remark}
For the rest of this section, let us focus on the chiral case (and the case of non-vanishing cosmological constant), the considerations for $\mathcal{A}^-$ are in fact completely analogous. Let us consider the complexification $\mathcal{P}^{\mathbb{C}}$ of $\mathcal{P}$ defined as the associated super $\mathrm{Spin}^+(1,3)_{\mathbb{C}}$-bundle 
\begin{equation}
\mathcal{P}^{\mathbb{C}}:=\mathcal{P}\times_{\mathrm{Spin}^+(1,3)}\mathrm{Spin}^+(1,3)_{\mathbb{C}}
\label{eq:}
\end{equation}
via the obvious mapping $\mathrm{Spin}^+(1,3)\rightarrow\mathrm{Spin}^+(1,3)_{\mathbb{C}}$. Due to (\ref{eq:5.1.0}), this bundle can be reduced to a super $\mathrm{SL}(2,\mathbb{C})$-bundle $\mathcal{Q}\hookrightarrow\mathcal{P}^{\mathbb{C}}$. It then follows from the chiral nature of the self-dual Ashtekar connection $A^+$ as well as the chiral sub components $\psi^A$ of the Rarita-Schwinger field that $\mathcal{A}^+$ induces a well-defined 1-form
\begin{equation}
\mathcal{A}^+\in\Omega^1(\mathcal{Q}_{/\mathcal{S}},\mathfrak{osp}(1|2)_{\mathbb{C}})_0
\label{eq:}
\end{equation}
which, by construction, satisfies the properties (i) and (ii) of definition \ref{Def:4.2}. That is, $\mathcal{A}^+$ defines a \emph{generalized super Cartan connection} on the $\mathcal{S}$-relative $\mathrm{SL}(2,\mathbb{C})$-bundle $\mathcal{Q}_{/\mathcal{S}}$. This is precisely the \emph{super Ashtekar connection} as first introduced in \cite{Fulop:1993wi}. There, this connection arose by studying the constraint algebra of the canonical theory to be discussed in the following section. Here, we have derived it using the geometrical description of $\mathcal{N}=1$, $D=4$ supergravity in terms of super Cartan geometry and studying the chiral structure of the underlying supersymmetry algebra corresponding to the super Klein geometry. In particular, it follows that that it has the interpretation in terms of a generalized super Cartan connection on the $\mathcal{S}$-relative $\mathrm{SL}(2,\mathbb{C})$-bundle $\mathcal{Q}_{/\mathcal{S}}$. 

Since the chiral structure of the supersymmetry algebra is even preserved in case of extended supersymmetry leading to the orthosymplectic Lie superalgebra $\mathfrak{osp}(\mathcal{N}|2)_{\mathbb{C}}$, this suggests that it might be possible to generalize this construction to extended $D=4$ supergravity theories including further chiral components of the supermultiplet in the definition of the super Ashtekar connection. In fact, graded Ashtekar variables for extended SUGRA theories have been studied in \cite{Tsuda:2000er,Sano:1992jw} and in context of constrained super BF-theory in \cite{Ezawa:1995nj,Ling:2003yw}. Recently, in \cite{Eder:2021nyb}, these variables have been derived for pure $\mathcal{N}=2$, $D=4$ supergravity in the presence of boundaries in a purely geometric way following the Cartan geometric approach as studied in this present article.
\begin{remark}
Let us emphasize that, since the above construction relied crucially on the chiral description of the theory, this construction cannot be carried over to \emph{real Barbero-Immirzi parameters}! In fact, real $\beta$ requires the consideration of both chiral components of the Majorana fermions $Q^r_{\alpha}$. But, the anti commutator between $Q^r_A$ and $Q^{r A'}$ is proportional to $P_I$ which is related to the soldering form $\theta$ corresponding to the dual electric field (see next section). Hence, this does not lead to a proper sub super Lie algebra and the super Ashtekar connection cannot be defined.
\end{remark}

\subsection{The canonical theory and the constraint algebra}\label{section:5.2}
As seen in he previous section, the chiral structure of the underlying supersymmetry algebra of $\mathcal{N}=1$, $D=4$ AdS supergravity enabled us to introduce a graded generalization of Ashtekar's self-dual variables. So far, the discussion has been purely off-shell. As a next step, one thus needs to derive a chiral variant of the supergravity action (\ref{eq:Action}) to specify the dynamics. We will therefore follow \cite{Jacobson:1987cj}. For a purely geometric discussion revealing the underlying $\mathrm{OSp}(1|2)_{\mathbb{C}}$-gauge symmetry of the chiral theory in a manifest way, see \cite{Eder:2021nyb}. Adapting to our conventions, it follows that the action of chiral supergravity takes the form (absorbing $\kappa$ in the Rarita-Schwinger field)
\begin{align}
S(e,\mathcal{A}^+)=&\int_M{\frac{i}{\kappa} \tensor{\Sigma}{^{AB}}\wedge\tensor{F(A^+)}{_{AB}}-e_{AA'}\wedge\bar{\psi}^{A'}\wedge D^{(A^+)}\psi^A}\label{eq:5.2.1}\\
&-\frac{1}{2L}\tensor{\Sigma}{^{AB}}\wedge\psi_A\wedge\psi_B+\frac{1}{2L}\tensor{\Sigma}{^{A'B'}}\wedge\bar{\psi}_{A'}\wedge\bar{\psi}_{B'}+\frac{i}{4\kappa L^2}\tensor{\Sigma}{^{AB}}\wedge\tensor{\Sigma}{_{AB}}\nonumber
\end{align}
with $\bar{\psi}^{A'}$ is the complex conjugate of $\psi^A$ and $F(A^+)=\mathrm{d}A^{+}+A^{+}\wedge A^+$ denotes the (self-dual) curvature of $A^+$. Moreover, the soldering form $e\in\Omega^1_{hor}(\mathcal{P}_{/\mathcal{S}},\mathbb{R}^{1,3})_0$ induces a horizontal 2-form $\Sigma\in\Omega^2(\mathcal{P}_{/\mathcal{S}},\mathfrak{spin}^+(1,3))_0$ which, in spinor indices, takes the form $\Sigma^{AA'BB'}=e^{AA'}\wedge e^{BB'}$. Due to antisymmetry, it can be decomposed according to  
\begin{equation}
\Sigma^{AA'BB'}=\epsilon^{AB}\Sigma^{A'B'}+\epsilon^{A'B'}\Sigma^{AB}
\label{eq:5.2.2}
\end{equation}
with $\Sigma^{AB}$ and $\Sigma^{A'B'}$ the self-dual anti self-dual part of $\Sigma^{AA'BB'}$, respectively, given by 
\begin{equation}
\Sigma^{AB}:=\frac{1}{2}\epsilon_{A'B'}\Sigma^{AA'BB'}\quad\text{and}\quad\Sigma^{A'B'}:=\frac{1}{2}\epsilon_{AB}\Sigma^{AA'BB'}
\label{eq:5.2.3}
\end{equation}
As we see, action (\ref{eq:5.2.1}) only depends on the chiral degrees freedom $(A^+,\psi^A)$ contained in the super Ashtekar connection as well as the soldering form $e$. The remaining field components are fixed via reality conditions (see discussion below). Provided these reality conditions are satisfied, it follows that the imaginary part of the action (\ref{eq:5.2.1}) becomes a boundary term \cite{Jacobson:1987cj}. Thus, in this way, one reobtains the field equations of ordinary real $\mathcal{N}=1$ supergravity.\\
\\
The canonical analysis of chiral $\mathcal{N}=1$ supergravity has been investigated in \cite{Jacobson:1987cj} and \cite{Fulop:1993wi} and, for arbitrary (real) Barbero-Immirzi parameters in \cite{Sawaguchi:2001wi,Tsuda:1999bg,Konsti-SUGRA:2020} and even higher spacetime dimensions in \cite{Bodendorfer:2011pb,Bodendorfer:2011pc}. Here, we will follow \cite{Konsti-SUGRA:2020,Konsti-Kos:2020}, where, for $\beta=-i$, the action takes the form
\begin{equation}
S(e,\mathcal{A}^+)=\int_{\mathbb{R}}{\mathrm{d}t\int_{\Sigma}\mathrm{d}^3x\,\frac{i}{\kappa}E^a_i L_{\partial_t}A^{+ i}_a}-\pi^a_A L_{\partial_t}\psi_a^A-A_t^i G_i+N^a H_a+N H-S^L_A\psi_t^A-\bar{\psi}_{t A'}S^{R A'}
\label{eq:5.2.7}
\end{equation} 
with $\pi_A^a$ the canonically conjugate momentum of $\psi^A_a$ which is related to the corresponding complex conjugate $\bar{\psi}_a^{A'}$ via the \emph{reality condition}
\begin{align}
\pi^a_A&=\tensor{\epsilon}{^{abc}}\bar{\psi}_{b}^{A'}e_{c AA'}
\label{eq:5.2.8}
\end{align} 
Furthermore, $E^a_i:=\sqrt{q}e^a_i$ is the usual electric field conjugate to $A^{+i}_a$. The canonical pairs $(A^{+i}_a,E^a_i)$ and $(\psi^A_a,\pi_A^a)$ build up a graded symplectic phase space with the nonvanishing Poisson brackets
\begin{equation}
\{E_a^i(x),A_b^{+ j}(y)\}=i\kappa\delta^{j}_i\delta^{(3)}(x,y)\quad\text{and}\quad\{\pi^a_A(x),\psi^B_b(y)\}=-\delta^{a}_b\delta_{A}^B\delta^{(3)}(x,y)
\label{eq:5.2.9}
\end{equation}
In the decomposition (\ref{eq:5.2.7}), $S^L_{A}$ and $S^{R A'}$ denote the so-called \emph{left} and \emph{right supersymmetry constraints} taking the form
\begin{equation}
S^L_A=D_a^{(A^+)}\pi^a_A+\frac{2}{L}E^{ai}\psi_a^B(\epsilon\tau_i)_{BA}
\label{eq:5.2.10}
\end{equation}
and 
\begin{align}
S^{R A'}=-\epsilon^{ijk}\frac{E^b_jE^c_k}{2\sqrt{q}}\sigma_i^{AA'}\left(2\epsilon_{AB}D_{[b}^{(A^+)}\psi_{c]}^B+\frac{1}{2L}\pi^a_A\epsilon_{abc}\right)
\label{eq:5.2.12}
\end{align}
respectively. Finally, the Gauss and vector constraint of the theory are given by
\begin{equation}
G_i=\frac{i}{\kappa}D^{(A^+)}_aE^a_i-\pi^a_A\tensor{(\tau_i)}{^A_B}\psi_a^B
\label{eq:5.2.13}
\end{equation}
and 
\begin{align}
H_a:=\frac{i}{\kappa}E_i^{b}F(A^+)^i_{ab}-\tensor{\epsilon}{^{bcd}}e_{a AA'}\bar{\psi}_{b}^{A'}D^{(A^+)}_{c}\psi_{d}^A
\label{eq:5.2.14}
\end{align}
respectively. The Hamiltonian constraint takes the form 
\begin{align}
H=&-\frac{E_i^aE_j^b}{2\kappa\sqrt{q}}\tensor{\epsilon}{^{ij}_k}F(A^+)^k_{ab}-\epsilon^{abc}\bar{\psi}_a^{A'}n_{AA'}D_b^{(A^+)}\psi_c^A\label{eq:5.2.15}\\
&+\frac{E^a_iE^b_j}{2L\sqrt{q}}\tensor{\mathring{\epsilon}}{^{ijk}}(\psi_{a A}n_{BA'}\sigma_k^{AA'}\psi_b^B-\bar{\psi}_a^{A'}n_{AA'}\sigma_k^{AB'}\bar{\psi}_{b B'})+\frac{3}{\kappa L^2}\sqrt{q}\nonumber
\end{align}
where $n^{AA'}$ is the spinor corresponding to the unit normal vector field $n^{\mu}$ orthogonal to the time slices $\Sigma_t$ in the 3+1-decomposition. As for the super Ashtekar connection, the canonically conjugate momenta $\pi$ and $\tilde{E}:=\frac{i}{\kappa}E$ can be combined to a super electric field 
\begin{equation}
\mathcal{E}:=\tensor[^i]{T}{^+}\otimes\tilde{E}_i+\tensor[^A]{Q}{}\otimes\pi_A
\label{eq:5.2.16}
\end{equation}
conjugate to $\mathcal{A}^{+}$ where $(\tensor[^i]{T}{^+},\tensor[^A]{Q}{})$ is a left-dual basis of $\tensor[^*]{\mathrm{Lie}(\mathcal{G})}{}$ with $\mathrm{Lie}(\mathcal{G})=\mathfrak{g}\otimes\Lambda$ the super Lie module of $\mathcal{G}:=\mathrm{OSp}(1|2)_{\mathbb{C}}$. Let us consider the left-adjoint representation 
\begin{equation}
\mathrm{ad}_L:\,\mathrm{Lie}(\mathcal{G})\rightarrow\underline{\mathrm{End}}_L(\mathrm{Lie}(\mathcal{G})):=\underline{\mathrm{Hom}}_L(\mathrm{Lie}(\mathcal{G}),\mathrm{Lie}(\mathcal{G})):\,X\mapsto[\,\cdot\,,X]
\label{eq:}
\end{equation}
Its left-dual of then defines a right-linear map $\tensor[^*]{\mathrm{ad}}{_L}:\,\mathrm{Lie}(\mathcal{G})\mapsto\underline{\mathrm{End}}_R(\tensor[^*]{\mathrm{Lie}(\mathcal{G})}{})$ on $\tensor[^*]{\mathrm{Lie}(\mathcal{G})}{}$ given by
\begin{equation}
\tensor[^*]{\mathrm{ad}}{_L}(X)(\ell):\,Y\mapsto\braket{Y|\tensor[^*]{\mathrm{ad}}{_L}(X)\diamond\ell}=\braket{[Y,X]|\ell}
\label{eq:}
\end{equation}
$\forall\ell\in\tensor[^*]{\mathrm{Lie}(\mathcal{G})}{}$. In fact, $\tensor[^*]{\mathrm{ad}}{_L}$ defines a representation of $\mathrm{Lie}(\mathcal{G})$ on $\tensor[^*]{\mathrm{Lie}(\mathcal{G})}{}$ since
\begin{align}
\braket{Z|[\tensor[^*]{\mathrm{ad}}{_L}(X),\tensor[^*]{\mathrm{ad}}{_L}(Y)]\diamond\ell}&=\braket{[[Z,X],Y]-(-1)^{|X||Y|}[[Z,Y],X]|\ell}\nonumber\\
&=\braket{[[Z,X],Y]+(-1)^{|X||Z|}[X,[Z,Y]]|\ell}\nonumber\\
&=\braket{[Z,[X,Y]]|\ell}=\braket{Z|\tensor[^*]{\mathrm{ad}}{_L}([X,Y])\diamond\ell}
\end{align}
for any $X,Y,Z\in\mathrm{Lie}(\mathcal{G})$ and $\ell\in\tensor[^*]{\mathrm{Lie}(\mathcal{G})}{}$, where, from the first to the second line, we used graded skew-symmetry of the Lie bracket and, from the second to the last line, the graded Jacobi identity (\ref{eq:Jacobi}). The covariant derivative of $\mathcal{E}$ w.r.t. the super Ashtekar connection $\mathcal{A}^+$ in the left co-adjoint representation then takes the form 
\begin{equation}
D^{(\mathcal{A}^+)}\mathcal{E}_{\alpha}=\mathrm{d}\mathcal{E}+\tensor[^*]{\mathrm{ad}}{_L}(\mathcal{A}^+)\wedge\mathcal{E}=\mathrm{d}\mathcal{E}_{\alpha}+(-1)^{|\beta|(|\beta|+|\gamma|)}\tensor{f}{_{\alpha\beta}^{\gamma}}\mathcal{A}^{+\beta}\wedge\mathcal{E}_{\gamma}
\label{eq:}
\end{equation}
where $\alpha=(i,A)$ and $\tensor{f}{_{\alpha\beta}^{\gamma}}$ denote the structure coefficients of the $\mathfrak{osp}(1|2)_{\mathbb{C}}$ Lie superalgebra defined via (\ref{eq:5.1.9.1})-(\ref{eq:5.1.9.4}) for $\mathcal{N}=1$. If one combines the Gauss constraint with the left supersymmetry constraint to the \emph{super Gauss constraint} $\mathscr{G}_{\alpha}:=(G_i,S^L_A)$ it then follows immediately from (\ref{eq:5.2.10}) and (\ref{eq:5.2.13}) as well as (\ref{eq:5.1.9.1})-(\ref{eq:5.1.9.4}) that these match exactly together to give the super Gauss law\footnote{Note that one has to replace $L$ by $L/\kappa$ in (\ref{eq:5.1.9.3}) as $\kappa$ has been absorbed in the Rarita-Schwinger field.}
\begin{equation}
\mathscr{G}=D^{(\mathcal{A}^+)}_a\mathcal{E}^a
\label{eq:}
\end{equation}
This was precisely the observation of \cite{Fulop:1993wi}. Smearing the super Gauss constraint with test functions $\Lambda^{\alpha}:=(\Lambda^i,\Lambda^A)$, it follows from the graded Poisson relations (\ref{eq:5.2.9}) that
\begin{equation}
\{\mathscr{G}[\Lambda],\mathscr{G}[\Lambda']\}=\mathscr{G}([\Lambda,\Lambda'])
\label{eq:}
\end{equation}
Hence, the super Gauss constraint indeed generates local $\mathrm{OSp}(1|2)_{\mathbb{C}}$-gauge transformations as was to be expected form our observations in the previous section due to the interpretation of $\mathcal{A}^+$ in terms of a generalized super Cartan connection. This was the starting point in \cite{Ling:1999gn,Livine:2003hn} for quantizing this theory studying super holonomies corresponding to $\mathcal{A}^+$ which we have constructed in mathematical rigorous way in section \ref{Holonomies}. The construction of the state space of the quantum theory will be considered in the following section.
\begin{remark}
Actually, in context of Yang-Mills theory, we need principal connections. Therefore, as in the previous section, using Prop. \ref{prop:4.6}, we simply lift $\mathcal{A}^+$ to the associated $\mathcal{S}$-relative $\mathrm{OSp}(1|2)_{\mathbb{C}}$-bundle $\mathcal{Q}_{/\mathcal{S}}\times_{\mathrm{SL}(2,\mathbb{C})}\mathrm{OSp}(1|2)_{\mathbb{C}}$. With respect to this connection, we can define holonomies and and construct the state space in the manifest approach to loop quantum supergravity (LQSG).
\end{remark}

Let us finally comment on the reality conditions enforced on the self-dual Ashtekar connection in order to recover ordinary real supegravity. As discussed in detail in \cite{Konsti-Kos:2020,Jacobson:1987cj}, in the canonical theory, it follows that the reality conditions are equivalent to the requirement that the 3D spin connection part $\Gamma^i:=-\tensor{\epsilon}{^i_{jk}}\omega^{jk}$ of $A^+$ satisfies the torsion equation
\begin{equation}
D^{(\Gamma)}e^i\equiv\mathrm{d}e^i+\tensor{\epsilon}{^{i}_{jk}}\Gamma^j\wedge e^k=\Theta^{(\Gamma) i}=\frac{i\kappa}{2}\psi^A\wedge\bar{\psi}^{A'}\sigma^i_{AA'}
\label{eq:2.19}
\end{equation} 
This equation has the unique solution 
\begin{equation}
\Gamma^i\equiv\Gamma^i(e)+C^i(e,\psi,\bar{\psi})
\label{eq:2.20}
\end{equation}
with $\Gamma^i(e)$ the torsion free metric connection
\begin{equation}
\Gamma^i_a(e)=-\epsilon^{ijk}e_j^b\left(\partial_{[a}e_{b]k}+\frac{1}{2}e^c_k e^l_a\partial_{[c}e_{b]l}\right)
\label{eq:2.21}
\end{equation}
and $C^i$ the \emph{contorsion tensor} given by
\begin{equation}
C^i_a=\frac{i\kappa}{4|e|}\epsilon^{bcd}e_d^i\left(2\psi^A_{[a}\bar{\psi}_{b]}^{A'}e_{c AA'}-\psi^A_{b}\bar{\psi}_{c}^{A'}e_{a AA'}\right)
\label{eq:2.22}
\end{equation}
Thus, the reality conditions for the bosonic degrees of freedom are given by
\begin{equation}
A^{+ i}_a+(A^{+ i}_a)^*=2\Gamma^i_a(e)+2C^i_a(e,\psi,\bar{\psi}),\quad E^a_i=\Re(E^a_i)
\label{eq:2.23}
\end{equation}
These ensure that, provided the initial conditions satisfy (\ref{eq:2.23}), the dynamical evolution remains in the real sector of the complex phase space, i.e., the phase space of ordinary real $\mathcal{N}=1$ supergravity.
\subsection{Invariant Haar measure on $\mathrm{OSp}(1|2)_{\mathbb{C}}$ and the state space of LQSG}\label{OSp}
In order to construct the (kinematical) Hilbert space of LQSG and to implement local $\mathrm{OSp}(1|2)_{\mathbb{C}}$-gauge transformations, we want to derive the super analog of a Haar measure on the complex super Lie group $\mathrm{OSp}(1|2)_{\mathbb{C}}$. In the literature, there exist various results in this direction in both the algebraic or concrete approach to supermanifold theory. In \cite{Coulembier:2012} for instance, invariant measures where constructed in the algebraic setting for various real super Lie groups including the super unitary groups $\mathrm{U}(m|n)$ as well the real orthosymplectic groups $\mathrm{OSp}(m|n)$. There, one uses the equivalent description of super Lie groups in terms of \emph{super Harish Chandra pairs}. In fact, it turns out that super Lie groups $\mathcal{G}$ have a relatively simple structure that they are completely determined by the respective body $\mathbf{B}(\mathcal{G})$ and the super Lie algebra $\mathfrak{g}$. However, in the algebraic setting, this correspondence remains rather implicit. \\
In the $H^{\infty}$ category (or more generally for $\mathcal{A}$-manifolds), in \cite{Tuynman:2004}, a concrete algorithm was given constructing invariant Haar measures for arbitrary (real) super Lie groups. This is based on the existence of a concrete relation between a super Lie group $\mathcal{G}$ and the data $(\mathbf{B}(\mathcal{G}),\mathfrak{g})$. More precisely, one has the following
\begin{theorem}[Super Harish-Chandra pair (after \cite{Tuynman:2004,Tuynman:2018})]
Let $\mathcal{G}$ be a $H^{\infty}$ super Lie group with body $G:=\mathbf{B}(\mathcal{G})$. Then, $\mathcal{G}$ is globally split, that is, it is diffeomorphic to split supermanifold $\mathbf{S}(\mathfrak{g}_1,G)$ associated to the trivial vector bundle $G\times\mathfrak{g}_1\rightarrow G$ via the canonical mapping
\begin{align}
\Phi:\,\mathbf{S}(\mathfrak{g}_1,G)&\rightarrow\mathcal{G}\label{eq:5.3.1}\\
(g,X)&\mapsto g\cdot\exp(X)\nonumber
\end{align}
where $\mathbf{S}(\mathfrak{g}_1,G)\cong\mathbf{S}(G)\times(\mathfrak{g}_1\otimes\Lambda)_0$. In particular, there exists a unique super Lie group structure on $\mathbf{S}(\mathfrak{g}_1,G)$ such that (\ref{eq:5.3.1}) turns into a morphism of super Lie groups. Hence, any $H^{\infty}$ super Lie group is uniquely determined via (\ref{eq:5.3.1}) by the data $(G,\mathfrak{g})$ called a \emph{super Harish Chandra pair} consisting of its body $G$ as well as the super Lie algebra $\mathfrak{g}=\mathfrak{g}_0\oplus\mathfrak{g}_1$. 
\end{theorem}
As it has been shown in \cite{Schuett:2018} in the Molotkov-Sachse approach, such a correspondence via (\ref{eq:5.3.1}) even holds in case of infinite dimensional Fréchet super Lie groups.\\
A $\mathbb{C}$-linear map $\int_{\mathcal{G}}:\,H_c^{\infty}(\mathcal{G},\mathbb{C})\rightarrow\mathbb{C}$ on $\mathcal{G}$ is called \emph{left-invariant integral} on $\mathcal{G}$ if \cite{Konsti-SPW:2020}
		\begin{equation}
	\mathds{1}\otimes\int_{\mathcal{G}}\circ\,\Theta_L^*=\mathds{1}\otimes\int_{\mathcal{G}}
	\label{eq:5.3.2}
	\end{equation}
where $\Theta_L:\,\mathcal{G}\times\mathcal{G}\rightarrow\mathcal{G}\times\mathcal{G}$ is defined as $\Theta_L(g,h)=(g,\mu_{\mathcal{G}}(g,h))$. Similarly one can define \emph{right invariant integrals} setting $\Theta_R:\,\mathcal{G}\times\mathcal{G}\rightarrow\mathcal{G}\times\mathcal{G},\,(g,h)\mapsto(\mu_{\mathcal{G}}(g,h),h)$. The reason for choosing $\Theta_L$ instead of just the group multiplication on $\mathcal{G}$ as usually done in the literature is that $\Theta_L$ is a \emph{proper map}, i.e., the preimage of compact sets in $\mathcal{G}$ is compact in $\mathcal{G}\times\mathcal{G}$ which necessary for condition (\ref{eq:5.3.2}) to be well-defined (this is true for the group multiplication only in case $\mathcal{G}$ is compact).\\
Integrals on supermanifolds can be formulated in terms of \emph{Berezinian densities} (see \cite{Leites:1980,Rothstein:1987,Varadarajan:2004} for more details). A Berezinian density on a supermanfold $\mathcal{M}$ is defined as a smooth section $\Gamma_c(\mathrm{Ber}(\mathcal{M}))$ with compact support of the \emph{Berezin line bundle}\footnote{in case of an ordinary $C^{\infty}$ manifold $M$ of dimension $n$, these can be identified with sections of the exterior bundle $\Exterior^nT^*M$, i.e., top-degree forms $\Omega^n(M)$ on $M$. This, however, is no longer true in case of supermanifolds as there a straightforward notion of a top-degree form turns out not to exist. For an alternative approach towards top-degree forms on supermanifolds using the concept of so-called \emph{integral forms} see e.g. \cite{Castellani:2014goa,Cremonini:2019aao,Catenacci:2018xsv}} $\mathrm{Ber}(\mathcal{M}):=\mathscr{F}(\mathcal{M})\times_{\mathrm{Ber}}\Lambda^{\mathbb{C}}$ which is a bundle associated to the frame bundle $\mathscr{F}(\mathcal{M})$ via the one-dimensional dual representation $\mathrm{GL}(m|n,\Lambda)\ni A\mapsto\mathrm{Ber}(A)^{-1}\in\mathrm{Aut}(\Lambda^{\mathbb{C}})$ where $\mathrm{dim}\,\mathcal{M}=(m,n)$.\\
For a super Lie group $\mathcal{G}$, a Berezinian density $\nu\in\Gamma_c(\mathrm{Ber}(\mathcal{G}))$ then induces a left-invariant integral iff its trivial extension $\hat{\nu}$ on $\mathcal{G}\times\mathcal{G}$ satisfies
\begin{equation}
\Theta_L^*\hat{\nu}=\hat{\nu}
\label{eq:5.3.3}
\end{equation}
Such a density will be called a \emph{left-invariant Haar measure} on $\mathcal{G}$. To construct such an invariant measure note that, for any super Lie group $\mathcal{G}$, the tangent bundle $T\mathcal{G}$ is always trivializable with a global frame $\mathfrak{p}\in\Gamma(\mathscr{F}(\mathcal{G}))$ induced by a homogeneous basis $(e_i,f_j)$ of left invariant vector fields $e_i,f_j\in\mathfrak{g}$ on $\mathcal{G}$, $i=1,\ldots,m$ and $j=1,\ldots,n$ with $\mathrm{dim}\,\mathcal{G}=(m,n)$. In particular, this yields a global section $\nu_{\mathfrak{g}}:=[\mathfrak{p},1]\in\Gamma(\mathrm{Ber}(\mathcal{G}))$ of the associated Berezin line bundle which, by construction, automatically defines a left-invariant Haar measure. With respect to local coordinates $(x,\xi)$, this density is of the form 
\begin{equation}
\nu_{\mathfrak{g}}=[\mathfrak{p}=(\partial_{x^i},\partial_{\xi_j})\cdot X,1]=[(\partial_{x^i},\partial_{\xi_j}),\mathrm{Ber}(X)^{-1}]
\label{eq:5.3.4}
\end{equation}
where $X$ denotes the matrix representation of the left-invariant vector fields w.r.t. the induced coordinate derivatives. To find an explicit expression for $X$, one can then use the equivalent description of $\mathcal{G}$ in terms of the corresponding super Harish-Chandra pair $(G,\mathfrak{g})$ via identification (\ref{eq:5.3.1}). This requires an intense use of the Baker-Campbell-Hausdorff formula and thus involves various powers of the (right) adjoint representation $\mathrm{ad}_R:\,\mathrm{Lie}(\mathcal{G})\rightarrow\mathrm{End}_R(\mathrm{Lie}(\mathcal{G})),\,X\mapsto[X,\cdot]$. As shown in \cite{Tuynman:2018}, the matrix representation then takes the form
\begin{equation}
X(x,\xi)=\begin{pmatrix}
	C(x) & C(x)\cdot H(\xi)\\
	A(\xi) & B(\xi)
\end{pmatrix}
\label{eq:5.3.5}
\end{equation} 
where $C(x)$ as well as $H(\xi)$, $A(\xi)$ and $B(\xi)$ are submatrices depending purely on even and odd coordinates, respectively, and which are defined via
\begin{align}
&\mathrm{ad}_R(v)(e_i)=:f_j\tensor{A(\xi)}{^j_i},\quad b_{+}(\mathrm{ad}_{R}(v))f_j=:f_k\tensor{B(\xi)}{^k_j}\nonumber\\
\text{and}\quad &h(\mathrm{ad}_{R}(v))f_j=e_i\tensor{H(\xi)}{^i_j}
\label{eq:5.3.6}
\end{align} 
with $v:=f_j\xi^j\in(\mathfrak{g}_1\otimes\Lambda)_0$ and real functions
\begin{equation}
b_{+}(t):=\frac{t\cosh(t)}{\sinh(t)}=1+\frac{1}{3}t^2-\frac{1}{45}t^4+\ldots,\quad h(t):=\frac{e^t-1}{e^t+1}=\frac{1}{2}t-\frac{1}{24}t^3+\ldots
\label{eq:5.3.6.1}
\end{equation}
Moreover, $C(x)$ is determined via the matrix representation of the even left invariant vector fields $e_i$ via
\begin{equation}
e_i|_{(g,v)}=\partial_{x_j}\tensor{C}{^j_i}(x)+\partial_{\xi_k}\tensor{A}{^k_i}(\xi)
\label{eq:5.3.7}
\end{equation}
In particular, when restricting on the body, $C$ can be identified with the matrix representation of the left-invariant vector fields on $G$. The left invariant integral on $\mathbf{S}(\mathfrak{g}_1,G)$ for smooth functions $f\in H^{\infty}_c(\mathbf{S}(\mathfrak{g}_1,G))\cong C_c^{\infty}(G)\otimes\Exterior\mathfrak{g}_1^{*}$ then takes the form 
\begin{align}
\int_{\mathcal{G}}{f\nu_{\mathfrak{g}}}&=\int{\mathrm{d}^nx\int_B{\mathrm{d}^n\xi\,f(x,\xi)\mathrm{Ber}(X)^{-1}(x,\xi)}}\nonumber\\
&=\int{\mathrm{d}^nx\,C(x)^{-1}\int_B{\mathrm{d}^n\xi\,\frac{\det B(\xi)}{\det(\mathds{1}-H(\xi)B(\xi)^{-1}A(\xi))}f(x,\xi)}}\nonumber\\
&=:\int_{G}{\mathrm{d}\mu_H(g)\int_B{\mathrm{d}^n\xi\,\Delta(\xi)f(g,\xi)}}
\label{eq:5.3.8}
\end{align}
where $\mu_H$ is the induced left-invariant Haar measure on $G$ and $\int_B$ denotes the usual \emph{Berezin integral} on $\Exterior\mathfrak{g}_1^{*}$. Hence, the derivation of the invariant integral on $\mathcal{G}$ boils down to the choice of an invariant Haar measure on the body $G$ as well as the derivation of the density $\Delta(\xi)$ in the Berezin integral which, according to (\ref{eq:5.3.6}) and (\ref{eq:5.3.6.1}), only involves the computation of the matrix representation of the adjoint representation on $\mathfrak{g}$.\\
\\
Let us apply this algorithm to compute the invariant Haar measure on the complex orthosymplectic group $\mathrm{OSp}(1|2)_{\mathbb{C}}$ (in case of the real orthosymplectic group, this has been done explcitly already in \cite{Tuynman:2018}). Therefore, let us introduce a homogeneous basis $(e_1,e_2,e_3,f_1,f_2)$ on $\mathfrak{osp}(1|2)_{\mathbb{C}}$ defining  
\begin{alignat}{3}
&e_1:=2iT_{3}^{+},\quad && e_{2}:=i(T_{1}^{+}+i T_{2}^{+}),\quad&& e_{3}:=i(T_{1}^{+}-i T_{2}^{+})\\
&f_1=\frac{1}{\sqrt{2}}(i-1)Q_{+},\quad && f_2:=\frac{1}{\sqrt{2}}(1-i)Q_{-}
\label{eq:5.3.9}
\end{alignat}
Setting $L=\frac{1}{2}$, it then follows from (\ref{eq:5.1.9.1})-(\ref{eq:5.1.9.4}) that the commutators among the even generators of $\mathfrak{osp}(1|2)_{\mathbb{C}}$ satisfy
\begin{equation}
[e_1,e_2]=2e_2,\quad [e_1,e_3]=-2e_3,\quad [e_2,e_3]=e_1
\label{eq:}
\end{equation}
which are the standard commutation relations of $\mathfrak{sl}(2,\mathbb{C})$. For the mixed commutators between even and odd generators, we find
\begin{alignat}{3}
&[e_1,f_1]=f_1,\quad && [e_2,f_1]=0,\quad && [e_3,f_1]=-f_2\\
&[e_1,f_2]=-f_2,\quad && [e_2,f_2]=-f_1,\quad && [e_3,f_2]=0
  \end{alignat}
and, finally, the anti commutators between odd generators yield 
\begin{align}
[f_1,f_1]=-2e_2,\quad [f_2,f_2]=2e_3,\quad [f_1,f_2]=-e_1
\label{eq:}
\end{align}
These are precisely the graded commutation relations of \emph{real} $\mathrm{OSp}(1|2)$. Similar as in \cite{Tuynman:2018}, using (\ref{eq:5.3.1}), we can then identify $\mathrm{OSp}(1|2)_{\mathbb{C}}$ with the split super Lie group $\mathbf{S}(\mathfrak{osp}(1|2)_{1},\mathrm{SL}(2,\mathbb{C}))$ according to
\begin{equation}
\Phi(g,v)=g\exp(v)=g\cdot\begin{pmatrix}
	1-\xi\eta & \eta & \xi\\
	\xi & 1+\frac{1}{2}\xi\eta & 0\\
	-\eta & 0 & 1+\frac{1}{2}\xi\eta
\end{pmatrix}
\label{eq:5.3.10}
\end{equation}
for $g\in\mathrm{SL}(2,\mathbb{C})$ and $v:=\xi f_1+\eta f_2\in(\mathfrak{osp}(1|2)_{1}\otimes\Lambda^{\mathbb{C}})_0$. Using the commutation relations above, it follows immediately that the matrix representation of $\mathrm{ad}(v)$ is given by 
\begin{equation}
\mathrm{ad}(v)=\begin{pmatrix}
	0 & 0 & 0 & -\eta & \xi\\
	0 & 0 & 0 & -2\xi & 0\\
	0 & 0 & 0 & 0 & 2\eta\\
	\xi & -\eta & 0 & 0 & 0\\
	-\eta & 0 & -\xi & 0 & 0
\end{pmatrix}
\label{eq:5.3.10}
\end{equation}
Hence, it follows from (\ref{eq:5.3.6}) as well as (\ref{eq:5.3.7}), using $\mathrm{ad}(v)^n=0$ for $n\geq 3$, \cite{Tuynman:2018}
\begin{equation}
\begin{pmatrix}
	\mathds{1} & H\\
	A & B
\end{pmatrix}=\begin{pmatrix}
	1 & 0 & 0 & -\frac{1}{2}\eta & -\frac{1}{2}\xi\\
	0 & 1 & 0 & \xi & 0\\
	0 & 0 & 1 & 0 & \eta\\
	\xi & -\eta & 0 & 1-\xi\eta & 0\\
	-\eta & 0 & -\xi & 0 & 1-\xi\eta
\end{pmatrix}
\label{eq:5.3.11}
\end{equation}
Actually, for the derivation of (\ref{eq:5.3.8}), it has been implicitly assumed that the super Lie group is real. Hence, we need to view $\mathrm{OSp}(1|2)_{\mathbb{C}}$ as a real super Lie group. A homogeneous basis of the realification of $\mathfrak{g}:=\mathfrak{osp}(1|2)_{\mathbb{C}}$ (resp. $\mathrm{Lie}(\mathcal{G}):=\mathfrak{g}\otimes\Lambda^{\mathbb{C}}$) is then given by $(e_i,ie_i,f_j,if_j)$. Let $\mathcal{R}:\,\underline{\mathrm{End}}_R(\mathrm{Lie}(\mathcal{G}))\rightarrow\underline{\mathrm{End}}_R(\mathrm{Lie}(\mathcal{G})_{\mathbb{R}})$ be the morphism which identifies any $X\in\underline{\mathrm{End}}_R(\mathrm{Lie}(\mathcal{G}))$ with the corresponding real endomorphism $\mathcal{R}(A)$ on the realification $\mathrm{Lie}(\mathcal{G})_{\mathbb{R}}$. For the density $\Delta\equiv\Delta(\xi,\bar{\xi},\eta,\bar{\eta})$ in the Berezin integral we then compute
\begin{equation}
\Delta=\frac{\det(\mathcal{R}(B))}{\det(\mathcal{R}(\mathds{1}-H\cdot B^{-1}\cdot A))}=(1+\xi\eta)(1+\bar{\xi}\bar{\eta})
\label{eq:5.3.12}
\end{equation}
Hence, going back to the original basis $(T_i^+,Q_A)$ of the super Lie algebra $\mathfrak{osp}(1|2)_{\mathbb{C}}$, we find the invariant integral for any smooth function $f\in V:=C_c^{\infty}(\mathrm{SL}(2,\mathbb{C}),\mathbb{C})\otimes\Exterior[\theta^A,\bar{\theta}^{A'}]$ is given by
\begin{equation}
\int_{\mathrm{OSp}(1|2)_{\mathbb{C}}}{f\nu_{\mathfrak{g}}}=\int_{\mathrm{SL}(2,\mathbb{C})}{\mathrm{d}\mu_H(g)\int_B{\mathrm{d}\mu(\theta,\bar{\theta})\,f(g,\theta,\bar{\theta})}}
\label{eq:5.3.13}
\end{equation}
with $\mathrm{d}\mu_H$ an invariant Haar measure on $\mathrm{SL}(2,\mathbb{C})$ and 
\begin{equation}
\mathrm{d}\mu(\theta,\bar{\theta}):=\mathrm{d}\theta^A\mathrm{d}\bar{\theta}^{A'}\left(1-\frac{i}{2}\theta^A\theta_A\right)\left(1+\frac{i}{2}\bar{\theta}^{A'}\bar{\theta}_{A'}\right)
\label{eq:5.3.14}
\end{equation}
This Haar measure induces a super scalar product $\mathscr{S}:\,V\times V\rightarrow\mathbb{C}$, i.e., a graded-symmetric and non-degenerate sesquilinear form on the super vector space $V$ (cf. definition \ref{bilinear map}) via
\begin{equation}
\mathscr{S}(f,g):=\int_{\mathrm{OSp}(1|2)_{\mathbb{C}}}{\bar{f}g\,\nu_{\mathfrak{g}}}
\label{eq:5.3.15}
\end{equation}
By construction, $\mathscr{S}$ is indefinite turning $(V,\mathscr{S})$ into an indefinite inner product space. In \cite{Tuynman:2018}, it has been shown for arbitrary (real) super Lie groups that one can always find a linear map $J:\,V\rightarrow V$ such that $\braket{\cdot,\cdot}:=\mathscr{S}(\cdot,J\cdot)$ defines a positive definite scalar product on $V$. Hence, $(V,\mathscr{S},J)$ has the structure of a \emph{Krein space}. Given such an endomorphism $J$, one can then use the induced scalar product to complete $V$ to a Hilbert space 
\begin{equation}
\mathcal{H}:=\overline{V}^{\|\cdot\|}
\label{eq:5.3.16}
\end{equation}
However, the choice of such an endomorphism is, a priori, not unique and thus needs to be fixed by additional physical requirements. Therefore, as a next step, let us turn to the construction of the kinematical state space in the manifest approach to LQSG. However, due to the non-compactness of $\mathrm{SL}(2,\mathbb{C}) $, even in the purely bosonic theory, it is still unclear how the (kinematical) Hilbert space with self-dual variables can be  defined consistently. Hence, the following considerations will only sketch some ideas following the stanndard procedure in LQG with real variables.\\
\\
In the standard quantization scheme in loop quantum gravity, one considers graphs $\gamma\subset\Sigma$ that are embedded into spatial slices $\Sigma$ of the spacetime manifold $\mathcal{M}_0$ which can be regarded as a purely bosonic supermanifold. This is also reasonable in case of supersymmetry due to the rheonomy principle stating that physical degrees of freedom are completely encoded on the body \cite{DAuria:1982uck,Castellani:1991et}. Hence, we consider graphs $\gamma$ that are embedded into spatial slices $\Sigma$ of the bosonic submanifold $\mathcal{M}_0\subset\mathcal{M}$ of the underlying supermanifold $\mathcal{M}$. Let us fix such a particular graph $\gamma$ consisting of $n$ edges $e_i\in\gamma$, $i=1,\ldots,n$, that intersect at certain vertices $v$ and require that the graph is suitably refined such that the topology of $\Sigma$ can be resolved \cite{Achour:2015zmk}. A \emph{cylindrical function} $\Psi_{\gamma}[\mathcal{A}^{+}]$ on that graph will then generically be of the form
\begin{equation}
\Psi_{\gamma}:=f_{\gamma}(h_{e_1}[\mathcal{A}^{+}],\ldots,h_{e_n}[\mathcal{A}^{+}])
\label{eq:5.3.17}
\end{equation}
for some smooth function $f_{\gamma}\in H^{\infty}_{(c)}(\mathrm{OSp}(1|2)^{\times n}_\mathbb{C},\mathbb{C})\cong \bigotimes_{i=1}^n H^{\infty}_{(c)}(\mathrm{OSp}(1|2)_{\mathbb{C}},\mathbb{C})$. 
In section \ref{Holonomies}, choosing a particular gauge, it has been shown that the holonomies $h_{e}[\mathcal{A}^+]$ associated to the super Ashtekar connection are of the form\footnote{As discussed in the end of section \ref{Holonomies}, if the parametrizing supermanifold $\mathcal{S}$ is chosen to be a superpoint corresponding to a Grassmann algebra $\Lambda$, the holonomies can be identified with group elements of $\mathrm{OSp}(1|2)_{\mathbb{C}}$ regarded as a $H^{\infty}$ super Lie group modeled over $\Lambda$.}
\begin{equation}
h_{e}[\mathcal{A}^+]=h_e[A^+]\cdot\mathcal{P}\mathrm{exp}\left(-\int_{e}{\mathrm{Ad}_{h_{e}[A^+]^{-1}}\psi}\right)
\label{eq:5.3.18}
\end{equation}
with $h_e[A^+]$ the induced holonomies of the (bosonic) self-dual Ashtekar connection $A^+$. According to (\ref{eq:5.1.9.2}), the adjoint representation of $\mathrm{SL}(2,\mathbb{C})$ on the odd part of the Lie superalgebra $\mathfrak{osp}(1|2)_{\mathbb{C}}$ is given by the fundamental respresentation such that 
\begin{equation}
\mathrm{Ad}_{h_{e}[A^+]^{-1}}\psi=Q_B\,\tensor{(h_{e}[A^+]^{-1})}{^B_A}\psi^A
\label{eq:5.3.19}
\end{equation}
From (\ref{eq:5.3.18}) and (\ref{eq:5.3.19}) it follows that holonomies of $\mathcal{A}^+$ are holomorphic functions in the fermionic variables $\psi^A$. Hence, in order to resolve the physical degrees of freedom, it is sufficient to restrict definition (\ref{eq:5.3.17}) of cylindrical functions to functions $f_{\gamma}$ of the form 
\begin{equation}
f_{\gamma}\in V^{\otimes n}:=\bigotimes_{i=1}^n C^{\infty}_{(c)}(\mathrm{SL}(2,\mathbb{C}),\mathbb{C})\otimes\Exterior[\theta^A]
\label{eq:5.3.20}
\end{equation}
According to (\ref{eq:5.3.17}) and (\ref{eq:5.3.20}), we can identify the space of cylindrical functions w.r.t. $\gamma$ with $V^{\otimes n}$. The invariant Haar measure on $\mathrm{OSp}(1|2)_{\mathbb{C}}$ as defined via (\ref{eq:5.3.13}) can be extended to invariant Haar measure on $\mathrm{OSp}(1|2)_{\mathbb{C}}^{\times n}$ taking the tensor product (as may be checked by direct computation, this will automatically satisfy \ref{eq:5.3.3} (resp. (\ref{eq:5.3.2}))). This, in turn, induces a super scalar product $\mathscr{S}$ on $V^{\otimes n}$ which again has the structure of a Krein space. In order to complete to a Hilbert space and construct a positive definite inner product, one has to choose a particular endomorphism $J:V^{\otimes n}\rightarrow V^{\otimes n}$. Therefore, note that the fermionic component $\mathcal{E}_A$ of the super electric field, when implemented as an operator on the resulting Hilbert space, needs to satisfy the quantum version of the reality condition (\ref{eq:5.2.8}). Hence, $J:\,V^{\otimes n}\rightarrow V^{\otimes n}$ has to be chosen such that (\ref{eq:5.2.8}) can be implemented on the kinematical Hilbert space $\mathcal{H}_{\mathrm{kin}}$ obtained via the  completion of $(V^{\otimes n},\braket{\cdot,\cdot})$ w.r.t. to the positive definite scalar product $\braket{\cdot,\cdot}:=\mathscr{S}(\cdot,J\cdot)$. This is in fact in accordance with \cite{Ashtekar:1991hf} where it is suggested that the inner product (for the bosonic degrees of freedom) has to be defined in such way that the reality conditions can be consistently implemented. \\
\\
As discussed in detail in \cite{Tuynman:2018}, there exists an endomorphism $J:\,V\rightarrow V$ such that, w.r.t. functions $f\in V$ written in the form 
\begin{equation}
f=f_0+f_A\psi^A+\frac{1}{2}f_2\psi^A\psi_A
\label{eq:5.3.21}
\end{equation}
with $f_0,f_A,f_2\in C_{(c)}^{\infty}(\mathrm{SL}(2,\mathbb{C}),\mathbb{C})$, the induced scalar product is given by
\begin{equation}
\braket{f,g}=\braket{\braket{f_0,g_0}}+\braket{\braket{f_+,g_+}}+\braket{\braket{f_-,g_-}}+\braket{\braket{f_2,g_2}}
\label{eq:5.3.22}
\end{equation}
with $\braket{\braket{\cdot,\cdot}}$ the scalar product on $C^{\infty}_{(c)}(\mathrm{SL}(2,\mathbb{C}))$ induced by the invariant Haar measure $\mu_H$ on $\mathrm{SL}(2,\mathbb{C})$. This scalar product is then invariant under local $\mathrm{SL}(2,\mathbb{C})$-gauge transformations.\\
As it is shown in \cite{Konsti-Kos:2020}, in the framework of symmetry reduced models, this in fact solves the reality conditions (\ref{eq:5.2.8}). This is also precisely the standard scalar product as used for instance in \cite{Thiemann:1997rq} in context of Dirac fermions or in \cite{Bodendorfer:2011pb} in context of the Rarita-Schwinger field in a complementary approach to LQSG with real variables.\\
\\
Considering the holonomies (\ref{eq:5.3.18}), this may suggest to restrict $C^{\infty}_{(c)}(\mathrm{SL}(2,\mathbb{C}))$ to holomorphic functions on $\mathrm{SL}(2,\mathbb{C})$. However, due to Liouville's theorem, if required to be nontrivial, general functions of this kind cannot be of compact support. This is of course problematic. Hence, either one does not restrict to holomorphic functions considering more general smooth functions on $\mathrm{SL}(2,\mathbb{C})$ (regarded as a real manifold) or the measure on $\mathrm{SL}(2,\mathbb{C})$ is changed appropriately. The last possibility has been studied in \cite{Wilson-Ewing:2015lia} and \cite{Konsti-FB:2020} in the context of symmetry reduced models. There, it was shown that this indeed allows an exact implementation of the remaining reality conditions (\ref{eq:2.23}) where the measure turned out to be distributional. How this can be extended to the full theory, however, is still unclear and remains a task for future investigations.\\
\\
Another proposal is to consider analytic continuations from real to imaginary Barbero-Immiri parameters (see for instance \cite{Frodden:2012dq,Han:2014xna,Achour:2014eqa,Achour:2015zmk,Bodendorfer:2013hla} and references therein for recent advances in this direction). This is based on the fact that $\mathrm{SL}(2,\mathbb{C})$ contains $\mathrm{SU}(2)$ as a compact real form. Perhaps, these kind considerations can be generalized to the present situation in context of supersymmetry with gauge group $\mathrm{OSp}(1|2)_{\mathbb{C}}$.\\
Nevertheless, in the present article, we have demonstrated that structure of the (kinematical) Hilbert space for the fermionic degrees of freedom as normally used in LQG has a very natural explanation in the framework of the manifest approach to LQSG. It would thus be very interesting to further analyze this structure and even generalize it including extended supersymmetries. 
\section{Symmetry reduction in supersymmetric field theories}\label{LQC}
In this last section, we want to generalize the notion of invariant connection forms to the super category. This will provide a solid basis for the construction of (spatially) symmetry reduced models in the context supersymmetric field theories. We will use these results in \cite{Konsti-Kos:2020} to study minisuperspace models in the framework of loop quantum cosmology with local supersymmetry.

Let $\mathcal{M}$ be a $H^{\infty}$ supermanifold and $\mathcal{S}$ be a super Lie group which in the most situations of interest will be the super Lie group of isometries of $\mathcal{M}$ (in fact, in most cases $\mathcal{M}$ will be an ordinary smooth manifold). Suppose, $\mathcal{S}$ acts from the left on $\mathcal{M}$, i.e., there exists a smooth map 
\begin{equation}
f:\,\mathcal{S}\times\mathcal{M}\rightarrow\mathcal{M}
\label{eq:6.1}
\end{equation}
such that
\begin{equation}
f\circ(\mathrm{id}\times f)=f\circ(\mu_{\mathcal{S}}\times\mathrm{id})\quad\text{and}\quad f_e(x)=x\quad\forall x\in\mathcal{M}
\label{eq:6.2}
\end{equation}
Furthermore, we assume that $\mathcal{S}$ acts transitively on $\mathcal{M}$. Hence, if $x\in\mathcal{M}$ is a body point and $\mathcal{S}_x$ is the stabilizer subgroup of $\mathcal{S}$, one can identify $\mathcal{M}\cong\mathcal{S}/\mathcal{S}_x$ which we shall do in what follows. The left action of $\mathcal{S}$ is then given by its standard action on the coset space $\mathcal{S}/\mathcal{S}_x$ which still will be denoted by $f$. Let $\mathcal{G}\stackrel{\pi}{\rightarrow}\mathcal{P}\rightarrow\mathcal{S}/\mathcal{S}_x$ be a principal super fiber bundle over $\mathcal{S}/\mathcal{S}_x$ with structure group $\mathcal{G}$ and $\mathcal{G}$-right action $\Phi:\,\mathcal{P}\times\mathcal{G}\rightarrow\mathcal{P}$. We want to the ask the question about the existence of an $\mathcal{S}$-left action $\hat{f}:\,\mathcal{S}\times\mathcal{P}\rightarrow\mathcal{P}$ on $\mathcal{P}$ such that $\hat{f}$ is $\mathcal{G}$-equivariant bundle automorphism on $\mathcal{P}$ projecting to the left multplication of $\mathcal{S}$ on $\mathcal{S}/\mathcal{S}_x$, i.e.,
\begin{equation}
\hat{f}\circ(\mathrm{id}_{\mathcal{S}}\times\Phi)=\Phi\circ(\hat{f}\times\mathrm{id}_{\mathcal{G}})
\label{eq:6.3}
\end{equation}
and $\pi\circ\hat{f}=f\circ(\mathrm{id}_{\mathcal{S}}\times\pi)$. Therefore, applying the forgetful functor, we consider the set of abstract group homomorphisms $\lambda:\,\mathcal{S}_x\rightarrow\mathcal{G}$. On this set, we introduce the equivalence relation 
\begin{equation}
\lambda\sim\lambda':\Leftrightarrow\exists g\in\mathcal{G}:\,\lambda'=\mathrm{Ad}_{g}\circ\lambda
\label{eq:}
\end{equation}
which yields the set of conjugacy classes $\mathrm{Conj}(\mathcal{S}_x\rightarrow\mathcal{G})$ of abstract group homomorphisms. An equivalence class $[\lambda]\in\mathrm{Conj}(\mathcal{S}_x\rightarrow\mathcal{G})$ will be called \emph{smoothly admissible}, if it contains a $H^{\infty}$-smooth super Lie group homomorphism as a representative. The set of such smoothly admissible conjugacy classes yields a proper subset $\mathrm{Conj}(\mathcal{S}_x\rightarrow\mathcal{G})_{\infty}\subset\mathrm{Conj}(\mathcal{S}_x\rightarrow\mathcal{G})$.
\begin{prop}\label{prop1}
There is a bijective correspondence between equivalence classes of principal $\mathcal{G}$-bundles over $\mathcal{S}/\mathcal{S}_x$ admitting an $\mathcal{S}$-left action which is $\mathcal{G}$-equivariant and projects to the multiplication of $\mathcal{S}$ on the coset space $\mathcal{S}/\mathcal{S}_x$ and smoothly admissible conjugacy classes $[\lambda]\in\mathrm{Conj}(\mathcal{S}_x\rightarrow\mathcal{G})_{\infty}$ of group homomorphisms $\lambda:\,\mathcal{S}_x\rightarrow\mathcal{G}$.
\end{prop}
\begin{proof}
Suppose $\lambda:\,\mathcal{S}_x\rightarrow\mathcal{G}$ is a smooth representative of a smoothly admissible conjugacy class of super Lie group homomorphisms. Consider then the associated principal supper fiber bundle $\mathcal{S}\times_{\lambda}\mathcal{G}$ with structure group $\mathcal{G}$. On $\mathcal{S}\times\mathcal{G}$, we define the smooth left action
\begin{equation}
\mathcal{S}\times(\mathcal{S}\times\mathcal{G})\rightarrow\mathcal{S}\times\mathcal{G},\quad (\phi,(\psi,g))\mapsto(\phi\circ\psi,g)
\label{eq:6.4}
\end{equation}
Since $(\phi\circ(\psi\circ\phi'),\lambda(\phi')^{-1}(g))=((\phi\circ\psi)\circ\phi',\lambda(\phi')^{-1}(g))$ $\forall\phi,\phi',\psi\in\mathcal{S}$ and $g\in\mathcal{G}$, it follows that (\ref{eq:6.4}) is constant on $\mathcal{G}$-orbits so that (\ref{eq:6.4}) induces a well-defined smooth $\mathcal{S}$-left action on $\mathcal{S}\times_{\lambda}\mathcal{G}$ which is $\mathcal{G}$-equivariant and projects to the multiplication of $\mathcal{S}$ on $\mathcal{S}/\mathcal{S}_x$.\\
Conversely, let $\hat{f}:\,\mathcal{S}\times\mathcal{P}\rightarrow\mathcal{P}$ be a $\mathcal{S}$-left action on $\mathcal{P}$. Let $p\in\mathbf{B}(\mathcal{P})$ be an element of the body. Since the $\mathcal{G}$-right action on $\mathcal{P}$ is transitive on each fiber and $\hat{f}$ is fiber preserving, for any $\phi\in\mathcal{S}$, there exists $\lambda(\phi)\in\mathcal{G}$ such that 
\begin{equation}
\hat{f}_{\phi}(p)=\Phi_{\lambda(\phi)}(p)
\label{eq:6.5}
\end{equation} 
Moreover, since $p\in\mathbf{B}(\mathcal{P})$, the map $\mathcal{S}\rightarrow\mathcal{P},\,\phi\mapsto f_{\phi}(p)$ is of class $H^{\infty}$ proving that $\lambda:\,\mathcal{S}_x\rightarrow\mathcal{G}$ is smooth. By $\mathcal{G}$-equivariance (\ref{eq:6.3}), it follows for $\phi,\psi\in\mathcal{S}$ 
\begin{align}
\hat{f}_{\phi\circ\psi}(p)&=f_{\phi}(f_{\psi}(p))=f_{\phi}(\Phi_{\lambda(\psi)}(p))=\Phi_{\lambda(\psi)}(f_{\phi}(p))=\Phi_{\lambda(\phi)\circ\lambda(\psi)}(p)\nonumber\\
&=\Phi_{\lambda(\phi\circ\psi)}(p)
\label{eq:6.6}
\end{align}
implying $\lambda(\phi\circ\psi)=\lambda(\phi)\circ\lambda(\psi)$, i.e.. $\lambda$ is indeed a super Lie group homomorphism. If $p'\in\mathcal{P}$ is any other point, then, again by transitivity, there exists $g\in\mathcal{G}$ with $\Phi_{g}(p)=p'$. Hence,
\begin{equation}
\hat{f}_{\phi}(p')=\Phi_{g}(\hat{f}(p))=\Phi_{\mathrm{Ad}_{h^{-1}}\lambda(\phi)}(p')
\label{eq:6.7}
\end{equation}
with $\mathrm{Ad}_{h^{-1}}\circ\lambda$ in the same equivalence class as $\lambda$. Finally, let $\mathcal{S}\times_{\lambda}\mathcal{G}$ be the associated principal $\mathcal{G}$-bundle  with smooth $\mathcal{S}$-left action as constructed in the first part of this proof. For $p\in\mathcal{P}$ a body point, consider the map
\begin{equation}
\mathcal{S}\times_{\lambda}\mathcal{G}\rightarrow\mathcal{P},\quad [\phi,g]\mapsto \Phi_{g}(f_{\phi}(p))
\label{eq:6.8}
\end{equation} 
By (\ref{eq:6.3}), it follows immediately that (\ref{eq:6.8}) is well-defined and in fact yields an isomorphism principal bundles.
\end{proof}
Proposition (\ref{prop1}) provides a complete classification of principal fiber bundles admitting such a smooth left action by equivalence classes of smooth super Lie group morphisms $\lambda:\,\mathcal{S}_x\rightarrow\mathcal{G}$. We next want to study connections on $\mathcal{P}$ that are invariant under this left action. Since, $\mathcal{M}$ is typically an ordinary smooth manifold and we would like to include fermionic degrees of freedom in our discussion, we go over to the category of relative supermanifolds. Hence, let us fix a parametrizing supermanifold $\mathcal{T}$. We lift all objects and morphisms to the relative category in the obvious way. If $\mathcal{P}:=\mathcal{S}\times_{\lambda}\mathcal{G}$ is a principal super fiber bundle as in Prop. (\ref{prop1}), it follows that $\mathcal{P}_{/\mathcal{T}}=\mathcal{S}_{/\mathcal{T}}\times_{\lambda}\mathcal{G}$. A $\mathcal{T}$-relative super connection 1-form $\mathcal{A}\in\Omega^1(\mathcal{P}_{/\mathcal{T}},\mathfrak{g})_0$ will be called \emph{$\mathcal{S}$-invariant}, if 
\begin{equation}
(\hat{f}_{\mathcal{T}})^*_{\phi}\mathcal{A}=\mathcal{A}\quad\forall\phi\in\mathcal{S}
\label{eq:6.9}
\end{equation}
with $\hat{f}_{\mathcal{T}}:\,\mathcal{S}\times\mathcal{P}_{/\mathcal{T}}\rightarrow\mathcal{P}_{/\mathcal{T}}$ the lift of the left multiplication $\mathcal{S}\times\mathcal{S}_{/\mathcal{T}}\rightarrow\mathcal{S}_{/\mathcal{T}}:\,(\phi,(t,\psi))\mapsto(t,\phi\circ\psi)$ to a left action on $\mathcal{P}_{/\mathcal{T}}$ (see (\ref{eq:6.4}) in the proof of prop. (\ref{prop1})).
\begin{prop}
Let $\mathcal{P}:=\mathcal{S}\times_{\lambda}\mathcal{G}$ be the associated principal super fiber bundle induced by a smooth super Lie group homomorphism $\lambda:\,\mathcal{S}_x\rightarrow\mathcal{G}$.\\
The $\mathcal{S}$ invariant $\mathcal{T}$-relative super connection 1-forms on $\mathcal{P}_{/\mathcal{T}}$ are in one-to-one correspondence to smooth even maps $\Lambda\in H^{\infty}(\mathcal{\mathcal{T}},\mathrm{Hom}_L(\mathrm{Lie}(\mathcal{S}),\mathrm{Lie}(\mathcal{G})))_0$ from the parametrizing supermanifold $\mathcal{T}$ to super Lie algebra homomorphisms $\mathrm{Hom}_L(\mathrm{Lie}(\mathcal{S}),\mathrm{Lie}(\mathcal{G}))$ satisfying
\begin{equation}
\Lambda(t)\big{|}_{\mathrm{Lie}(\mathcal{S}_x)}=\lambda_{*}
\label{eq:6.10}
\end{equation}
and 
\begin{equation}
\mathrm{Ad}_{\phi^{-1}}\diamond\Lambda(t)=\Lambda(t)\diamond\mathrm{Ad}_{\lambda(\phi)^{-1}}\quad\text{on }\mathrm{Lie}(\mathcal{S})
\label{eq:6.11}
\end{equation}
$\forall t\in\mathcal{T}$ and $\phi\in\mathcal{S}_x$.
\end{prop}
\begin{proof}
Suppose $\mathcal{A}\in\Omega^1(\mathcal{P}_{/\mathcal{T}},\mathfrak{g})_0$ is a $\mathcal{S}$ invariant $\mathcal{T}$-relative super connection 1-form. Let $\iota:\,\mathcal{S}_{/\mathcal{T}}\rightarrow\mathcal{P}_{/\mathcal{T}}$ be the smooth map as defined above. Consider then $\mathcal{A}_{\mathcal{S}}:=\iota^*\mathcal{A}\in\Omega^1(\mathcal{S}_{/\mathcal{T}},\mathfrak{g})_0$. Since $\iota\circ f=\hat{f}\circ(\mathrm{id}\times\iota)$, it follows from the $\mathcal{S}$-invariance of $\mathcal{A}$ that
\begin{equation}
f_{\phi}^*\mathcal{A}_{\mathcal{S}}=\iota^*(\hat{f}_{\phi}^*\mathcal{A})=\iota^*\mathcal{A}=\mathcal{A}_{\mathcal{S}}
\label{eq:6.12}
\end{equation}
i.e., $\mathcal{A}_{\mathcal{S}}$ is left invariant w.r.t. to the standard multiplication on $\mathcal{S}$. As a consequence, $\mathcal{A}_{\mathcal{S}}$ is uniquely determined by its restriction $\mathcal{A}_{\mathcal{S}}|_{T_e\mathcal{S}}:\,\mathcal{T}\times T_{e}\mathcal{S}\rightarrow\mathrm{Lie}(\mathcal{G})$. As this map is left linear in the second argument it follows from lemma IV.3.17 in \cite{Tuynman:2004} that it defines an even smooth map $\mathcal{A}_{\mathcal{S}}|_{T_e\mathcal{S}}:\,\mathcal{T}\rightarrow\mathrm{Hom}_L(\mathrm{Lie}(\mathcal{S}),\allowbreak\mathrm{Lie}(\mathcal{G}))$. Moreover, since the Maurer-Cartan form $\theta^{(\mathcal{S})}_{\mathrm{MC}}|_{T_e\mathcal{S}}:\,T_{e}\mathcal{S}\rightarrow\mathrm{Lie}(\mathcal{S})$ on $T_{e}\mathcal{S}$ is the identity, it follows that
\begin{equation}
\mathcal{A}_{\mathcal{S}}(t)=\theta^{(\mathcal{S})}_{\mathrm{MC}}\diamond\Lambda(t)\quad\forall t\in\mathcal{T}
\label{eq:6.13}
\end{equation} 
on $T_{e}\mathcal{S}$ for some smooth map $\Lambda\in H^{\infty}(\mathcal{T},\mathrm{Hom}_L(\mathrm{Lie}(\mathcal{S}),\allowbreak\mathrm{Lie}(\mathcal{G})))_0$. By left-invariance, it follows that (\ref{eq:6.13}) indeed holds on all of $\mathcal{S}$.\\
Remains to proof that $\Lambda$ satisfies the properties (\ref{eq:6.10}) and (\ref{eq:6.11}) of the proposition. Therefore, for $X\in\mathrm{Lie}(\mathcal{S}_x)$, we compute
\begin{align}
\iota_{*}(\mathds{1}\otimes X)_{p}&=D_{(p,e_{\mathcal{G}})}\hat{\pi}(0_t,X_{\phi},0_{e_{\mathcal{G}}})=D_{(p,e_{\mathcal{G}})}\hat{\pi}(0_{p},\lambda_{*}(X))\nonumber\\
&=\widetilde{\lambda_{*}(X)}_{[p,e_{\mathcal{G}}]}
\label{eq:6.14}
\end{align}
$\forall p=(t,\phi)\in\mathcal{S}_{/\mathcal{T}}$, where in the second equality we used that the kernel of $\hat{\pi}_{*}$ is given by 
\begin{equation}
\mathrm{ker}\,D_{(p,g)}\hat{\pi}=\{(\mathds{1}\otimes Y_{p},-R_{g*}\lambda_{*}(X))|Y\in\mathrm{Lie}(\mathcal{S}_x)\}
\label{eq:6.15}
\end{equation}
Using (\ref{eq:6.14}), this yields
\begin{align}
\lambda_{*}(X)&=\braket{\widetilde{\lambda_{*}(X)}|\mathcal{A}}=\braket{(\mathds{1}\otimes X)_{p}|\mathcal{A}_{\mathcal{S}}}\nonumber\\
&=\braket{\braket{X_{\phi}|\theta_{\mathrm{MC}}^{(\mathcal{S})}}|\Lambda(t)}=\braket{X|\Lambda(t)}
\label{eq:6.16}
\end{align}
$\forall X\in\mathrm{Lie}(\mathcal{S}_x)$. Finally, since $\iota\circ(\hat{R}_{\mathcal{T}})_{\phi}=(\hat{\Phi}_{\mathcal{T}})_{\lambda(\phi)}\circ\iota$, it follows
\begin{align}
\braket{X|\mathrm{Ad}_{\phi^{-1}}\diamond\Lambda(t)}&=\braket{\mathrm{Ad}_{\phi^{-1}}\braket{X|\theta_{\mathrm{MC}}^{(\mathcal{S})}}|\Lambda(t)}=\braket{\braket{R_{\phi*}X|\theta_{\mathrm{MC}}^{(\mathcal{S})}}|\Lambda(t)}\nonumber\\
&=\braket{R_{\phi*}X|\theta_{\mathrm{MC}}^{(\mathcal{S})}\diamond\Lambda(t)}=\braket{R_{\phi*}|\mathcal{A}_{\mathcal{S}}(t)}=\mathrm{Ad}_{\lambda(\phi)^{-1}}\braket{X|\mathcal{A}_{S}(t)}\nonumber\\
&=\mathrm{Ad}_{\lambda(\phi)^{-1}}\braket{X|\Lambda(t)}
\label{eq:6.17}
\end{align}
$\forall X\in\mathrm{Lie}(\mathcal{S})$. Conversely, suppose one has given a smooth map $\Lambda\in H^{\infty}(\mathcal{T},\mathrm{Hom}_L(\mathrm{Lie}(\mathcal{S}),\allowbreak\mathrm{Lie}(\mathcal{G})))_0$ satisfying (\ref{eq:6.10}) and (\ref{eq:6.11}) above. We have to show that there indeed exists a unique super connection 1-form $\mathcal{A}\in\Omega^1(\mathcal{P}_{/\mathcal{T}},\mathfrak{g})_0$ such that $\iota^*\mathcal{A}(t)=\theta^{(\mathcal{S})}_{\mathrm{MC}}\diamond\Lambda(t)$ for any $t\in\mathcal{T}$. This follows along the lines of the proof of prop. \ref{prop:4.6}. As there, one can show that, if $\mathcal{A}$ exists, it necessarily has to be of the form 
\begin{equation}
\braket{D_{(p,g)}\hat{\pi}(X_p,Y_g)|\mathcal{A}_{[p,g]}}=\mathrm{Ad}_{g^{-1}}\braket{X_p|\theta^{(\mathcal{S})}_{\mathrm{MC}}\diamond\Lambda(t)}+\braket{Y_g|\theta^{(\mathcal{G})}_{\mathrm{MC}}}
\label{eq:6.21}
\end{equation}
Moreover, as $\hat{\pi}$ is a submersion, it is uniquely determined by (\ref{eq:6.21}). Similar as in the proof of prop. \ref{prop:4.6}, one concludes that this indeed provides a well-defined super connection 1-form on $\mathcal{P}_{/\mathcal{T}}$.
\end{proof}
\begin{remark}
Note that if $\lambda:\,\mathcal{S}_0\rightarrow\mathcal{G}$ is a group morphism corresponding to a left action of a bosonic super Lie group $\mathcal{S}_0$, i.e., a split super Lie group corresponding to an ordinary symmetry group, then $\lambda$ only takes values in the bosonic super Lie subgroup $\mathcal{G}_0$ of $\mathcal{G}$.\\
This can be cured generalizing definition (\ref{eq:6.1}) to the relative category including parametrizing supermanifolds $\mathcal{T}$ with nontrivial odd dimensions. A classification of these type of actions is then given by maps of the form $\lambda':\,\mathcal{T}\times\mathcal{S}_x\rightarrow\mathcal{G}$ satisfying $\lambda'(t,\phi\circ\psi)=\lambda'(t,\phi)\circ\lambda'(t,\psi)$. In this way, one may be able to construct symmetry reduced connections which are invariant under ordinary spatial symmetries up to local gauge \emph{and} supersymmetry transformations.
\end{remark}

\section{Conclusions}
In this paper, we have studied the Cartan geometric approach to supergravity as well as its application to loop quantum supergravity. To this end, we have provided a mathematical rigorous foundation for the formulation of super Cartan geometries. A crucial ingredient for supersymmetry is the anticommutative nature of fermionic fields. However, as we have seen, modeling anticommuting classical fermion fields turns out to be by far non-straightforward. A resolution is given considering enriched categories as initiated already in \cite{Schmitt:1996hp} based on standard techniques in algebraic geometry. This procedure requires the choice of an additional \emph{parametrizing} supermanifold which encodes the fermionic degrees of freedom. Since the choice is arbitrary, one needs to ensure that physical quantities behave functorially under a change of parametrization. As we have seen, this property follows naturally, if one works in the category of relative supermanifolds. This also reflects the interpretation of supermanifolds in the sense of Molotkov-Sachse \cite{Mol:10,Sac:08} in terms of a functor $\mathbf{Gr}\rightarrow\mathbf{Top}$ assigning Grassmann-algebras to Rogers-DeWitt supermanifolds.

Having formulated the notion of super Cartan geometries in the framework of enriched categories, we have then derived $\mathcal{N}=1$, $D=4$ supergravity via the super MacDowell-Mansouri action. Moreover it follows that, in this geometric framework, SUSY transformations have the interpretation in terms of (field dependent) super gauge transformations.

These results were then applied to the manifest approach towards loop quantum supergravity. To this end, studying the chiral structure of the underlying supersymmetry algebra (corresponding to the super Klein geometry) we have derived the graded analoga of Ashtekar's self-dual variables. Moreover, we were able to interpret these variables in terms of generalized Cartan connections giving rise to super principal connections to the associated principal bundle. In this way, this provides a link between $\mathcal{N}=1$, $D=4$ supergravity and Yang-Mills theory with gauge group given the super anti-de Sitter or super Poincaré group in $D=2$ (where the latter can be obtained via a Inönü-Wigner contraction of the former) in case with or without a cosmological constant, respectively. This is in fact in complete analogy to the classical theory where the ordinary self-dual variables give ordinary gravity the structure of a $\mathrm{SL}(2,\mathbb{C})$ Yang-Mills theory. Moreover, as it turns out, the possibility of defining the graded Ashtekar connection is based crucially on the particular properties of the (bosonic) self-dual variables making them in a sense unique. Furthermore, it follows that the chiral structure of the supersymmetry algebra is even preserved in case of extended supersymmetry. 

This shows that the existence of the graded Ashtekar connection in the $\mathcal{N}=1$ case is not just mere coincidence and might also appear in context of extended SUGRA theories. In fact, chiral $\mathcal{N}=2$, $D=4$ AdS supergravity has been studied in \cite{Tsuda:2000er,Sano:1992jw} and in context of constrained super BF-theory in \cite{Ezawa:1995nj,Ling:2003yw} using the graded Ashtekar connection associated to the superalgebra $\mathrm{OSp}(2|2)_{\mathbb{C}}$. More recently, these variables have been derived in \cite{Eder:2021nyb} for pure $\mathcal{N}=2$, $D=4$ AdS supergravity in the presence of boundaries in a purely geometric way following the Cartan geometric approach as studied in this present article. Thus, based on these observations, these graded variables seem to be right starting point for quantizing supergravity in the framework of loop quantum gravity. Moreover, this Cartan geometric approach may open the possibility for a systematic study of these type of variables including extended supergravity theories with $\mathcal{N}>2$ as well as matter coupled or even higher dimensional SUGRA theories. Among other things, this may also lead to a very natural quantization of higher gauge fields in the framework of LQG. For an interesting treatment of higher gauge fields in a complementary approach that does not keep a part of the supersymmetry manifest but can handle higher dimensional SUGRA theories see \cite{Bodendorfer:2011pc}.

By its very definition, the graded Ashtekar connection encodes both gravity and matter degrees of freedom. As such, it leads to a unified description of both. Since it provides a link between supergravity and Yang-Mills theory, it opens the possibility for a quantization of the theory using standard techniques of LQG following \cite{Gambini:1995db,Ling:1999gn}. Therefore, we have constructed the parallel transport map associated to the super connection 1-form in a mathematical rigorous way. In order to model anticommuting classical fermion fields we again worked in the framework of enriched categories. It follows that the parallel transport map constructed this way indeed has the right properties, i.e., it behaves functorially under a change of parametrization. Moreover, Wilson-loop observables are indeed invariant under (parametrized) super gauge transformations.

 Using the parallel transport map, we then studied the explicit construction of the state space of (manifest) loop quantum supergravity. To this end, we have, in particular, derived the invariant Haar measure of $\mathrm{OSp}(1|2)_{\mathbb{C}}$. The resulting Hilbert space turned out to have a very intriguing structure, giving it the structure of a Krein-space, i.e., an indefinite inner product space that can be completed to a Hilbert space with $\mathrm{SL}(2,\mathbb{C})$-invariant scalar product. However, the quantization of fermions along edges is fundamentally different resulting in 1-dimensional quantum excitations of the fermionic fields similar to gravity. It would be very interesting to see in which sense these different kinds of quantizations can be related.

The construction of the full state space of manifest LQSG, however, still remains incomplete, due to the non-compactness of the underlying gauge group as well as the implementation of the reality conditions. However, in \cite{Konsti-Kos:2020} we made some progress in this direction studying certain symmetry reduced models where the reality conditions can be implemented exactly. We therefore have provided the required mathematical tools in the last section of this present article.

Finally, let us emphasize that this approach also uncovers some of the underlying mathematical structure behind the quantization of fermions in LQG \cite{MoralesTecotl:1994ns,Thiemann:1997rq}. Due to this observation, this approach may also give a concrete idea how matter fields could be quantized in the spin-foam approach to quantum gravity as suggested already in \cite{Livine:2003hn}. The Cartan geometric approach as studied in this article therefore seems to be the right starting point. 



\subsection*{Acknowledgments}
I thank Hanno Sahlmann for helpful discussions at various stages of this work. I also thank Thomas Thiemann and Jerzy Lewandowski for helpful and enlightening discussions. I thank the German Academic Scholarship Foundation (Studienstiftung des Deutschen Volkes) for financial support. 
\newpage
\appendix

\section{Super linear algebra}\label{SuperAlg}
This section is meant to fix some terminology of important aspects in super linear algebra used in the main text. Therefore, we will exclusively focus on $\mathbb{Z}_2$ grading as these are commonly used in physics in the context of (supersymmetric) field theories modeling commuting bosonic and anticommuting fermionic fields. A very deep and detailed account on this fascinating subject can be found, for instance, in the great reference of \cite{Carmeli:2011,Tuynman:2004}.
\begin{definition}\label{Def:2.1}
A \emph{$\mathbb{Z}_2$-graded} or simply \emph{super vector space} $V$ is a vector space over a field $\mathbb{K}$ ($\mathbb{K}=\mathbb{R}$ or $\mathbb{C}$) of the form 
\begin{equation}
V=V_0\oplus V_1
\label{eq:}
\end{equation}
together with a map $|\cdot|:\,\bigcup_{i\in\mathbb{Z}_2}V_i\rightarrow\mathbb{Z}_2$ called \emph{parity map} such that $|v|:=i$ $\forall v\in V_i$. Elements in $V_i$ are called \emph{homogeneous with parity $i\in\mathbb{Z}_2$}. If the dimension of $V_0$ and $V_1$ are given by $\mathrm{dim}\,V_0=m$ and $\mathrm{dim}\,V_1=n$, respectively, then the dimension of $V$ is denoted by $\mathrm{dim}\,V=m|n$. \\
A morphism $\phi:\,V\rightarrow W$ between super vector spaces is a linear map between vector spaces preserving the parity, i.e., $\phi(V_i)\subseteq W_i$ for $i\in\mathbb{Z}_2$. The set of such super vector space morphisms is denoted by $\mathrm{Hom}(V,W)$
\end{definition}

\begin{remark}\label{Remark:2.2}
Instead of just looking at parity preserving morphisms between super vector spaces $V$ and $W$, one can also consider all possible linear maps between the underlying vector spaces. This yields the internal $\underline{\mathrm{Hom}}(V,W)$ which has the structure of a super vector space with the even and odd part $\underline{\mathrm{Hom}}(V,W)_0$ and $\underline{\mathrm{Hom}}(V,W)_1$ given by the parity preserving and parity reversing linear maps between $V$ and $W$, respectively. Hence, $\underline{\mathrm{Hom}}(V,W)_0$ coincides with $\mathrm{Hom}(V,W)$ in definition \ref{Def:2.1}.
\end{remark}

\begin{example}
A trivial but also very important example of a super vector space is given by $\mathbb{R}^{m|n}:=\mathbb{R}^m\oplus\mathbb{R}^n$ called the \emph{$(m,n)$-dimensional superspace}. In fact, any super vector space $V=V_0\oplus V_1$ is isomorphic to a superspace. If $\mathrm{dim}\,V=m|n$, let $\{e_i\}_{i=1,\ldots,m+n}$ be a basis of $V$ such that $\{e_i\}_{i=1,\ldots,m}$ is a basis of $V_0$ and $\{e_j\}_{j=m,\ldots,m+n}$ is a basis of $V_1$. Such a basis is called a \emph{homogeneous basis of} $V$. Then, $V$ is isomorphic, as a super vector space, to the superspace $\mathbb{R}^{m|n}$.
\end{example}

\begin{definition}
A \emph{superalgebra} $A$ is a super vector space $A=A_0\oplus A_1$ together with a bilinear map $m:\,A\times A\rightarrow A$ such that
\begin{equation}
m(A_i,A_j)=:A_i\cdot A_j\subseteq A_{i+j}\quad\forall i,j\in\mathbb{Z}_2
\label{eq:}
\end{equation}
The superalgebra $A$ is called \emph{super commutative} if 
\begin{equation}
a\cdot b=(-1)^{|a||b|}b\cdot a
\label{eq:}
\end{equation}
for all homogeneous $a,b\in A$. 
\end{definition}

\begin{definition}
Let $A$ be a superalgebra. A \emph{super left $A$-module} $\mathcal{V}$ is a super vector space which, in addition, has the structure of a left $A$-module such that
\begin{equation}
A_i\cdot \mathcal{V}_j\subseteq\mathcal{V}_{i+j}\quad\forall i,j\in\mathbb{Z}_2
\label{eq:}
\end{equation}
A morphism $\phi:\,\mathcal{V}\rightarrow\mathcal{W}$ between super left $A$-modules is a map between the underlying super vector spaces such that $\phi(a\cdot v)=a\phi(v)$ $\forall a\in A$, $v\in\mathcal{V}$. Analogously, one defines a \emph{super right $A$-modules} and morphisms between them.
\end{definition}
\begin{remark}
Given super left $A$-module $\mathcal{V}$ one can also turn it into a super right $A$-module setting
\begin{equation}
v\cdot a:=(-1)^{|v||a|}a\cdot v
\label{eq:}
\end{equation}
for homogeneous $a\in A$ and $v\in\mathcal{V}$. For this reason, in the following, we will simply say super $A$-module if we do not want to specify whether it should be regarded as a super left or right $A$-module. If $\mathcal{V}$ and $\mathcal{W}$ are super commutative super $A$-modules, their tensor product $\mathcal{V}\otimes\mathcal{W}$ is defined viewing $\mathcal{V}$ as a left and $\mathcal{W}$ as a right $A$-module.
\end{remark}
Given a super left $A$-modules $\mathcal{V}$ and $\mathcal{W}$ we denote the set of left $A$-module morphisms $\phi:\,\mathcal{V}\rightarrow\mathcal{W}$ by $\mathrm{Hom}_L(\mathcal{V},\mathcal{W})$ (and similarly $\mathrm{Hom}_R(\mathcal{V},\mathcal{W})$ for right $A$-module morphisms). As in remark \ref{Remark:2.2}, instead of just looking at parity preserving morphisms, one can also consider all possible linear maps $\phi:\,\mathcal{V}\rightarrow\mathcal{W}$ between the underlying vector spaces satisfying 
\begin{equation}
\phi(a\cdot v)=a\phi(v),\quad\forall v\in\mathcal{V},\,a\in A
\label{eq:}
\end{equation}
This again yields an internal $\underline{\mathrm{Hom}}_L(\mathcal{V},\mathcal{W})$ which has the structure of a super right $A$-module with $\underline{\mathrm{Hom}}_L(\mathcal{V},\mathcal{W})_0=\mathrm{Hom}(\mathcal{V},\mathcal{W})$.\\
As usual, we denote the evaluation of a morphism $\phi\in\underline{\mathrm{Hom}}_L(\mathcal{V},\mathcal{W})$ at $v\in\mathcal{V}$ via
\begin{equation}
\braket{v|\phi}\in\mathcal{W}
\label{eq:}
\end{equation}
This has the advantage that one does not need to care about signs due to super commutativity after right multiplication with elements $a\in A$, i.e., $\braket{v|\phi a}=\braket{v|\phi}a$ $\forall a\in A$, $v\in\mathcal{V}$.\\
Finally, let us define via $\phi\diamond\psi\in\underline{\mathrm{Hom}}_L(\mathcal{V},\mathcal{W}')$ the composition of two left linear morphisms $\phi:\,\underline{\mathrm{Hom}}_L(\mathcal{V},\mathcal{W})$ and $\psi:\,\underline{\mathrm{Hom}}_L(\mathcal{W},\mathcal{W}')$ given by
\begin{equation}
\braket{v|\phi\diamond\psi}:=\braket{\braket{v|\phi}|\psi},\quad\forall v\in\mathcal{V}\label{eq:}
\end{equation}
which, by definition, is well-behaved under right multiplication.
\begin{definition}
A \emph{super $A$-Lie module} (or \emph{super Lie algebra} or \emph{Lie superalgebra} if $A=\mathbb{K}$) is a super $A$-module $L$ such that bilinear map $m\equiv[\cdot,\cdot]:\,L\times L\rightarrow L$, also called the \emph{bracket}, is graded skew-symmetric, i.e., 
\begin{equation}
[a,b]=-(-1)^{|a||b|}[b,a]
\label{eq:}
\end{equation}
and satisfies the graded Jacobi identity
\begin{equation}
[a,[b,c]]+(-1)^{|a|(|b|+|c|)}[b,[c,a]]+(-1)^{|c|(|a|+|b|)}[c,[a,b]]=0
\label{eq:Jacobi}
\end{equation}
for all homogeneous $a,b,c\in L$.
\end{definition}

\begin{example}
\begin{enumerate}[label=(\roman*)]
 \item If $V$ is a vector space (finite or infinite dimensional), then the exterior power $\bigwedge{V}:=\bigoplus_{k=0}^{\infty}{\bigwedge^{k}V}$ also called \emph{Grassmann algebra} naturally defines a superalgebra with even and odd part given by $(\bigwedge{V})_0:=\bigoplus_{k=0}^{\infty}{\bigwedge^{2k}V}$ and $(\bigwedge{V})_1:=\bigoplus_{k=0}^{\infty}{\bigwedge^{2k+1}V}$, respectively. The Grassmann algebra is super commutative, associative and unital with unit $1\in\mathbb{R}=\bigwedge^{0}V$.
\item For $A$ a superalgebra, the tensor product $A\otimes\mathbb{K}^{m|n}$ defines a super $A$-module with grading $(A\otimes\mathbb{K}^{m|n})_0=A_0\otimes\mathbb{K}^{m}\oplus A_1\otimes\mathbb{K}^{n}$ and $(A\otimes\mathbb{K}^{m|n})_1=A_0\otimes\mathbb{K}^{n}\oplus A_1\otimes\mathbb{K}^{m}$.
\end{enumerate}
\end{example}

\begin{definition}
A super $A$-module $\mathcal{V}$ is called \emph{free} if it contains a homogeneous basis $\{e_i\}_{i=1,\ldots,m+n}$, $\mathrm{dim}\,\mathcal{V}=m|n$, such that any element $v\in\mathcal{V}$ can be written in the form $v=a^ie_i$ with coefficients $a^i\in A$ $\forall i=1,\ldots,m+n$. Equivalently, $\mathcal{V}$ is a free super super $A$-module iff it is isomorphic to $A\otimes\mathbb{K}^{m|n}$.
\end{definition}
\begin{definition}
\begin{enumerate}[label=(\roman*)]
 \item  Let $\mathcal{V}$ be a free super $A$-module. Two homogeneous bases $\{e_i\}_i$ and $\{f_j\}_j$ of $\mathcal{V}$ are called equivalent if they are related to each other by scalar coefficients, i.e., there exists real numbers $\tensor{a}{_i^j}\in\mathbb{K}$ such that $e_i=\tensor{a}{_i^j}f_j$ $\forall i.j$ (for a proof that this indeed defines an equivalence relation see \cite{Tuynman:2004}).
\item A free super $A$-module $\mathcal{V}$ together with a distinguished equivalence class of homogeneous bases of $\mathcal{V}$ is called a \emph{super $A$-vector space}. A morphism $\phi:\,\mathcal{V}\rightarrow\mathcal{W}$ between super $A$-vector spaces is a morphism of super $A$-modules such that $\phi$ maps the equivalence class of bases of $\mathcal{V}$ to the real vector space spanned by the equivalence class of bases of $\mathcal{W}$.
\end{enumerate}
\begin{remark}
If $\mathcal{V}$ is a super $A$-vector space with equivalence class $[\{e_i\}_i]$ of homogeneous bases of $\mathcal{V}$, then $\mathcal{V}\cong A\otimes V$ with $V$ the super vector space spanned by $\{e_i\}_i$ which, in particular, is independent on the choice of a representative of that equivalence class. Hence, the choice of such an equivalence class yields a well-defined super vector space $V$ also called the \emph{body of $\mathcal{V}$}.\\
On the other hand, if $\mathcal{V}$ is a free super $A$-module, one can always choose a homogeneous basis $\{e_i\}_i$ of $\mathcal{V}$ such that $\mathcal{V}$ becomes a super $A$-vector space w.r.t. the equivalence class $[\{e_i\}_i]$. However, such a choice may not be canonical and various different bases exist which are not related by scalar coefficients.
\end{remark}
\end{definition}

\section{Spinor calculus}
We summarize some important formulas concerning gamma matrices in Minkowksi spacetime and their chiral representation as we will need them in section \label{section:2.3} and \label{section:5.2}. The Clifford algebra $\mathrm{Cl}(\mathbb{R}^{1,3},\eta)$ of four-dimensional Minkowski spacetime is generated by gamma matrices $\gamma_I$, $I\in\{0,\ldots,3\}$, satisfying the Clifford algebra relations
\begin{equation}
    \left\{\gamma_{I},\gamma_{J}\right\}=2\eta_{IJ}
		\label{eq:A1.1}
\end{equation}
where the Minkowski $\eta$ is chosen with signature $\eta=\mathrm{diag}(-+++)$. In the chiral theory, we are working in the \emph{chiral} or \emph{Weyl representation} in which the gamma matrices take the form
\begin{equation}
\gamma_{I}=\begin{pmatrix}
0 & \sigma_{I}\\
\bar{\sigma}_{I} & 0
\end{pmatrix}\quad\text{and}\quad\gamma_{*}=\begin{pmatrix}
\mathds{1} & 0\\
0&-\mathds{1}
\end{pmatrix}
\label{eq:A1.2}
\end{equation}
with $\gamma_*:=i\gamma_0\gamma_1\gamma_2\gamma_3$ the highest rank Clifford algebra element also commonly denoted by $\gamma_5$ in four spacetime dimensions. Here,
\begin{equation}
\sigma_{I}=(-\mathds{1},\sigma_i)\quad\text{and}\quad\bar{\sigma}_{I}=(\mathds{1},\sigma_i)
\label{eq:A1.3}
\end{equation}
with $\sigma_i$, $i=1,2,3$ the ordinary Pauli matrices. With respect to spinorial indices, (\ref{eq:A1.3}) are written as $\sigma_I^{AA'}$ and $\bar{\sigma}_{I A'A}$, respectively. These can be used to map the internal indices $I$ of the co-frame $e^I$ to spinorial indices setting
\begin{equation}
e_{\mu}^{AA'}=e^I_{\mu}\sigma_I^{AA'}
\label{eq:A1.4}
\end{equation}
Primed and unprimed Spinor indices are raised and lowered with respect to the complete antsymmetric symbols $\epsilon^{A'B'}$ and $\epsilon^{AB}$, respectively, with the convention
\begin{equation}
\psi_A=\psi^B\epsilon_{BA}\quad\text{and}\quad\psi^A=\epsilon^{AB}\psi_B
\label{eq:A1.5}
\end{equation}
and analogously for primed indices. Due to $\epsilon\sigma_i\epsilon=\sigma_i^T$, one has the useful formula
\begin{equation}
\sigma_{I AA'}=\sigma_I^{BB'}\epsilon_{BA}\epsilon_{B'A'}=-\bar{\sigma}_{I A'A}
\label{eq:A1.6}
\end{equation}
Using (\ref{eq:A1.3}) as well as (\ref{eq:A1.6}), it is easy to see that
\begin{equation}
\sigma^I_{AA'}\sigma_J^{AA'}=-2\delta_J^I
\label{eq:A1.7}
\end{equation}
Finally, in a $3+1$-decomposition $M\cong\mathbb{R}\times\Sigma$ of the four dimensional spacetime $M$, one considers the spinor-valued one-forms $e_a^{AA'}$ which are related to the spatial metric $q$ on $\Sigma$ according to
\begin{equation}
2q_{ab}=-e_{a AA'}e_b^{AA'}
\label{eq:}
\end{equation}
with $a=1,2,3$. These, together with the future-directed unit normal vector field $n^{AA'}$ which is normal to the time slices $\Sigma_t$ and satisfies
\begin{equation}
n_{AA'}e_a^{AA'}=0\quad\text{and}\quad n_{AA'}n^{AA'}=2
\label{eq:A1.8}
\end{equation}
form a basis of spinors with one primed and one unprimed index. Finally, let us mention the important identities
\begin{align}
n_{AA'}n^{AB'}&=\delta_{A'}^{B'}\label{eq:A1.9}\\
n_{AA'}n^{BA'}&=\delta_{A}^{B}\label{eq:A1.10}\\
\sigma_{i AA'}\sigma^{AB'}_j&=-\delta_{ij}\delta_{A'}^{B'}-i\tensor{\epsilon}{_{ij}^k}n_{AA'}\sigma^{AB'}_k\label{eq:A1.11}\\
\sigma_{i AA'}\sigma^{BA'}_j&=-\delta_{ij}\delta_{A}^{B}+i\tensor{\epsilon}{_{ij}^k}n_{AA'}\sigma^{BA'}_k
\label{eq:A1.12}
\end{align}


\begin{thebibliography}{56}
\addcontentsline{toc}{chapter}{Bibliography}

\bibitem{Ashtekar:1986yd}
A.~Ashtekar,
``New Variables for Classical and Quantum Gravity,''
Phys. Rev. Lett. \textbf{57} (1986), 2244-2247
doi:10.1103/PhysRevLett.57.2244

\bibitem{Ashtekar:1991hf}
A.~Ashtekar,
``Lectures on nonperturbative canonical gravity,''
Adv. Ser. Astrophys. Cosmol. \textbf{6} (1991), 1-334
doi:10.1142/1321

\bibitem{Ashtekar:2004eh}
A.~Ashtekar and J.~Lewandowski,
``Background independent quantum gravity: A Status report,''
Class. Quant. Grav. \textbf{21} (2004), R53
doi:10.1088/0264-9381/21/15/R01
[arXiv:gr-qc/0404018 [gr-qc]].


\bibitem{Freedman:1976xh}
D.~Z.~Freedman, P.~van Nieuwenhuizen and S.~Ferrara,
``Progress Toward a Theory of Supergravity,''
Phys. Rev. D \textbf{13} (1976), 3214-3218
doi:10.1103/PhysRevD.13.3214

\bibitem{vanNieuwenhuizen:2004rh}
P.~van Nieuwenhuizen,
``Supergravity as a Yang\textendash{}Mills Theory,''
doi:10.1142/9789812567147\_0018
[arXiv:hep-th/0408137 [hep-th]].

\bibitem{MacDowell:1977jt}
S.~W.~MacDowell and F.~Mansouri,
``Unified Geometric Theory of Gravity and Supergravity,''
Phys. Rev. Lett. \textbf{38} (1977), 739
[erratum: Phys. Rev. Lett. \textbf{38} (1977), 1376]
doi:10.1103/PhysRevLett.38.739


\bibitem{Neeman:1978zvv}
Y.~Ne'eman and T.~Regge,
``Gravity and Supergravity as Gauge Theories on a Group Manifold,''
Phys. Lett. B \textbf{74} (1978), 54-56
doi:10.1016/0370-2693(78)90058-8

\bibitem{DAuria:1982uck}
R.~D'Auria and P.~Fre,
``Geometric Supergravity in d = 11 and Its Hidden Supergroup,''
Nucl. Phys. B \textbf{201} (1982), 101-140
[erratum: Nucl. Phys. B \textbf{206} (1982), 496]
doi:10.1016/0550-3213(82)90281-4

\bibitem{Castellani:1991et}
L.~Castellani, R.~D'Auria and P.~Fre,
``Supergravity and superstrings: A Geometric perspective. Vol. 1: Mathematical foundations,''
World Scientific (1991) 1-603, Singapore. 

\bibitem{Castellani:2013iq}
L.~Castellani,
``OSp(1|4) supergravity and its noncommutative extension,''
Phys. Rev. D \textbf{88} (2013) no.2, 025022
doi:10.1103/PhysRevD.88.025022
[arXiv:1301.1642 [hep-th]].

\bibitem{Castellani:2014goa}
L.~Castellani, R.~Catenacci and P.~A.~Grassi,
``Supergravity Actions with Integral Forms,''
Nucl. Phys. B \textbf{889} (2014), 419-442
doi:10.1016/j.nuclphysb.2014.10.023
[arXiv:1409.0192 [hep-th]].

\bibitem{Cremonini:2019aao}
C.~A.~Cremonini and P.~A.~Grassi,
``Pictures from Super Chern-Simons Theory,''
JHEP \textbf{03} (2020), 043
doi:10.1007/JHEP03(2020)043
[arXiv:1907.07152 [hep-th]].

\bibitem{Catenacci:2018xsv}
R.~Catenacci, P.~A.~Grassi and S.~Noja,
``Superstring Field Theory, Superforms and Supergeometry,''
J. Geom. Phys. \textbf{148} (2020), 103559
doi:10.1016/j.geomphys.2019.103559
[arXiv:1807.09563 [hep-th]].

\bibitem{Cortes:2018lan}
V.~Cort\'es, C.~I.~Lazaroiu and C.~S.~Shahbazi,
``${\mathcal {N}}=1$ Geometric Supergravity and Chiral Triples on Riemann Surfaces,''
Commun. Math. Phys. \textbf{375} (2019) no.1, 429-478
doi:10.1007/s00220-019-03476-7
[arXiv:1810.12353 [hep-th]].

\bibitem{Fiorenza:2013nha}
D.~Fiorenza, H.~Sati and U.~Schreiber,
``Super Lie n-algebra extensions, higher WZW models, and super p-branes with tensor multiplet fields,''
Int. J. Geom. Meth. Mod. Phys. \textbf{12} (2014), 1550018
doi:10.1142/S0219887815500188
[arXiv:1308.5264 [hep-th]].

\bibitem{Sati:2015yda}
H.~Sati and U.~Schreiber,
``Lie n-algebras of BPS charges,''
JHEP \textbf{03} (2017), 087
doi:10.1007/JHEP03(2017)087
[arXiv:1507.08692 [math-ph]].

\bibitem{Jacobson:1987cj}
T.~Jacobson,
``New Variables for Canonical Supergravity,''
Class. Quant. Grav. \textbf{5} (1988), 923
doi:10.1088/0264-9381/5/6/012

\bibitem{Fulop:1993wi}
G.~Fulop,
``About a superAshtekar-Renteln ansatz,''
Class. Quant. Grav. \textbf{11} (1994), 1-10
doi:10.1088/0264-9381/11/1/005
[arXiv:gr-qc/9305001 [gr-qc]].

\bibitem{Gambini:1995db}
R.~Gambini, O.~Obregon and J.~Pullin,
``Towards a loop representation for quantum canonical supergravity,''
Nucl. Phys. B \textbf{460} (1996), 615-631
doi:10.1016/0550-3213(95)00582-X
[arXiv:hep-th/9508036 [hep-th]].

\bibitem{Ling:1999gn}
Y.~Ling and L.~Smolin,
``Supersymmetric spin networks and quantum supergravity,''
Phys. Rev. D \textbf{61} (2000), 044008
doi:10.1103/PhysRevD.61.044008
[arXiv:hep-th/9904016 [hep-th]].

\bibitem{Livine:2003hn}
E.~R.~Livine and R.~Oeckl,
``Three-dimensional Quantum Supergravity and Supersymmetric Spin Foam Models,''
Adv. Theor. Math. Phys. \textbf{7} (2003) no.6, 951-1001
doi:10.4310/ATMP.2003.v7.n6.a2
[arXiv:hep-th/0307251 [hep-th]].

\bibitem{Sawaguchi:2001wi}
M.~Sawaguchi,
``Canonical formalism of N=1 supergravity with the real Ashtekar variables,''
Class. Quant. Grav. \textbf{18} (2001), 1179-1186
doi:10.1088/0264-9381/18/7/303

\bibitem{Tsuda:1999bg}
M.~Tsuda,
``Generalized Lagrangian of N=1 supergravity and its canonical constraints with the real Ashtekar variable,''
Phys. Rev. D \textbf{61} (2000), 024025
doi:10.1103/PhysRevD.61.024025
[arXiv:gr-qc/9906057 [gr-qc]].

\bibitem{Tsuda:2000er}
M.~Tsuda and T.~Shirafuji,
``The Canonical formulation of N=2 supergravity in terms of the Ashtekar variable,''
Phys. Rev. D \textbf{62} (2000), 064020
doi:10.1103/PhysRevD.62.064020
[arXiv:gr-qc/0003010 [gr-qc]].

\bibitem{Ling:2003yw}
Y.~Ling, R.~S.~Tung and H.~Y.~Guo,
``Supergravity and Yang-Mills theories as generalized topological fields with constraints,''
Phys. Rev. D \textbf{70} (2004), 044045
doi:10.1103/PhysRevD.70.044045
[arXiv:hep-th/0310141 [hep-th]].

\bibitem{Sano:1992jw}
T.~Sano,
``The Ashtekar formalism and WKB wave functions of N=1, N=2 supergravities,''
[arXiv:hep-th/9211103 [hep-th]].

\bibitem{Ezawa:1995nj}
K.~Ezawa,
``Ashtekar's formulation for N=1, N=2 supergravities as constrained BF theories,''
Prog. Theor. Phys. \textbf{95} (1996), 863-882
doi:10.1143/PTP.95.863
[arXiv:hep-th/9511047 [hep-th]].


\bibitem{MoralesTecotl:1994ns}
H.~A.~Morales-Tecotl and C.~Rovelli,
``Fermions in quantum gravity,''
Phys. Rev. Lett. \textbf{72} (1994), 3642-3645
doi:10.1103/PhysRevLett.72.3642
[arXiv:gr-qc/9401011 [gr-qc]].

\bibitem{Thiemann:1997rq}
T.~Thiemann,
``Kinematical Hilbert spaces for Fermionic and Higgs quantum field theories,''
Class. Quant. Grav. \textbf{15} (1998), 1487-1512
doi:10.1088/0264-9381/15/6/006
[arXiv:gr-qc/9705021 [gr-qc]].

\bibitem{Thiemann:2007pyv}
T.~Thiemann,
``Modern Canonical Quantum General Relativity,''
doi:10.1017/CBO9780511755682

\bibitem{Bodendorfer:2011pb}
N.~Bodendorfer, T.~Thiemann and A.~Thurn,
``Towards Loop Quantum Supergravity (LQSG) I. Rarita-Schwinger Sector,''
Class. Quant. Grav. \textbf{30} (2013), 045006
doi:10.1088/0264-9381/30/4/045006
[arXiv:1105.3709 [gr-qc]].

\bibitem{Bodendorfer:2011pc}
N.~Bodendorfer, T.~Thiemann and A.~Thurn,
``Towards Loop Quantum Supergravity (LQSG) II. p-Form Sector,''
Class. Quant. Grav. \textbf{30} (2013), 045007
doi:10.1088/0264-9381/30/4/045007
[arXiv:1105.3710 [gr-qc]].

\bibitem{Frodden:2012dq}
E.~Frodden, M.~Geiller, K.~Noui and A.~Perez,
``Black Hole Entropy from complex Ashtekar variables,''
EPL \textbf{107} (2014) no.1, 10005
doi:10.1209/0295-5075/107/10005
[arXiv:1212.4060 [gr-qc]].

\bibitem{Han:2014xna}
M.~Han,
``Black Hole Entropy in Loop Quantum Gravity, Analytic Continuation, and Dual Holography,''
[arXiv:1402.2084 [gr-qc]].

\bibitem{Achour:2014eqa}
J.~Ben Achour, A.~Mouchet and K.~Noui,
``Analytic Continuation of Black Hole Entropy in Loop Quantum Gravity,''
JHEP \textbf{06} (2015), 145
doi:10.1007/JHEP06(2015)145
[arXiv:1406.6021 [gr-qc]].

\bibitem{Achour:2015zmk}
J.~Ben Achour,
``Towards self dual Loop Quantum Gravity,''
[arXiv:1511.07332 [gr-qc]].

\bibitem{Bodendorfer:2013hla}
N.~Bodendorfer and Y.~Neiman,
``Imaginary action, spinfoam asymptotics and the \textquoteleft{}transplanckian\textquoteright{} regime of loop quantum gravity,''
Class. Quant. Grav. \textbf{30} (2013), 195018
doi:10.1088/0264-9381/30/19/195018
[arXiv:1303.4752 [gr-qc]].

\bibitem{Wilson-Ewing:2015lia}
E.~Wilson-Ewing,
``Loop quantum cosmology with self-dual variables,''
Phys. Rev. D \textbf{92} (2015) no.12, 123536
doi:10.1103/PhysRevD.92.123536
[arXiv:1503.07855 [gr-qc]].

\bibitem{Wise:2006sm}
D.~K.~Wise,
``MacDowell-Mansouri gravity and Cartan geometry,''
Class. Quant. Grav. \textbf{27} (2010), 155010
doi:10.1088/0264-9381/27/15/155010
[arXiv:gr-qc/0611154 [gr-qc]].

\bibitem{Deligne:1999qp}
P.~Deligne, P.~Etingof, D.~Freed, L.~Jeffrey, D.~Kazhdan, J.~Morgan, D.~Morrison and E.~Witten,
``Quantum fields and strings: A course for mathematicians. Vol. 1''

\bibitem{Carmeli:2011}
C.~Carmeli, L.~Caston, R.~Fioresi,
``Mathematical Foundations of Supersymmetry,''
EMS Series of Lectures in Mathematics 15, European Mathematical Society, 2011.


\bibitem{Nicolai:1984hb}
H.~Nicolai,
``REPRESENTATIONS OF SUPERSYMMETRY IN ANTI-DE SITTER SPACE,''
CERN-TH-3882.

\bibitem{Freedman:1983na}
D.~Z.~Freedman and H.~Nicolai,
``Multiplet Shortening in Osp($N$,4),''
Nucl. Phys. B \textbf{237} (1984), 342-366
doi:10.1016/0550-3213(84)90164-0

\bibitem{Freedman:2012zz}
D.~Z.~Freedman and A.~Van Proeyen,
``Supergravity,''
Cambridge Univ. Press, 2012.

\bibitem{Wipf:2016}
A.~Wipf,
``Introduction to Supersymmetry,''
Vorlesungsskript, Universit\"at Jena, 2016. 


\bibitem{Leites:1980}
D.~A.~Leites,
``Introduction to the theory of supermanifolds,''
Russian Math. Surveys 35(1980), 1–64.

\bibitem{Rothstein:1987}
M.~J.~Rothstein, 
``Integration on noncompact supermanifolds,''
Trans. A.M.S. 299(1987), 387–396.


\bibitem{Varadarajan:2004}
V.~S.~Varadarajan, 
``Supersymmetry for mathematicians: An introduction,''
Courant Lecture Notes 11, 
AMS, Providence, RI, 2004.

\bibitem{Coulembier:2012}
K.~Coulembier and R.~B.~Zhang. 
``Invariant integration on orthosymplectic and unitary supergroups,'' Journal of Physics A: Mathematical and Theoretical 45.9 (2012): 095204.

\bibitem{Greenstein:2012}
J.~Greenstein and V.~Mazorchuk, 
``Koszul duality for semidirect products and generalized Takiff algebras.,'' Algebras and Representation Theory 20.3 (2017): 675-694.

\bibitem{Konsti-SPW:2020}
K.~Eder,
``Towards super Peter-Weyl theory,''
in preparation.


\bibitem{Konsti-FB:2020}
K.~Eder,
``Super fiber bundles, connection forms and parallel transport,''
[arXiv:2101.00924 [math.DG]].

\bibitem{Konsti-Kos:2020}
K.~Eder and H.~Sahlmann,
``Supersymmetric minisuperspace models in self-dual loop quantum cosmology,''
JHEP \textbf{21} (2020), 064
doi:10.1007/JHEP03(2021)064
[arXiv:2010.15629 [gr-qc]].

\bibitem{Konsti-SUGRA:2020}
K.~Eder and H.~Sahlmann,
``$\mathcal{N}=1$ Supergravity with loop quantum gravity methods and quantization of the SUSY constraint,''
Phys. Rev. D \textbf{103} (2021) no.4, 046010
doi:10.1103/PhysRevD.103.046010
[arXiv:2011.00108 [gr-qc]].

\bibitem{Eder:2021nyb}
K.~Eder and H.~Sahlmann,
``Holst-MacDowell-Mansouri action for (extended) supergravity with boundaries and super Chern-Simons theory,''
[arXiv:2104.02011 [gr-qc]].

\bibitem{Berezin:1976}
F.~Berezin and D.~Leites
``Supermanifolds,''
Soviet Maths Doknaja \textbf{16}, pp. 1218-1222.

\bibitem{Kostant:1975qe}
B.~Kostant,
``Graded Manifolds, Graded Lie Theory, and Prequantization,''
Lect. Notes Math. \textbf{570} (1977), 177-306
doi:10.1007/BFb0087788


\bibitem{Batchelor:1979}
M.~Batchelor, 
``The structure of supermanifolds,''
Transactions of the American Mathematical Society 253 (1979): 329-338.


\bibitem{Batchelor:1980}
M.~Batchelor, 
``Two approaches to supermanifolds,''
Transactions of the American Mathematical Society 258.1 (1980): 257-270.

\bibitem{DeWitt:1984}
B.~DeWitt,
``Supermanifolds,''
Cambridge University Press, Cambridge 1984.

\bibitem{Rogers:1980}
A.~Rogers,
``A global theory of supermanifolds,''
Journal of Mathematical Physics 21.6 (1980): 1352-1365.

\bibitem{Rogers:2007}
A.~Rogers,
``Supermanifolds: Theory and Applications,''
World Scientific (2007).

\bibitem{Carmeli:2011}
C.~Carmeli, L.~Caston, and R.~Fioresi. 
``Mathematical foundations of supersymmetry,''
Vol. 15. European Mathematical Society, 2011.

\bibitem{Tuynman:2004}
G.~M.~Tuynman,
``Supermanifolds and Supergroups-Basic Theory,''
Kluwer Academic Publishers (2004),
doi:10.1007/1-4020-2297-2

\bibitem{Tuynman:2018}
G.~M.~Tuynman,
``Super unitary representations revisited,''
arXiv:1711.00233v2 [math.DG]

\bibitem{Cheng:2012}
S.~J.~Cheng and W.~Wang, 
``Dualities and Representations of Lie Superalgebras,''
Graduate Studies in Mathematics 144, Amer. Math. Soc., Providence, RI, 2012.


\bibitem{Jost:2014wfa}
J.~Jost, E.~Ke\ss{}ler and J.~Tolksdorf,
``Super Riemann surfaces, metrics and gravitinos,''
Adv. Theor. Math. Phys. \textbf{21} (2017), 1161-1187
doi:10.4310/ATMP.2017.v21.n5.a2
[arXiv:1412.5146 [math-ph]].

\bibitem{Kessler:2019bwp}
E.~Ke\ss{}ler,
``Supergeometry, Super Riemann Surfaces and the Superconformal Action Functional,''
Lect. Notes Math. \textbf{2230} (2019), pp.
doi:10.1007/978-3-030-13758-8


\bibitem{Hack:2015vna}
T.~P.~Hack, F.~Hanisch and A.~Schenkel,
``Supergeometry in locally covariant quantum field theory,''
Commun. Math. Phys. \textbf{342} (2016) no.2, 615-673
doi:10.1007/s00220-015-2516-4
[arXiv:1501.01520 [math-ph]].
\bibitem{Mol:10}
V.~Molotkov,
``infinite-dimensional and colored supermanifolds,''
Journal of Nonlinear Mathematical Physics, 17:sup1, 375-446 (2010), doi: 10.1142/S140292511000088X

\bibitem{Sac:08}
C.~Sachse,
``A categorical formulation of superalgebra and supergeometry,''
arXiv:0802.4067v1 [math.AG].

\bibitem{Sac:09}
C.~Sachse,
``Global Analytic Approach to Super Teichmueller Spaces,''
arXiv:0902.3289v1 [math.AG].


\bibitem{Rejzner:2011au}
K.~Rejzner,
``Fermionic fields in the functional approach to classical field theory,''
Rev. Math. Phys. \textbf{23} (2011), 1009-1033
doi:10.1142/S0129055X11004503
[arXiv:1101.5126 [math-ph]].

\bibitem{Rejzner:2016hdj}
K.~Rejzner,
``Perturbative Algebraic Quantum Field Theory: An Introduction for Mathematicians,''
Math. Phys. Stud. (2016)
doi:10.1007/978-3-319-25901-7

\bibitem{Schuett:2018}
J.~Sch\"utt,
``Infinite-Dimensional Supermanifolds via Multilinear Bundles,''
arXiv preprint arXiv:1810.05549 (2018).

\bibitem{Schmitt:1996hp}
T.~Schmitt,
``Supergeometry and quantum field theory, or: What is a classical configuration?,''
Rev. Math. Phys. \textbf{9} (1997), 993-1052
doi:10.1142/S0129055X97000348
[arXiv:hep-th/9607132 [hep-th]].

\bibitem{Florin:2008}
F.~Dumitrescu,
``Superconnections and parallel transport,''
Pacific Journal of Mathematics 236.2 (2008): 307-332.

\bibitem{Groeger:2013aja}
J.~Groeger,
``Super Wilson Loops and Holonomy on Supermanifolds,''
[arXiv:1312.4745 [math-ph]].



\bibitem{Mason:2010yk}
L.~Mason and D.~Skinner,
``The Complete Planar S-matrix of N=4 SYM as a Wilson Loop in Twistor Space,''
JHEP \textbf{12} (2010), 018
doi:10.1007/JHEP12(2010)018
[arXiv:1009.2225 [hep-th]].

\end{thebibliography}
\end{document}